\documentclass{article}
\usepackage{kr}
\usepackage{times}
\usepackage{soul}
\usepackage{url}
\usepackage[hidelinks]{hyperref}
\usepackage[utf8]{inputenc}
\usepackage[small]{caption}
\usepackage{graphicx}
\usepackage{amsmath}
\usepackage{booktabs}
\usepackage{algorithm2e}
\usepackage{makecell}
\usepackage{xcolor,colortbl}
\definecolor{Gray}{gray}{0.85}

\usepackage{bbm}
\usepackage{misc}
\usepackage{pifont}
\usepackage{xspace}
\usepackage{url}
\usepackage{boxedminipage}
\usepackage{amssymb}
\usepackage{amsmath}
\usepackage{mathtools}

\usepackage{thmtools,thm-restate}

\usepackage{mathtools}

\usepackage{stmaryrd} 

\usepackage{tikz}
\usetikzlibrary{decorations.pathreplacing}
\usetikzlibrary{matrix}
\usetikzlibrary{arrows}
\usetikzlibrary{positioning}

\newtheorem{lemma}{Lemma}
 
\newtheorem{theorem}{Theorem} 
\newtheorem{example}{Example}
\newtheorem{definition}{Definition}

\usepackage{amsfonts}
\usepackage{amssymb}
\usepackage{amsmath}
\usepackage{mathtools}
\usepackage{amstext}
\usepackage{array}
\usepackage{pifont}

\usepackage{enumitem}

\usepackage{misc}
\usepackage{xspace}
\usepackage{pgf}
\usepackage{tikz}
\usetikzlibrary{fit,calc}

\newcommand{\DLS}{{\sf DL}_{{\sf ni}}}

\title{Separating Data Examples by \\[1mm]
 Description
  Logic Concepts
  with
  Restricted Signatures}

\author{Jean Christoph Jung$^{1}$
	\and
	Carsten Lutz$^{2}$ \and
	Hadrien Pulcini$^3$
	\and
	Frank Wolter$^3$
	\affiliations
$^1$Institute of Computer Science, University of Hildesheim, Germany\\
$^2$Department of Computer Science, University of Bremen, Germany\\
$^3$Department of Computer Science, University of Liverpool, UK\\
\emails
jungj@uni-hildesheim.de, clu@uni-bremen.de,
\{H.Pulcini,wolter\}@liverpool.ac.uk
}
\begin{document}

\date{}

\maketitle

\begin{abstract}
We study the separation of positive and negative data examples in
terms of description logic concepts in the presence of an ontology. In contrast to previous work, we add a signature that specifies a subset of the symbols that can be used for separation, and we admit
individual names in that signature. We consider
weak and strong versions of the resulting problem that differ in how
the negative examples are treated and we distinguish between separation with and without
helper symbols. Within this framework, we compare the separating power
of different languages and investigate the complexity of deciding
separability. While weak separability is shown to be closely related
to conservative extensions, strongly separating concepts coincide with
Craig interpolants, for suitably defined encodings of the data and
ontology. This enables us to transfer known results from those fields
to separability.
Conversely, we obtain original results on separability that can be transferred backward.
For example, rather surprisingly, conservative extensions and weak
separability in $\mathcal{ALCO}$ are both \ThreeExpTime-complete.
\end{abstract}

\section{Introduction}

\newcommand\blfootnote[1]{%
	\begingroup
	\renewcommand\thefootnote{}\footnote{#1}%
	\addtocounter{footnote}{-1}%
	\endgroup
}

There are several applications that fall under the broad term of
supervised learning and seek to compute a logical expression that
separates positive from negative examples given in the form of labeled
data items in a knowledge base (KB). A prominent example is concept
learning for description logics (DLs) where inductive logic programming methods 
are applied to construct separating concepts that can then
be used, for instance, in ontology
engineering~\cite{DBLP:journals/ml/LehmannH10}.
Another example is reverse engineering of database queries (or query by example, QBE)~\cite{martins2019reverse}
which has also been studied in the presence of a DL ontology
\cite{GuJuSa-IJCAI18,DBLP:conf/gcai/Ortiz19}. A closed world
semantics is adopted for QBE in databases while an open
world semantics is required 
in the presence of ontologies; the latter is the case also in reverse
engineering of SPARQL queries~\cite{DBLP:conf/www/ArenasDK16}. Further
applications are entity
comparison in RDF graphs, where one aims to find meaningful
descriptions that separate one entity from
another~\cite{DBLP:conf/semweb/PetrovaSGH17,DBLP:conf/semweb/PetrovaKGH19}
and generating referring expressions (GRE) where the aim
is to describe a single data item by a logical expression such as a DL
concept, separating it from all other data items~\cite{DBLP:journals/coling/KrahmerD12,DBLP:conf/kr/BorgidaTW16}.

A fundamental problem common to all these applications is to decide
whether a separating formula exists at all. There are several degrees
of freedom in defining this problem. The first concerns the negative
examples: is it enough that they do not entail the separating formula
(\emph{weak separability}) or are they required to entail its negation
(\emph{strong separability})?  Another one concerns the question
whether additional helper symbols are admitted in the separating
formula (\emph{projective separability}) or not (\emph{non-projective
separability}). The emerging family of problems has recently been
investigated in \cite{DBLP:conf/ijcai/FunkJLPW19,KR}, concentrating on
the case where the separating expression is a DL concept or a formula
from a fragment of first-order logic (FO) such as the guarded fragment
(GF) and unions of conjunctive queries (UCQs).
 
In this paper, we add a signature $\Sigma$ (set of concept, role, and
individual names) that is given as an additional input and require
separating expressions to be formulated in~$\Sigma$.  This makes it
possible to `direct' separation towards expressions based on desired
features and accordingly to exclude features that are not supposed to
be used for separation. For example, consider an online book store
where a user has labeled some books with \emph{likes} (positive
examples) or \emph{dislikes} (negative examples). A ``good''
separating expression might include relevant features of books like genre
or language, but exclude information about the author's age or gender. 

The aim of this paper is to investigate the effect of adding a signature
to the framework, and in particular to compare the
separating power of different languages and determine the
computational complexity of deciding separability. 
We focus on the case in which both the knowledge base and the
separating expressions are formulated in DLs between $\mathcal{ALC}$
and its extension $\mathcal{ALCIO}$ with inverse roles and
nominals. DLs with nominals are of particular interest to
  us as separating expressions formulated in such DLs may refer to
  individual names in the signature $\Sigma$. Returning to the book
  store example, one can use the standard DL representation of
  specific authors (`Hemingway') and languages (`English') as
  individuals in separating expressions. To understand the robustness of
our results, we also discuss in how far they extend to the guarded
fragment (GF) and the two-variable fragment (FO$^{2}$) of FO.

We start with weak projective separability. We first observe that
helper symbols, which must be `fresh' in that they do not occur in the
given knowledge base, increase the ability to separate and lead to
more succinct separating expressions.
We concentrate on the case where helper symbols are
concept names because admitting individual names leads to
undecidability of the separability problem while admitting roles names
either does not make a difference (in $\mathcal{ALC}$ and
$\mathcal{ALCI}$) or makes a difference but is polynomial time
reducible to separation without role names as helper symbols
(in $\mathcal{ALCO}$ and $\mathcal{ALCIO}$).
To investigate further the relationship between
non-projective and projective weak separability, we then introduce
the extension of UCQs in which compound DL concepts are admitted in 
atoms
and show that in some important cases, non-projective weak separability in
that language coincides with projective weak separability in
the original description logic. In thise sense, helper concept names
are thus `captured' by UCQs.

We next investigate the complexity of projective weak separability
with signature for the DLs above. A fundamental observation is that,
due to the presence of the signature, the problem to decide
\emph{projective conservative extensions} at the ontology level is
polynomial time reducible to the complement of projective weak
separability. Here, `projective' refers to the fact that
conservativity is also required for expressions using fresh concept
names. Conservative extensions have been studied in detail in the
context of modular ontologies~\cite{DBLP:journals/jair/GrauHKS08,DBLP:conf/rweb/BotoevaKLRWZ16}. The
projective version is motivated by the requirement of \emph{robustness
under vocabulary
extensions} in applications with frequent changes in the
ontology~\cite{KonevLWW09}. It coincides with the non-projective
one for DLs with the Craig Interpolation property (CIP) such as $\ALC$
and $\ALCI$~\cite{JLMSW17}, but not for DLs with nominals, such as
$\mathcal{ALCO}$. 


The reduction from conservative extensions yields a \TwoExpTime lower
bound for weak projective separability in $\ALC$ and $\ALCI$
\cite{DBLP:conf/kr/GhilardiLW06,DBLP:conf/ijcai/LutzWW07}. We prove a
matching upper bound by providing a bisimulation-based
characterization of weak projective separability and then deciding the
characterization by a reduction to the emptiness problem of suitable
tree automata.


For \ALCO, we show the unexpected result that both projective
conservative extensions and projective weak separability are
\ThreeExpTime-complete. The lower bound is a substantial extension of
the \TwoExpTime-lower bound for conservative extensions in \ALC from
\cite{DBLP:conf/kr/GhilardiLW06}, and it holds for non-projective conservative
extensions and non-projective weak separability as well. The upper bound is again by an encoding into tree
automata.

We then turn to strong separability where we observe that the
projective and non-projective case coincide. We further observe that
separating expressions are identical to \emph{Craig interpolants}
between formulas that encode the KBs with the positive and negative
examples, respectively. Since FO enjoys the CIP, the existence of FO
separating formulas is equivalent to the entailment between the
encoding formulas. This entailment question is \ExpTime-complete
if the KBs are given in a DL between \ALC and
\ALCIO.
Moreover, any FO-theorem prover that computes interpolants can be used
to compute separating expressions~\cite{DBLP:conf/aplas/HoderHKV12}.

Interestingly, while DL concepts alone have a strictly weaker
separating power than FO, a version of the aforementioned extension of UCQs with DL
concepts is expressive enough to capture the separating power of FO.
Regarding the decision problem, we use recent results on the
complexity of Craig interpolant existence~\cite{AJMOW-AAAI21} to show
that for any DL between $\mathcal{ALC}$ and $\mathcal{ALCIO}$, strong
separability is \TwoExpTime-complete if one separates using
concepts from the same DL.
  
We finally consider weak and strong inseparability in the case where
both the ontology and the separating formulas are in GF or
FO$^{2}$. For GF, we prove that weak (projective) separability with
signature is undecidable which is in contrast to decidability of weak
separability when no signature restriction can be imposed on the
separating formula~\cite{KR}. For FO$^2$, already weak separability
without signature is undecidable~\cite{KR}. In the case of strong
separability, the link between Craig interpolants and strongly
separating formulas generalizes to both GF and FO$^{2}$. Both
languages fail to have the
CIP~\cite{DBLP:journals/sLogica/HooglandM02,comer1969,Pigozzi71}, but
recent results on the existence of Craig interpolants can be used to
prove that strong separability in GF is 3\ExpTime-complete and in
FO$^{2}$ is in {\sc N}\TwoExpTime and
\TwoExpTime-hard~\cite{jung2020living}.

\section{Preliminaries}
\label{sec:prelim}
Let \NC, \NR, and \NI be countably infinite sets of \emph{concept},
\emph{role}, and \emph{individual names}. A \emph{role} is a role name
$r$ or an \emph{inverse
	role} $r^-$, with $r$ a role name and $(r^-)^- = r$. A \emph{nominal} takes the
form $\{c\}$ with $c\in \NI$. An
\emph{$\mathcal{ALCIO}$-concept} is defined according to the syntax rule
$$C, D ::= \top \mid \bot \mid A \mid \{c\} \mid \neg C \mid C \sqcap D \mid \exists R.C $$
where $A$ ranges over concept names, $c$ over individual names, and $R$ over roles. We use
$C \sqcup D$ as abbreviation for $\neg (\neg C \sqcap \neg D)$,
$\forall R.C$ for $\neg \exists R.(\neg C)$, and $C \rightarrow D$ for $\neg C \sqcup D$. An \emph{$\ALCI$-concept} is
an $\mathcal{ALCIO}$-concept without nominals, an \emph{$\mathcal{ALCO}$-concept} 
an $\mathcal{ALCIO}$-concept without inverse roles, and an
\emph{$\ALC$-concept} is an $\mathcal{ALCO}$-concept without nominals.
Let $\DLS$ denote the set of languages just introduced, where
$\mn{ni}$
stands
for nominals and inverses. For $\Lmc \in \DLS$, an \emph{\Lmc-ontology} is a finite set of \emph{concept inclusion
	(CIs)} $C \sqsubseteq D$ with $C$ and~$D$ \Lmc-concepts.

A \emph{database} $\Dmc$ is a finite set of \emph{facts} of the form $A(a)$ or
$r(a,b)$ where $A \in \NC$, $r \in \NR$, and $a,b\in \NI$. 
An \emph{$\Lmc$-knowledge base ($\Lmc$-KB)} takes the form $\Kmc=(\Omc,\Dmc)$, where $\Omc$ is an $\Lmc$-ontology and $\Dmc$ a database. 
We assume w.l.o.g.~that any nominal used in $\Omc$ also occurs in $\Dmc$.

A \emph{signature} $\Sigma$ is a set of concept, role, and individual names,
uniformly referred to as \emph{symbols}. $\Sigma$ is called \emph{relational} if it does
not contain individual names. We use $\text{sig}(X)$ to
denote the set of symbols used in any syntactic object $X$ such as a
concept or an ontology. For a database $\Dmc$ we denote by $\text{ind}(\Dmc)$
the set of individual names in $\Dmc$.

Description logics are interpreted in \emph{structures} 
$$
\Amf = (\text{dom}(\Amf),(A^{\Amf})_{A\in \NC}, (r^{\Amf})_{r\in \NR}, (c^{\Amf})_{c \in \NI})
$$
with $A^{\Amf} \subseteq \text{dom}(\Amf)$,
$r^{\Amf} \subseteq \text{dom}(\Amf)^{2}$, and $c^{\Amf} \in
\text{dom}(\Amf)$.
The
\emph{extension} $C^\Amf$ of $\mathcal{ALCIO}$-concepts $C$ is then defined as usual
\cite{DL-Textbook}.
For $D \subseteq \text{dom}(\Amf)$, we use $\Amf|_D$ to denote 
the restriction of $\Amf$ to $D$. 
A \emph{pointed structure}  
takes the form $\Amf,a$ with \Amf a structure and $a\in \text{dom}(\Amf)$.   
A structure \Amf \emph{satisfies} CI
$C \sqsubseteq D$ if $C^\Amf \subseteq D^\Amf$, fact $A(a)$ if
$a^{\Amf} \in A^\Amf$, and fact $r(a,b)$ if $(a^{\Amf},b^{\Amf}) \in r^\Amf$. 
\Amf is a \emph{model} of an ontology, database, or KB if it
satisfies all CIs and facts in it. A KB is
\emph{satisfiable} if it has a model, and a concept $C$ is
\emph{satisfiable w.r.t.\ a KB \Kmc} if \Kmc has a model \Amf with
$C^\Amf\neq\emptyset$.

We use standard notation for first-order logic (FO), and consider
formulas constructed using concept names as unary relation symbols,
role names as binary relation symbols, and individual names as
constants. Equality is admitted.  It is well-known
that every DL concept $C$ is equivalent to an FO-formula
$\varphi_C(x)$ with a single free variable $x$. For a KB \Kmc, an
FO-formula $\varphi(x)$ with a single free variable $x$, and a
constant $a$, we write $\Kmc\models \varphi(a)$ if $\Amf\models
\varphi(a)$ in all models \Amf of~\Kmc.
 
We associate with every structure $\Amf$ a directed graph
$G_{\Amf}^{d}=(\text{dom}(\Amf), \bigcup_{r\in \NR}r^{\Amf})$. Let
$G_{\Amf}^{u}=(\text{dom}(\Amf),E')$ be the undirected version of $G_{\Amf}^{d}$
obtained by forgetting edge directions.  We can thus apply graph
theoretic terminology to structures.  The directed graph
$G_{\Amf}^{d}$ is relevant for the DLs $\ALC$ and $\mathcal{ALCO}$
that do not support inverse roles while the undirected graph
$G_{\Amf}^{u}$ is relevant for $\ALCI$ and $\mathcal{ALCIO}$. To
simplify notation, we often prefix a property of structures with the
language for which it is relevant. For example, if
$\Lmc\in \{\ALC,\mathcal{ALCO}\}$ then we say that $\Amf$ has
\emph{finite $\Lmc$-outdegree} if $G_{\Amf}^{d}$ has and we call
$\Amf$ \emph{$\Lmc$-rooted in $a$} if every node in $\Amf$ is
reachable from $a$ in $G_{\Amf}^{d}$. For
$\Lmc\in \{\ALCI,\mathcal{ALCIO}\}$, the two notions are defined
in the same way, but based on $G_{\Amf}^{u}$ in place of 
$G_{\Amf}^{d}$.
For $\Lmc\in \{\ALCI,\mathcal{ALCIO}\}$, \Amf is an $\Lmc$-\emph{tree}
if $G_{\Amf}^{u}$ is acyclic (also excluding self loops) and there are
no multi-edges in the sense that 
$R_1^\Amf$ and $R_2^\Amf$ are disjoint for all distinct roles $R_1,R_2$. For
$\Lmc\in \{\ALC,\mathcal{ALCO}\}$, \Amf is an $\Lmc$-\emph{tree} if,
in addition, every node in $G_{\Amf}^{d}$ has at most one incoming
edge.

Let $\Lmc\in \{\ALC,\ALCI\}$. A model $\Amf$ of an $\Lmc$-KB
$\Kmc=(\Omc,\Dmc)$ is an \emph{$\Lmc$-forest model} of $\Kmc$ if
$\Amf$ with all $r(a^{\Amf},b^{\Amf})$, $a,b\in \text{ind}(\Dmc)$, removed is the disjoint union of $\Lmc$-trees rooted at $a^{\Amf}$, $a\in \text{ind}(\Dmc)$.
$\Amf$ is an \emph{$\mathcal{ALCO}$-forest model} of an
$\mathcal{ALCO}$-KB $\Kmc=(\Omc,\Dmc)$ if $\Amf$ with all
$r(a,b^{\Amf})$, $a\in \text{dom}(\Amf)$, $b\in \text{ind}(\Dmc)$,
removed is the disjoint union of $\mathcal{ALCO}$-trees rooted at
$a^{\Amf}$, $a\in \text{ind}(\Dmc)$. The following completeness result
is well-known~\cite{DL-Textbook}.
%
\begin{lemma}\label{lem:forestm}
Let $\Lmc\in \{\ALC,\ALCI,\mathcal{ALCO}\}$, $\Kmc$ an $\Lmc$-KB, and $C$ an $\Lmc$-concept. If $\Kmc\not\models C(a)$, then there exists an $\Lmc$-forest model 
$\Amf$ for $\Kmc$ of finite $\Lmc$-outdegree with $a\not\in C^{\Amf}$.
\end{lemma}
Note that Lemma~\ref{lem:forestm} does not hold for $\mathcal{ALCIO}$, a counterexample is given in the appendix.

Besides DL-concepts, we use FO-formulas
with a single free variable as separating formulas. Of particular
importance are the following FO-fragments which combine the
expressive power of (unions of) conjunctive queries with DLs.
Let $\Lmc\in \DLS$. Then CQ$^\Lmc$ denotes the
language of all
FO-formulas $\varphi(x)=\exists \vec y\, \psi$ where $\psi$ is a
conjunction of atoms $C(t)$, $C$ an \Lmc-concept, or $r(t_1,t_2)$ with
$t,t_1,t_2$ variables or constants, and $x$ is the single free
variable of $\varphi(x)$. UCQ$^\Lmc$ contains all formulas
$\varphi(x)=\varphi_1(x)\vee\cdots\vee\varphi_n(x)$ with
$\varphi_i(x)\in \text{CQ}^\Lmc$. Clearly, CQ$^\Lmc$ and
UCQ$^\Lmc$ contain all unary conjunctive queries (CQ) and unions of unary
conjunctive queries (UCQ), respectively. Note that UCQ$^{\ALCI}$ is a
fragment of the unary negation fragment (UNFO), a decidable fragment
of FO that generalizes many modal and description
logics~\cite{DBLP:journals/corr/SegoufinC13}. 
We next define rooted versions of these languages. We may view
formulas in CQ$^{\Lmc}$ as structures, in the obvious way by
ignoring atoms $C(t)$. For $\Lmc\in \DLS$, CQ$_{r}^{\Lmc}$
denotes the formulas $\varphi(x)$ in CQ$^{\Lmc}$ that are $\Lmc$-rooted in
$x$ and similar for UCQ$_{r}^{\Lmc}$. Finally note that, although the
languages UCQ$_{r}^{\Lmc}$ and UCQ$^{\Lmc}$ are
not syntactically closed under conjunction, every conjunction is again
equivalent to a formula in the respective language.

For any $\Lmc\in \DLS$ and signature $\Sigma$ the definition of an
$\Lmc(\Sigma)$-bisimulation $S$ between structures $\Amf$ and $\Bmf$
is standard, for details we refer to~\cite{TBoxpaper,goranko20075}. 
We write $\Amf,d\sim_{\Lmc,\Sigma}\Bmf,e$ and call pointed structures $\Amf,d$ and $\Bmf,e$
\emph{$\Lmc(\Sigma)$-bisimilar} if there exists an
$\Lmc(\Sigma)$-bisimulation $S$ such that $(d,e)\in S$. 
We say that $\Amf,d$ and $\Bmf,e$ are \emph{$\Lmc(\Sigma)$-equivalent},
in symbols $\Amf,d\equiv_{\ALCI,\Sigma}\Bmf,e$ if $d\in C^{\Amf}$ iff
$e\in C^{\Bmf}$ for all $C\in \Lmc(\Sigma)$; $\omega$-saturated structures are defined and discussed in~\cite{modeltheory}.
 
 \begin{lemma}\label{lem:equivalence}
 	Let $\Lmc\in \DLS$. Let $\Amf,d$ and $\Bmf,e$ be pointed structures
 	of finite $\Lmc$-outdegree or $\omega$-saturated and $\Sigma$ a signature. 
 	Then
 	$$
 	\Amf,d \equiv_{\Lmc,\Sigma} \Bmf,e \quad \text{ iff } \quad
 	\Amf,d \sim_{\Lmc,\Sigma}\Bmf,e.
 	$$
 	For the ``if'' direction, the condition ``finite outdegree or $\omega$-saturated''  can be
 	dropped.
 \end{lemma}
The definition of a $\Sigma$-homomorphism $h$ from a structure $\Amf$ to a structure $\Bmf$ is standard. Every database $\Dmc$ gives rise to a finite
structure $\Amf_{\Dmc}$ in the obvious way. A
$\Sigma$-homomorphism from database $\Dmc$ to structure $\Amf$ is a
$\Sigma$-homomorphism from $\Amf_{\Dmc}$ to $\Amf$.  
  
We combine homomorphisms and bisimulations to characterize the languages CQ$^{\Lmc}$ and CQ$_{r}^{\Lmc}$. Consider pointed structures $\Amf,d$ and $\Bmf,e$, and a subset $D$ of $\text{dom}(\Amf)$ such that $d\in D \subseteq \text{dom}(\Amf)$.
Let $\Lmc\in \DLS$ and $\Sigma$ a signature. Then
a \emph{CQ$^{\Lmc}(\Sigma)$-homomorphism} with domain $D$
	between $\Amf,d$ and $\Bmf,e$ is a $\Sigma$-homomorphism $h: \Amf_{|D} \rightarrow \Bmf$ such that $h(d)=e$ and $\Amf,c \sim_{\Lmc,\Sigma} \Bmf,h(c)$ for all $c\in D$. In this case we write $\Amf,d\rightarrow_{D,\Lmc,\Sigma} \Bmf,e$.

We write $\Amf,d \Rightarrow_{\text{CQ}_{r}^{\Lmc},\Sigma} \Bmf,e$ if
$\Amf\models \varphi(a)$ implies $\Bmf\models \varphi(b)$ for all
$\varphi(x)$ in $\text{CQ}_{r}^{\Lmc}(\Sigma)$, and we write $\Amf,d \Rightarrow_{\text{CQ}_{r}^{\Lmc},\Sigma}^{\text{mod}} \Bmf,e$
if for all finite $D\subseteq \text{dom}(\Amf)$ such that the $\Sigma$-reduct of $\Amf_{|D}$ is $\Lmc$-rooted in $d$, we have
$\Amf,d\rightarrow_{D,\Lmc,\Sigma} \Bmf,e$. The definitions for
CQ$^{\Lmc}$ are analogous except that the $\Sigma$-reduct of
$\Amf_{|D}$ need not be $\Lmc$-rooted in $d$.
\begin{restatable}{lemma}{lemequivalencetwo}
\label{lem:equivalence2}
	Let $\Lmc\in \DLS$ and let $\Amf,d$ and $\Bmf,e$ be pointed structures
	of finite $\Lmc$-outdegree or $\omega$-saturated, and $\Sigma$ a signature. 
	Then
	$$
	\Amf,d \Rightarrow_{\text{CQ}_{r}^{\Lmc},\Sigma} \Bmf,e \quad \text{ iff } \quad
	\Amf,d \Rightarrow^{\text{mod}}_{\text{CQ}_{r}^{\Lmc},\Sigma} \Bmf,e.
	$$
	This equivalence holds for $\text{CQ}^{\Lmc}$ if $\Amf$ and $\Bmf$ are $\omega$-saturated. In both cases, for the ``if''-direction, the condition ``finite outdegree or $\omega$-saturated'' can be dropped. 
\end{restatable}	

\section{Weak Separability with Signature}
\label{sec:weaksep}
\newcommand{\mLOKB}{labeled $\Lmc$-KB}

We start with introducing the problem of (weak) separability with
signature, in its
projective and non-projective version. 
%
%
\newcommand{\LmcO}{\Lmc}
Let $\LmcO\in \DLS$. A \emph{\mLOKB} takes the form
$(\Kmc,P,N)$ with $\Kmc=(\Omc,\Dmc)$ an $\LmcO$-KB and
$P,N\subseteq \text{ind}(\Dmc)$ non-empty sets of \emph{positive}
and \emph{negative examples}.

\begin{definition}
	\label{def:separa}
	Let $\LmcO\in \DLS$, $(\Kmc,P,N)$ be a \mLOKB, and let $\Sigma \subseteq \text{sig}(\Kmc)$ be a signature. 
	An FO-formula $\varphi(x)$ 
	\emph{$\Sigma$-separates $(\Kmc,P,N)$} if $\text{sig}(\varphi) \subseteq \Sigma \cup \Sigma_{\text{help}}$ for some set $\Sigma_{\text{help}}$ of concept names disjoint from $\text{sig}(\Kmc)$ and  
	\begin{enumerate}
		\item 
		$\Kmc\models \varphi(a)$ for all $a\in P$ and 
		\item $\Kmc\not\models \varphi(a)$ for all $a\in N$.
		
	\end{enumerate}
	%
	%
	Let $\Lmc_S$ be a fragment of FO. We say that $(\Kmc,P,N)$ is
	\emph{projectively $\Lmc_S(\Sigma)$-separable} if there is an
	$\Lmc_S$-formula $\varphi(x)$ that
	$\Sigma$-separates $(\Kmc,P,N)$ and \emph{non-projectively
		$\Lmc_S(\Sigma)$-separable} if there is such a
	$\varphi(x)$ with $\text{sig}(\varphi) \subseteq
	\Sigma$.\footnote{It is worth clarifying the interplay between nominals in the
	  separating language and individual names in $\Sigma$: If
	  $\Sigma$ does not contain individual names, then
	  $\ALCO(\Sigma)$-separability coincides with $\ALC(\Sigma)$-separability; 
	  conversely, if $\Lmc_S$ does not allow for nominals,
	  $\Lmc_S(\Sigma)$-separability coincides with
	  $\Lmc_S(\Sigma\setminus \mn{N_I})$-separability.}
	%
\end{definition}
In $\Sigma$-separating formulas, concept names from $\Sigma_{\text{help}}$ should be thought of as helper symbols.  
Their availability sometimes makes inseparable KBs separable,
examples are provided below where we also discuss the effect of admitting
role or individual names as helper symbols.
We only consider FO-fragments $\Lmc_{S}$ that are closed under
conjunction. In this case, a labeled KB $(\Kmc,P,N)$ is
$\Lmc_{S}(\Sigma)$-separable if and only if all
$(\Kmc,P,\{b\})$, $b\in N$, are
$\Lmc_{S}(\Sigma)$-separable, and likewise for projective
$\Lmc_{S}(\Sigma)$-separability~\cite{KR}. 
In
what follows, we thus mostly consider labeled KBs with singleton sets
$N$ of negative examples.

%
Each choice of an ontology language $\Lmc$ and a separation language
$\Lmc_{S}$ give rise to a projective and to a non-projective
separability problem. 
\begin{center}
	\fbox{\begin{tabular}{@{\,}l@{\;}l@{\,}}
			\small{PROBLEM}: & (Projective)
			$(\Lmc,\Lmc_{S})$-separability w.\ signature\\
			{\small INPUT}: &  A \mLOKB\xspace
                                          $(\Kmc,P,N)$ \\
                & and signature $\Sigma\subseteq \text{sig}(\Kmc)$ \\
			{\small QUESTION}: & Is $(\Kmc,P,N)$ (projectively) $\Lmc_{S}(\Sigma)$-separable?  \end{tabular}}
\end{center}
\smallskip
\noindent
If $\Lmc=\Lmc_{S}$, then we simply speak of (projective) $\Lmc$-separability. 
We study the complexity of \Lmc-separability with 
signature where the KB $\Kmc$ and 
sets of examples $P$ and $N$ are all taken to be part of the 
input. All lower bounds proved in this paper still hold if $P$ and $N$ are singleton sets.

We next provide an example that illustrates the importance of the distinction between projective and non-projective separability.
\begin{example}\label{exp:gg}
	Let $\Dmc$ contain $r(a_{1},a_{2}), \ldots, r(a_{n-1},a_{n})$,
        $r(a_n,a_1)$
	and $r(b,b_{1})$,
	where 
        $n>1$. Thus, the
	individual $a_{1}$ is part of an $r$-cycle of length $n$ but $b$ is not. 
	Let $\Omc= \{\top \sqsubseteq \exists r.\top
	\sqcap \exists r^{-}.\top\}$, 
	$\Kmc=(\Omc,\Dmc)$, 
	$P=\{a_{1}\}$,
	$N=\{b\}$, and $\Sigma=\{r\}$. Then $(\Kmc,P,N)$ is non-projectively CQ$(\Sigma)$-separable (take the CQ that states that $x$ participates in a cycle of length $n$), but $(\Kmc,P,N)$ is not non-projectively $\mathcal{ALCI}(\Sigma)$-separable
	because for any $\mathcal{\ALCI}(\Sigma)$-concept $C$ either
	$\Omc \models \top \sqsubseteq C$ or $\Omc \models C
	\sqsubseteq \bot$. If, however, a helper symbol $A$ is
	allowed, then $A \rightarrow \exists r^{n}.A$ $\Sigma$-separates $(\Kmc,P,N)$. 
\end{example}	
We discuss the effect of also admitting individual names as helper symbols.
Then already for $\mathcal{ALC}$-KBs, projective inseparability becomes undecidable. The proof is inspired by reductions of undecidable tiling problems in the context of conservative extensions and modularity~\cite{DBLP:conf/ijcai/LutzWW07,DBLP:journals/jair/GrauHKS08}.
\begin{restatable}{theorem}{thmnominalhelper}
  \label{thm:nominalhelper} 
  Projective $(\ALC,\ALCO)$-separability with signature becomes undecidable when
  additionally individual names are admitted as helper symbols. 
\end{restatable}
Admitting role names as helper symbols has a less dramatic impact. For  $\mathcal{ALC}$ and $\mathcal{ALCI}$-separability they do not make any difference 
at all and for $\mathcal{ALCO}$ and $\mathcal{ALCIO}$ their effect can be 
captured by a single additional role name which enables a straightforward  polynomial reduction to separability without role names as helper symbols.
\begin{restatable}{theorem}{proprolehelp}
	\label{prop:proprolehelp}
	(1) Let $\Lmc\in \{\ALC,\ALCI\}$. Then projective $\Lmc$-separability coincides with projective $\Lmc$-separability with concept and role names as helper symbols.
	
	(2) Let $\Lmc\in \{\mathcal{ALCO},\mathcal{ALCIO}\}$ and $(\Kmc,P,N)$ be a labeled $\Lmc$-KB and $\Sigma\subseteq \text{sig}(\Kmc)$ a signature. Let $r_{I}$ be a fresh role name and let $\Kmc'$ be the extension of $\Kmc$ by the `dummy' inclusion $\exists r_{I}.\top \sqsubseteq \exists r_{I}.\top$. Then the following conditions are equivalent:
	\begin{itemize}
		\item $(\Kmc,P,N)$ is projectively $\Lmc(\Sigma)$-separable with concept and role names as helper symbols;
		\item $(\Kmc',P,N)$ is projectively $\Lmc(\Sigma\cup \{r_{I}\})$-separable.
	\end{itemize}
\end{restatable}
The proof uses the model-theoretic characterization of separability
given in Theorem~\ref{thm:L-modeltheory0} below. The next example
illustrates the use of a  helper role name in $\mathcal{ALCO}$.
\begin{example} Let $\Kmc=(\Omc,\Dmc)$, where 
	$
	\Omc= \{A_{0}\sqcap \exists r.\top \sqsubseteq \bot, 
	B \sqsubseteq \forall r.A\}
	$ 
	and $\Dmc= \{r(c,a),A_{0}(a),A_{0}(b)\}$. Let $\Sigma= \{c,B,A\}$. 
	Then $(\Kmc,\{a\},\{b\})$ is not projectively
        $\mathcal{ALCIO}(\Sigma)$-separable, but the $\ALCO(\Sigma)$-concept
	$\exists r_{I}.(\{c\} \sqcap B) \rightarrow A$ separates $(\Kmc,\{a\},\{b\})$ using the helper symbol $r_{I}$.
	%
\end{example}
We next make first observations regarding the separating power of
several relevant separating languages.
In~\cite{DBLP:conf/ijcai/FunkJLPW19,KR}, projective and non-projective
separability are studied without signature restrictions, that
is, 
all symbols used in the KB except individual names can appear in
separating formulas. We call this the \emph{full relational
  signature}. Surprisingly, it turned out that in this case many
different separation languages have exactly the same separating
power. In particular, a labeled $\ALCI$-KB is FO-separable iff it is
UCQ-separable,  and projective and non-projective separability
coincide. 
No such
result can be expected for separability with signature restrictions,
as illustrated by the next example.
\begin{example}
	Let $\Kmc=(\Omc,\Dmc)$, where $\Omc = \{A \sqsubseteq \exists
	r.B \sqcap \exists r.\neg B\}$ and $\Dmc= \{A(a),r(b,c)\}$. Let $P=\{a\}$, $N=\{b\}$,
	and $\Sigma=\{r\}$. Clearly, the formula
	$$
	\exists y \exists y' (r(x,y) \wedge r(x,y') \wedge \neg (y=y'))
	$$
	$\Sigma$-separates $(\Kmc,P,N)$, but $(\Kmc,P,N)$ is not
        UCQ$(\Sigma)$-separable. 
\end{example}
It is also shown in~\cite{DBLP:conf/ijcai/FunkJLPW19,KR} that for
labeled \ALCI-KBs and with the full relational signature,
UCQ-separability (projectively or not) coincides with projective
$\mathcal{ALCI}$-separability. The next example shows that with
restricted signatures, it is not even true that non-projective
$\mathcal{ALCI}$-separability implies UCQ-separability.
\begin{example}
	Let $\Omc = \{A \sqsubseteq \forall r.B\}$ and $\Dmc= \{A(a),C(b)\}$. Let $P=\{a\}$, $N=\{b\}$,
	and $\Sigma=\{r,B\}$. Clearly, the \ALC-concept $\forall r.B$
	$\Sigma$-separates $(\Omc,\Dmc,P,N)$, but $(\Omc,\Dmc,P,N)$ is not
	UCQ$(\Sigma)$-separable.
\end{example}
Conversely, it follows from Example~\ref{exp:gg} that UCQ-separability
does not imply non-projective $\mathcal{ALCI}$-separability, even with
the full relational signature. Interestingly, in the projective case,
this implication holds even with restricted signatures: every
UCQ($\Sigma$)-separable labeled $\mathcal{ALCI}$-KB is also
projectively $\mathcal{ALCI}(\Sigma)$-separable. This follows from
more powerful equivalences proved below (Theorem~\ref{thm:equival}).

In this paper, we mainly focus on projective separability. In fact, it
emerges from \cite{DBLP:conf/ijcai/FunkJLPW19,KR} that insisting on
non-projective separability is a source of significant technical
difficulties while not always delivering more natural separating
concepts. As our main aim is to study the impact of signature
restrictions on separability, which is another source of significant
technical challenges, we prefer to leave out the first such
source and stick with projective separability.

We close this introduction with the observation that 
in contrast to the case of full relational signatures,
FO-separability with signature is undecidable for labeled $\ALC$-KBs. We prove this
using the same technique as for
Theorem~\ref{thm:nominalhelper}. Undecidability applies even when one
separates in the decidable extension $\mathcal{ALCFIO}$ of
$\mathcal{ALCIO}$ with unqualified  number restrictions of the form $(\leq 1\; r)$.

%
\begin{restatable}{theorem}{thmfoundec}
	\label{thm:foundec}
	$(\mathcal{ALC},\Lmc_{S})$-separability with signature is
	undecidable for any fragment $\Lmc_{S}$ of FO that contains $\mathcal{ALCFIO}$, both in the projective and non-projective case.
\end{restatable}	
\section{Model-Theoretic Criteria and Equivalence Results}
\label{sec:modelandequi}
We provide powerful model-theoretic criteria that underly the decision
procedures given later on. Moreover, we use these criteria to establish
equivalences between projective separability and non-projective
separability in more expressive languages that shed light on the role of
helper symbols.
%

We start with the model-theoretic criteria using \emph{functional}
bisimulations.  For $\Lmc\in \DLS$ we write
$\Amf,a \sim_{\Lmc,\Sigma}^{f} \Bmf,b$ if there exists an
$\Lmc(\Sigma)$-bisimulation $S$ between $\Amf$ and $\Bmf$ that
contains $(a,b)$ and is \emph{functional}, that is,
$(d,d_{1}),(d,d_{2})\in S$ implies $d_{1}=d_{2}$. Note that
$\Amf,a \sim_{\Lmc,\Sigma}^{f} \Bmf,b$ implies that there is a
homomorphism from $\Amf,a$ to $\Bmf,b$ if \Amf is connected
and $\Lmc=\ALCI$, but not otherwise.
\begin{restatable}{theorem}{thmLmodeltheorynull}
	\label{thm:L-modeltheory0}
	Let $\Lmc\in \{\ALC,\ALCI,\mathcal{ALCO}\}$. Assume that $(\Kmc,P,\{b\})$ is a labeled $\Lmc$-KB with $\Kmc=(\Omc,\Dmc)$
	and $\Sigma\subseteq \text{sig}(\Kmc)$. Then the following
	conditions are equivalent:
	\begin{enumerate}
		
		\item $(\Kmc,P,\{b\})$ is projectively
		$\Lmc(\Sigma)$-separable.
		
		\item there exists an $\Lmc$-forest model $\Amf$ of $\Kmc$ of finite
		$\Lmc$-outdegree and a set $\Sigma_{\text{help}}$ of concept names disjoint from $\text{sig}(\Kmc)$ such that for all models $\Bmf$ of
		$\Kmc$ and all $a\in P$: $\Bmf,a^{\Bmf}
		\not\sim_{\Lmc,\Sigma\cup \Sigma_{\text{help}}} \Amf, b^{\Amf}$.
		
		\item there exists an $\Lmc$-forest model $\Amf$ of $\Kmc$ of finite
		$\Lmc$-outdegree such that for all models $\Bmf$ of $\Kmc$ and all
		$a\in P$: $\Bmf,a^{\Bmf} \not\sim_{\Lmc,\Sigma}^{f} \Amf,
		b^{\Amf}$. 
	\end{enumerate}
\end{restatable}
The equivalence between Points~1 and~2 of Theorem~\ref{thm:L-modeltheory0} is a direct consequence of the following characterization in the non-projective case (which can be proved
using Lemmas~\ref{lem:forestm} and~\ref{lem:equivalence}): a labeled $\Lmc$-KB
$(\Kmc,P,\{b\})$ is non-projectively $\Lmc(\Sigma)$-separable
iff there exists an $\Lmc$-forest model $\Amf$ of $\Kmc$ of finite
$\Lmc$-outdegree such that for all models $\Bmf$ of $\Kmc$ and all $a\in P$:
$\Bmf,a^{\Amf} \not\sim_{\Lmc,\Sigma} \Amf,b^{\Amf}$. Due to cycles in the databases the general bisimulations used in this criterion and in 
Point~2 of Theorem~\ref{thm:L-modeltheory0} are hard to encode in an automata based decision procedure. Moreover, in Point~2 one has to ``guess'' the
number of helper symbols needed. The criterion given in Point~3, in contrast, is much better suited for this purpose
and does not speak about helper symbols.

The equivalence of 2.\ and 3.\ is surprisingly straightforward to show as one can work
with the same model $\Amf$. As Lemma~\ref{lem:forestm} fails to hold
for $\Lmc=\mathcal{ALCIO}$, Theorem~\ref{thm:L-modeltheory0} also does
not hold for this choice of \Lmc. An example that illustrates the
situation is given in the appendix.

As a first important application of Theorem~\ref{thm:L-modeltheory0},
we show that projective $\ALCI$-separability is equivalent to
non-projective separability in UCQ$_{r}^{\ALCI}$ and that projective
$(\ALC,\mathcal{ALCO})$-separability is equivalent to non-projective
$(\ALC,\text{UCQ}_{r}^{\mathcal{ALCO}})$-separability. 
%
The following example illustrates why the languages UCQ$_{r}^{\Lmc}$
can non-projectively separate labeled KBs that cannot be separated
non-projectively in a natural way in languages from $\DLS$.
\begin{example}
	Let $\Kmc=(\Omc,\Dmc)$, where $\Omc = \{B \sqsubseteq \forall
	t.A\}$ and $\Dmc$ is depicted below:
	 
	\begin{center}

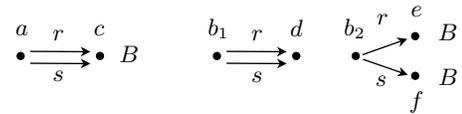
		\tikzset{every picture/.style={line width=0.5pt}} 
		
		\begin{tikzpicture}[x=0.75pt,y=0.75pt,yscale=-1,xscale=1]
			
			\draw  [fill={rgb, 255:red, 0; green, 0; blue, 0 }  ,fill opacity=1 ] (262.1,270.25) .. controls (262.1,269.28) and (261.32,268.5) .. (260.35,268.5) .. controls (259.38,268.5) and (258.6,269.28) .. (258.6,270.25) .. controls (258.6,271.22) and (259.38,272) .. (260.35,272) .. controls (261.32,272) and (262.1,271.22) .. (262.1,270.25) -- cycle ;
			\draw  [fill={rgb, 255:red, 0; green, 0; blue, 0 }  ,fill opacity=1 ] (302.1,270.25) .. controls (302.1,269.28) and (301.32,268.5) .. (300.35,268.5) .. controls (299.38,268.5) and (298.6,269.28) .. (298.6,270.25) .. controls (298.6,271.22) and (299.38,272) .. (300.35,272) .. controls (301.32,272) and (302.1,271.22) .. (302.1,270.25) -- cycle ;
			\draw  [fill={rgb, 255:red, 0; green, 0; blue, 0 }  ,fill opacity=1 ] (361.1,270.25) .. controls (361.1,269.28) and (360.32,268.5) .. (359.35,268.5) .. controls (358.38,268.5) and (357.6,269.28) .. (357.6,270.25) .. controls (357.6,271.22) and (358.38,272) .. (359.35,272) .. controls (360.32,272) and (361.1,271.22) .. (361.1,270.25) -- cycle ;
			\draw  [fill={rgb, 255:red, 0; green, 0; blue, 0 }  ,fill opacity=1 ] (401.1,270.25) .. controls (401.1,269.28) and (400.32,268.5) .. (399.35,268.5) .. controls (398.38,268.5) and (397.6,269.28) .. (397.6,270.25) .. controls (397.6,271.22) and (398.38,272) .. (399.35,272) .. controls (400.32,272) and (401.1,271.22) .. (401.1,270.25) -- cycle ;
			\draw  [fill={rgb, 255:red, 0; green, 0; blue, 0 }  ,fill opacity=1 ] (431.1,270.25) .. controls (431.1,269.28) and (430.32,268.5) .. (429.35,268.5) .. controls (428.38,268.5) and (427.6,269.28) .. (427.6,270.25) .. controls (427.6,271.22) and (428.38,272) .. (429.35,272) .. controls (430.32,272) and (431.1,271.22) .. (431.1,270.25) -- cycle ;
			\draw  [fill={rgb, 255:red, 0; green, 0; blue, 0 }  ,fill opacity=1 ] (461.1,260.25) .. controls (461.1,259.28) and (460.32,258.5) .. (459.35,258.5) .. controls (458.38,258.5) and (457.6,259.28) .. (457.6,260.25) .. controls (457.6,261.22) and (458.38,262) .. (459.35,262) .. controls (460.32,262) and (461.1,261.22) .. (461.1,260.25) -- cycle ;
			\draw  [fill={rgb, 255:red, 0; green, 0; blue, 0 }  ,fill opacity=1 ] (461.1,280.25) .. controls (461.1,279.28) and (460.32,278.5) .. (459.35,278.5) .. controls (458.38,278.5) and (457.6,279.28) .. (457.6,280.25) .. controls (457.6,281.22) and (458.38,282) .. (459.35,282) .. controls (460.32,282) and (461.1,281.22) .. (461.1,280.25) -- cycle ;
			\draw    (265.18,267.92) -- (292.56,268.07) ;
			\draw [shift={(295.56,268.08)}, rotate = 180.31] [fill={rgb, 255:red, 0; green, 0; blue, 0 }  ][line width=0.08]  [draw opacity=0] (5,-2.5) -- (0,0) -- (5,2.5) -- (3.5,0) -- cycle    ;
			\draw    (432.85,269.58) -- (452.54,262.61) ;
			\draw [shift={(455.36,261.61)}, rotate = 520.49] [fill={rgb, 255:red, 0; green, 0; blue, 0 }  ][line width=0.08]  [draw opacity=0] (5,-2.5) -- (0,0) -- (5,2.5) -- (3.5,0) -- cycle    ;
			\draw    (433.02,271.92) -- (452.78,278.77) ;
			\draw [shift={(455.61,279.75)}, rotate = 199.12] [fill={rgb, 255:red, 0; green, 0; blue, 0 }  ][line width=0.08]  [draw opacity=0] (5,-2.5) -- (0,0) -- (5,2.5) -- (3.5,0) -- cycle    ;
			\draw    (265.18,273.92) -- (292.56,274.07) ;
			\draw [shift={(295.56,274.08)}, rotate = 180.31] [fill={rgb, 255:red, 0; green, 0; blue, 0 }  ][line width=0.08]  [draw opacity=0] (5,-2.5) -- (0,0) -- (5,2.5) -- (3.5,0) -- cycle    ;
			\draw    (364.18,267.92) -- (391.56,268.07) ;
			\draw [shift={(394.56,268.08)}, rotate = 180.31] [fill={rgb, 255:red, 0; green, 0; blue, 0 }  ][line width=0.08]  [draw opacity=0] (5,-2.5) -- (0,0) -- (5,2.5) -- (3.5,0) -- cycle    ;
			\draw    (364.18,273.92) -- (391.56,274.07) ;
			\draw [shift={(394.56,274.08)}, rotate = 180.31] [fill={rgb, 255:red, 0; green, 0; blue, 0 }  ][line width=0.08]  [draw opacity=0] (5,-2.5) -- (0,0) -- (5,2.5) -- (3.5,0) -- cycle    ;
			
			\draw (256.41,253.56) node [anchor=north west][inner sep=0.75pt]  [font=\small] [align=left] {$\displaystyle a$};
			\draw (295.96,253.11) node [anchor=north west][inner sep=0.75pt]  [font=\small] [align=left] {$\displaystyle c$};
			\draw (353.41,250.33) node [anchor=north west][inner sep=0.75pt]  [font=\small] [align=left] {$\displaystyle b_{1}$};
			\draw (394.96,250.67) node [anchor=north west][inner sep=0.75pt]  [font=\small] [align=left] {$\displaystyle d$};
			\draw (422.07,250.33) node [anchor=north west][inner sep=0.75pt]  [font=\small] [align=left] {$\displaystyle b_{2}$};
			\draw (455.85,243.78) node [anchor=north west][inner sep=0.75pt]  [font=\small] [align=left] {$\displaystyle e$};
			\draw (454.96,286.89) node [anchor=north west][inner sep=0.75pt]  [font=\small] [align=left] {$\displaystyle f$};
			\draw (469.54,252.87) node [anchor=north west][inner sep=0.75pt]  [font=\small] [align=left] {$\displaystyle B$};
			\draw (469.76,275.42) node [anchor=north west][inner sep=0.75pt]  [font=\small] [align=left] {$\displaystyle B$};
			\draw (308.85,263.91) node [anchor=north west][inner sep=0.75pt]  [font=\small] [align=left] {$\displaystyle B$};
			\draw (275.21,255.76) node [anchor=north
			west][inner sep=0.75pt]  [font=\small]
			[align=left] {$\displaystyle r$};
			\draw (375.41,255.76) node [anchor=north
			west][inner sep=0.75pt]  [font=\small]
			[align=left] {$\displaystyle r$};
			\draw (438.61,247.56) node [anchor=north
			west][inner sep=0.75pt]  [font=\small]
			[align=left] {$\displaystyle r$};
			\draw (275.21,276.36) node [anchor=north
			west][inner sep=0.75pt]  [font=\small]
			[align=left] {$\displaystyle s$};
			\draw (375.41,277.16) node [anchor=north
			west][inner sep=0.75pt]  [font=\small]
			[align=left] {$\displaystyle s$};
			\draw (438.01,278.76) node [anchor=north
			west][inner sep=0.75pt]  [font=\small]
			[align=left] {$\displaystyle s$};

		\end{tikzpicture}
	\end{center}
Let $P=\{a\},N=\{b_{1},b_{2}\}$, and $\Sigma=\{r,s,t,A\}$. Then 
$$
\exists y\, r(x,y) \wedge s(x,y) \wedge (\forall t.A)(y) \in \text{CQ}_{r}^{\ALC}
$$
$\Sigma$-separates
$(\Kmc,P,N)$. The `simplest' $\ALC$-concept $\Sigma$-separating
$(\Kmc,P,N)$ is $(\exists r.\forall t.A) \sqcap (\forall r.X
\rightarrow \exists s.X)$, where $X$ is fresh.
\end{example}

We next state the announced equivalences. Informally spoken, they show
that admitting helper concept names corresponds to `adding rooted UCQs'.
\begin{restatable}{theorem}{thmequival}
   \label{thm:equival}
	Let $(\Lmc,\Lmc_{S})$ be either $(\ALCI,\ALCI)$ or $(\ALC, \mathcal{ALCO})$ and let $(\Kmc,P,\{b\}\})$ be a labeled $\Lmc$-KB and $\Sigma\subseteq \text{sig}(\Kmc)$ a signature. Then the following conditions are equivalent:
\begin{enumerate}
	\item $(\Kmc,P,\{b\})$ is projectively $\Lmc_{S}(\Sigma)$-separable;
	\item $(\Kmc,P,\{b\})$ is non-projectively UCQ$_{r}^{\Lmc_{S}}(\Sigma)$-separable. 
\end{enumerate}  
\end{restatable}
\begin{proof} \
The proof has two main steps. First, using Lemma~\ref{lem:equivalence2},
one can characterize non-projective UCQ$_{r}^{\Lmc_{S}}(\Sigma)$-separability 
in terms of CQ$^{\Lmc_{S}}(\Sigma)$-homomorphisms. Namely, $(\Kmc,P,\{b\})$ is non-projectively UCQ$_{r}^{\Lmc_{S}}(\Sigma)$-separable iff there exist an $\Lmc_{S}$-forest model $\Amf$ of $\Kmc$ of finite
	$\Lmc_{S}$-outdegree and $n>0$ such that for all models $\Bmf$ of $\Kmc$ and all $a\in P$, $\Bmf,a^{\Bmf}\not\rightarrow_{D,\Lmc_{S},\Sigma} \Amf,
	b^{\Amf}$, for some $D$ with $|D|\leq n$ such that the $\Sigma$-reduct
	of $\Bmf_{|D}$ is $\Lmc_{S}$-rooted in $a^{\Bmf}$. Secondly, one can prove that this characterization is equivalent to Condition~3 of Theorem~\ref{thm:L-modeltheory0}. Observe, for example, that functional $\Sigma$-bisimulations give rise to the combination of $\Sigma$-homomorphisms and $\Sigma$-bisimulations given in the characterization above.
\end{proof}
We observe that the equivalences of Theorem~\ref{thm:equival}
do not hold when the ontology contains nominals.
\begin{example}
	Let $\Kmc=(\Omc,\Dmc)$, where 
	$
	\Omc= \{\{a\} \sqsubseteq \forall r.\{a\}, 
	\top \sqsubseteq \exists r.\top\},
	$
    and $\Dmc = \{A(a),r(b,b)\}$. Let $\Sigma=\{r\}$. Then $(\Kmc,\{a\},\{b\})$ is 
	projectively separated by the $\ALC(\Sigma)$-concept
        $X\rightarrow \forall r.X$ with $X$ a fresh concept name, but it is not
	non-projectively \text{UCQ}$_{r}^{\mathcal{ALCO}}(\Sigma)$-separable.
\end{example}
It remains open whether there is any natural fragment of FO such that
a labeled $\mathcal{ALCO}$-KBs is non-projectively separable in the fragment if
and only if it is projectively separable in $\mathcal{ALCO}$.

\section{The Complexity of Weak Separability}
\label{sec:complexity}

We study the decidability and computational complexity of projective
$\Lmc$-separability for $\Lmc\in\{\ALC,\ALCI,\ALCO\}$.  The
results established in this section are closely related to
conservative extensions of ontologies and we also observe new
results for that problem. For \Lmc-ontologies 
$\Omc$ and $\Omc'$, 
we say that $\Omc \cup \Omc'$
is a \emph{conservative extension of $\Omc$ in $\Lmc$} if,  for all concept inclusions
$C \sqsubseteq D$ with $C,D$ $\Lmc$-concepts that use only symbols from
$\text{sig}(\Omc)$: if 
$\Omc \cup \Omc'$ entails $C \sqsubseteq D$ then already
$\Omc$ entails $C \sqsubseteq D$. \emph{Projective conservative extensions in $\Lmc$} are
defined in the same way except that $C$ and $D$ may additionally use
fresh concept names, that is, concept names that are not in
$\text{sig}(\Omc \cup \Omc')$. If $\Omc \cup \Omc'$ is not a
conservative extension of $\Omc$ in $\Lmc$, then there exists an \Lmc-concept
$C$ that uses only symbols from $\text{sig}(\Omc)$ and is satisfiable
w.r.t.~$\Omc$, but not w.r.t.\ $\Omc \cup \Omc'$. We
call such a concept $C$ a \emph{witness concept} for $\Omc$ and
$\Omc'$.

%
%
\begin{restatable}{lemma}{lemcereduction}\label{lem:ce-reduction}
  Let $\Lmc\in \DLS$. Then deciding conservative extensions in $\Lmc$
  can be reduced in polynomial time to the complement of
  $\Lmc$-separability, both in the projective and non-projective case.
\end{restatable}
\begin{proof}
  \ The proof uses relativizations. Intuitively, given $\Omc$ and
  $\Omc'$ one computes a new ontology $\Omc_1$ which contains~$\Omc$
  and the relativization of $\Omc'$ to a fresh concept name $A$.
  Then, a concept $C$ is a witness concept for $\Omc$ and $\Omc'$ iff
  $\neg C$ separates (w.r.t.~$\Omc_1$) an individual that satisfies
  $A$ from an individual that does not satisfy $A$. If $\Lmc$ contains
  nominals the proof is slightly more involved.
\end{proof}

We start our analysis with the DLs $\ALC$ and
$\ALCI$.
\begin{theorem}\label{thm:complexity-alci}
  Projective \Lmc-separability with signature is \TwoExpTime-complete,
  for $\Lmc\in\{\ALC,\ALCI\}$.
\end{theorem}
The lower bound follows from Lemma~\ref{lem:ce-reduction} and also
holds for non-projective separability. In fact, it is known that
deciding (non-projective) conservative extensions in
$\Lmc\in \{\ALC,\ALCI\}$ is
\TwoExpTime-hard~\cite{DBLP:conf/kr/GhilardiLW06,DBLP:conf/ijcai/LutzWW07} and that
conservative extensions and projective conservative extensions
coincide in logics that enjoy Craig interpolation~\cite{JLMSW17},
which \ALC and \ALCI do.

For the upper bound, we concentrate on \ALCI; the case
of \ALC is very similar, but simpler. The idea is to use 
two-way alternating tree automata
(2ATA)~\cite{DBLP:conf/icalp/Vardi98} to decide Condition~3 of
Theorem~\ref{thm:L-modeltheory0}. More precisely, given
$(\Kmc,P,\{b\}),\Sigma$ with $\Kmc=(\Omc,\Dmc)$, we construct a 2ATA
\Amc such that the language recognized by \Amc is non-empty
if and only if there is a forest model \Amf of \Kmc as described in
Condition~3 of Theorem~\ref{thm:L-modeltheory0}. The use of tree
automata is enabled by the fact that Condition~3 refers to
\emph{forest models} of~\Kmc. Indeed, forest structures can be encoded
in labeled trees using an appropriate alphabet. Intuitively, each node
in the tree corresponds to an element in the forest structure and the
label contains its type, the connection to its predecessor, and
connections to individuals from \Dmc.

It is not difficult to devise a
2ATA \Bmc (of polynomial size) that recognizes the finite outdegree forest
models of \Kmc, see e.g.~\cite{JLMSW17}. 
Observe next that it suffices to construct, for each $a\in P$, a 2ATA
$\Amc_a$ such that $\Amc_a$ accepts \Amf iff 
\begin{description}

  \item[$(\ast_a)$] there is a model \Bmf of \Kmc with
    $\Bmf,a^\Bmf\sim^f_{\ALCI,\Sigma} \Amf,b^\Amf$.

\end{description}
Indeed, a 2ATA that recognizes the following language is as required:
\[L(\Bmc)\cap\textstyle\bigcap_{a\in P}\overline{L(\Amc_a)}\]
where $\overline L$ denotes the complement of $L$.
As complementation and intersection of 2ATAs involve only a polynomial
blowup, we obtain the desired 2ATA $\Amc$ from $\Bmc$ and the
$\Amc_a$.  

In principle, the existence of a (not necessarily functionally)
bisimilar model \Bmf can be checked using alternating automata as
follows. We assume w.l.o.g.\ that the model \Bmf is a forest model, because we
can always consider an appropriate unraveling.  Then,
the alternating automaton `virtually' traverses $\Bmf$
element-by-element, storing at each moment only the type of the
current element in its state and visiting a bisimilar element in \Amf. Alternation is crucial as the
automaton has to extend the bisimulation \emph{for all} successors of
the current element in $\Bmf$ and symmetrically \emph{for all}
successors of the currently visited element in \Amf. Functionality of the
bisimulation poses a challenge:
different parts of the run of the automaton can visit the same
individual from \Dmc in \Bmf, and functionality requires that the automaton
visits the same element in \Amf. 
In order to solve that (and get tight bounds), we replace~$(\ast_a)$
with an equivalent condition in which functional bisimulations are
carefully `compiled away'. 

We introduce some additional notation.  An \emph{extended database} is
a database that additionally may contain `atoms' of the form $C(a)$
with $C$ an \ALCI-concept. The semantics of extended databases is
defined in the expected way.  Let $\text{sub}(\Kmc)$ denote the set of
concepts that occur in $\Kmc$, closed under single negation and under
subconcepts. The \emph{$\Kmc$-type realized in a pointed structure
$\Amf,a$} is defined as $$ \text{tp}_{\Kmc}(\Amf,a) = \{ C\in
  \text{sub}(\Kmc) \mid a\in C^{\Amf}\}.  $$ A \emph{$\Kmc$-type} is
  any set $t \subseteq \text{sub}(\Kmc)$ of the form
  $\text{tp}_{\Kmc}(\Amf,a)$.  For a pointed database $\Dmc,a$, we
  write $\Dmc_{\text{con}(a)},a \rightarrow^{\Sigma}_{c} \Amf,b^{\Amf}$ if
  there is a $\Sigma$-homomorphism $h$ from the maximal connected
  component $\Dmc_{\text{con}(a)}$ of $a$ in $\Dmc$ to $\Amf$ such
  that $h(a)=b^\Amf$ and there is a $\Kmc$-type $t_{d}$ for each $d\in
  \text{ind}(\Dmc_{\text{con}(a)})$ such that:
\begin{enumerate}[label=(\roman*)]

  \item there exists a model $\Bmf_{d}$ of $\Omc$ with
    $\text{tp}_{\Kmc}(\Bmf_{d},d) = t_{d}$ and
    $\Bmf_{d},d\sim_{\ALCI,\Sigma} \Amf,h(d)$; 

  \item $(\Omc,\Dmc')$ is satisfiable, for the extended database
    $\Dmc'=\Dmc\cup \{ C(d) \mid C \in t_d, \ d\in
      \text{ind}(\Dmc_{\text{con}(a)})\}$.

\end{enumerate}
\begin{restatable}{lemma}{lemequivalci}\label{lem:equiv}
  For all forest models \Amf of \Kmc and all $a\in P$,
  Condition~$(\ast_a)$ is equivalent to $\Dmc_{\text{con}(a)},a
  \rightarrow^{\Sigma}_{c} \Amf,b^{\Amf}$.
\end{restatable}
Intuitively, the homomorphism $h$ fixes the image of the
bisimulation of the individuals from \Dmc,
and a 2ATA can decide $\Dmc_{\text{con}(a)},a
\rightarrow^{\Sigma}_{c} \Amf,b^{\Amf}$ as follows. It 
first non-deterministically guesses types $t_d$,
$d\in\text{ind}(\Dmc_{\text{con}(a)})$ that satisfy Item~(ii) above and stores them in its states.
Then it gradually guesses a $\Sigma$-homomorphism from
$\Dmc_{\text{con}(a)}$ to $\Amf$. Whenever, it guesses a new image
$h(d)$ for some $d$, it verifies
the bisimulation condition in Item~(i) as described above. 
Overall, $\Amc_a$ (and thus $\Amc$) uses exponentially many states. The
\TwoExpTime upper bound follows as non-emptiness can be decided in
exponential time~\cite{DBLP:conf/icalp/Vardi98}. 


For \ALCO, we show the surprising result that separability becomes
harder than in \ALC and \ALCI, by one exponent. We establish the same
result also for the more basic problem of conservative extensions.
\begin{theorem}\label{thm:complexity-alco} Projective
  \ALCO-separability with signature and projective conservative
  extensions in \ALCO are \ThreeExpTime-complete.  \end{theorem}
We show in the appendix that the lower bound also applies to
non-projective conservative extensions and, by Lemma~\ref{lem:equiv},
to non-projective  \ALCO-separability with signature. An upper bound for that case
remains open. The upper bound easily extends to the variant of
projective conservative extensions where we are interested only in the
entailment of concept inclusions $C \sqsubseteq D$ formulated in a
given subsignature $\Sigma \subseteq \text{sig}(\Omc)$,
c.f.~\cite{DBLP:conf/kr/GhilardiLW06}.

By Lemma~\ref{lem:ce-reduction}, it suffices to show the lower bound
in Theorem~\ref{thm:complexity-alco} for conservative extensions and
the upper bound for separability. We start with the former, which is
proved by reduction of the word problem of double exponentially space
bounded ATMs. The reduction strategy follows and extends the one used
in \cite{DBLP:conf/kr/GhilardiLW06} to prove that deciding conservative extensions of
\ALC-ontologies is \TwoExpTime-complete. The reduction proceeds in two
steps. First, for every $n \geq 1$ one crafts ontologies $\Omc_n$ and
$\Omc'_n$ of size polynomial in $n$ such that $\Omc_n \cup \Omc'_n$ is
not a conservative extension of $\Omc_n$, but all witness concepts for
$\Omc_n$ and $\Omc'_n$ are of size quadruple exponential in $n$. More
precisely, $\Omc_n$ and $\Omc'_n$ implement a binary counter that is
able to count the length of role paths up to $2^{2^{2^n}}$ and witness
concepts need to enforce a binary tree of that depth. The triple
exponential counter is implemented by building on a double exponential
counter which in turn builds on a single exponential counter. The two
latter counters are implemented exactly as in~\cite{DBLP:conf/kr/GhilardiLW06} and
the implementation of the triple exponential counter crucially uses a
nominal. In fact, $\Omc_n$ does not use any nominals and a single
nominal in $\Omc'_n$ suffices. The implementation of the counters is
quite subtle. For the third counter, we independently send multiple
$\Omc'_n$-types down a path in the binary tree generated by a witness
concept and use the nominal to `re-synchronize' them again later.  In
the second step of the reduction, we simulate the computation of a fixed ATM on a
given input in the binary trees of triple exponential depth generated by
witness concepts for $\Omc_n$ and $\Omc'_n$.

\smallskip

For the upper bound in Theorem~\ref{thm:complexity-alco}, we again
pursue an automata-based approach. As for \ALCI, we encode forest
structures as inputs to 2ATAs and the goal is to construct a 2ATA
$\Amc_a$ that accepts an input \Amf if and only if
Condition~$(\ast_a)$ is true, with \ALCI replaced by \ALCO. However,
instead of going via an intermediate characterization such as
Lemma~\ref{lem:equiv}, we directly use~$(\ast_a)$ (at the cost of
one exponent).
 

The problem of synchronizing different visits of the
individuals in \Dmc during the (virtual) construction of \Bmf is
addressed as follows. We first construct a 2ATA $\Amc_a'$ over an extended
alphabet. A labeled tree over that alphabet does not
only contain the structure $\Amf$, but also marks a \emph{possible}
choice of the elements in \Amf that are bisimilar to the individuals
in \Dmc. Now, when the automaton is in a state representing an
individual $d\in \text{ind}(\Dmc)$ during the construction of \Bmf, it 
ensures that the currently visited element of \Amf is marked with
$d$ in the input. 
The desired automaton
$\Amc_a$ is then obtained by projecting $\Amc_a'$ to the original
input alphabet.

The 2ATA $\Amc_a'$ can be constructed in exponential time and has at
most exponentially many states. Since projection of alternating
automata involves an exponential blow-up, $\Amc_a$ is of double
exponential size. 
%
Together with the exponential non-emptiness test, we obtain the \ThreeExpTime-upper
bound.

\section{Strong Separability with Signature}
\label{sec:dfstrong}

We discuss strong separability of labeled KBs. The crucial
difference to weak separability is that the negation of the separating
formula must be entailed at all negative examples.

\begin{definition}
	Let $\Lmc\in \DLS$, $(\Kmc,P,N)$ be a labeled $\Lmc$-KB, and let $\Sigma\subseteq \text{sig}(\Kmc)$ be a signature.
	An FO-formula $\varphi(x)$ \emph{strongly $\Sigma$-separates}
	$(\Kmc,P,N)$ if $\text{sig}(\varphi)\subseteq \Sigma$ and
	\begin{enumerate}
		
		\item 
		$\Kmc\models \varphi(a)$ for all $a\in P$ and 
		
		\item $\Kmc\models \neg\varphi(a)$ for all $a\in N$.
		
	\end{enumerate}
	Let $\Lmc_S$ be a fragment of FO. We say that $(\Kmc,P,N)$ is
	\emph{strongly $\Lmc_S(\Sigma)$-separable} if there 
	exists an
	$\Lmc_S$-formula $\varphi(x)$ that strongly $\Sigma$-separates
	$(\Kmc,P,N)$.	
\end{definition}
In contrast to weak separability, we do not consider a projective
version of strong separability as any formula $\varphi$ that strongly
$\Sigma$-separates a labeled KB $(\Kmc,P,N)$ and uses helper symbols
can easily be transformed into a strongly separating formula that uses
only symbols from $\Sigma$: simply replace any occurrence of such a
formula $A(x)$, $A\not\in \Sigma$, by $x=x$ (or a concept name $A$ by
$\top$).  Then, if $\varphi$ strongly separates $(\Kmc,P,N)$, so does
the resulting formula $\varphi'$.

Note that for languages $\Lmc_{S}$ closed under conjunction and
disjunction a labeled KB $(\Kmc,P,N)$ is strongly
$\Lmc_{S}(\Sigma)$-separable iff every
$(\Kmc,\{a\},\{b\})$ with $a\in P$ and
$b\in N$ is strongly $\Lmc_{S}(\Sigma)$-separable.  In fact, if
$\varphi_{a,b}$ strongly separates
$(\Kmc,\{a\},\{b\})$ for $a\in P$ and
$b\in N$, then
$\bigvee_{a\in P}\bigwedge_{b\in
	N}\varphi_{a,b}$ strongly separates $(\Kmc,P,N)$.
Without loss of generality, we may thus work with labeled KBs with
singleton sets of positive and negative examples.

Each choice of an ontology language $\Lmc$ and a separation language
$\Lmc_{S}$ thus gives rise to a (single) strong separability problem
that we refer to as \emph{strong $(\Lmc,\Lmc_S)$-separability},
defined in the expected way:

\begin{center}
	\fbox{\begin{tabular}{@{\,}l@{\;}l@{\,}}
			\small{PROBLEM} : & Strong $(\Lmc,\Lmc_{S})$ separability with signature\\
			{\small INPUT} : &  Labeled $\Lmc$-KB $(\Kmc,P,N)$ and \\
			       & signature $\Sigma \subseteq \text{sig}(\Kmc)$ \\
			{\small QUESTION} : & Is $(\Kmc,P,N)$ strongly $\LmcO_{S}(\Sigma)$-separable?  
	\end{tabular}}
\end{center}
If $\Lmc=\Lmc_{S}$, then we simply speak of strong $\Lmc$-separability. The study of strong separability is very closely linked to the study of interpolants and the Craig interpolation property (CIP). Given FO-formulas $\varphi(x),\psi(x)$ and a fragment $\Lmc$ of FO, we say that an $\Lmc$-formula $\chi(x)$ is an \emph{$\Lmc$-interpolant of $\varphi,\psi$} if $\varphi(x)\models
\chi(x)$, $\chi(x) \models \psi(x)$, and $\text{sig}(\chi)
\subseteq \text{sig}(\varphi)\cap \text{sig}(\psi)$.
We say that \emph{$\Lmc$ has the CIP} if for all $\Lmc$-formulas $\varphi(x),\psi(x)$ such that $\varphi(x)\models \psi(x)$, there exists an $\Lmc$-interpolant of $\varphi,\psi$. FO and many of its fragments have the CIP~\cite{craig,TenEtAl13,GabMaks}.
The link between interpolants and strongly separating formulas is easy to see: assume a labeled $\Lmc$-KB $(\Kmc,\{ a\},\{ b\})$ and a signature $\Sigma \subseteq
\text{sig}(\Kmc)$ are given. Obtain $\Kmc_{\Sigma,a}$ (and $\Kmc_{\Sigma,b}$) from $\Kmc$ 
by taking the standard translation of $\Kmc$ into FO and then
\begin{itemize}
	
	\item replacing all concept and role names $X\not\in \Sigma$ by
	fresh symbols $X_{a}$ ($X_{b}$, respectively);
		
	\item replacing all individual names $c\not\in \Sigma \cup \{a\}$  
	by fresh variables $x_{c}$ (all $c\not\in \Sigma \cup \{b\}$
	by fresh variables $y_{c}$, respectively);
	
	\item replacing $a$ by $x$ (and $b$ by $x$, respectively) for a single fresh variable $x$;
	
	\item adding $x=a$ if $a\in \Sigma$ ($x=b$ if $b\in \Sigma$,
	  respectively).
%
%
\end{itemize}	Let $\varphi_{\Kmc,\Sigma,a}(x) = \exists \vec{z}
(\bigwedge \Kmc_{\Sigma,a})$, where $\vec z$ is the sequence of free
variables in $\Kmc_{\Sigma,a}$ without the variable $x$ and
$(\bigwedge \Kmc_{\Sigma,a})$ is the conjunction of all formulas in
$\Kmc_{\Sigma,a}$. $\varphi_{\Kmc,\Sigma,b}(x)$ is defined in the same
way, with $a$ replaced by $b$. The following lemma is a direct
consequence of the construction.
\begin{lemma}\label{lem:int} Let $(\Kmc,\{a\},\{b\})$ be a labeled
  $\Lmc$-KB, $\Sigma\subseteq \text{sig}(\Kmc)$ a signature, and
  $\Lmc_{S}$ a fragment of FO. Then the following conditions are
  equivalent for any formula $\varphi(x)$ in $\Lmc_{S}$:
\begin{enumerate} \item $\varphi$ strongly $\Sigma$-separates
    $(\Kmc,\{a\},\{b\})$; \item $\varphi$ is an $\Lmc_{S}$-interpolant
      for $\varphi_{\Kmc,\Sigma,a}(x),\neg
      \varphi_{\Kmc,\Sigma,b}(x)$.  \end{enumerate}	\end{lemma}
\begin{example}\label{exm:strong1}
  To illustrate Lemma~\ref{lem:int}, 
  let $\Sigma=\{r\}$ and $\Kmc=(\Omc,\Dmc)$, with $\Omc= \{A \sqsubseteq \forall r.\neg
  A\}$ and $\Dmc=\{A(a),r(b,b)\}$.  
  Then, $\neg r(x,x)$
  strongly $(\Sigma)$-separates $(\Kmc,\{a\},\{b\})$ and 
  is an interpolant for $\varphi_{\Kmc,\Sigma,a}$, $\neg
  \varphi_{\Kmc,\Sigma,b}$ where $\varphi_{\Kmc,\Sigma,a}$,
  $\varphi_{\Kmc,\Sigma,b}$ are the following two formulas:
  \begin{align*}
    \exists x_{b}\, r(x_{b},x_{b})\wedge A_a(x)\wedge \forall yz(r(y,z) \wedge
    A_a(y)\to \neg A_a(z)),\\
    \exists y_a\, A_b(y_a)\wedge r(x,x)\wedge
  \forall yz\,(r(y,z) \wedge A_b(y)\to \neg A_b(z)).
  \end{align*}
\end{example}

Thus, the problem whether a labeled KB $(\Kmc,P,N)$ is strongly $\Lmc_{S}(\Sigma)$-separable and the computation of a strongly $\Sigma$-separating formula can be equivalently formulated as an interpolant existence problem. As FO has the CIP, we obtain the following characterization and complexity result for the existence of strongly FO$(\Sigma)$-separating formulas.
\begin{restatable}{theorem}{thmfostrongsep}
	\label{thm:fostrongsep}
	Let $\Lmc\in \DLS$. The following conditions are equivalent for any $\Lmc$-KB $(\Kmc,\{a\},\{b\})$ and signature $\Sigma\subseteq \text{sig}(\Kmc)$:
	\begin{enumerate}
		\item $(\Kmc,\{a\},\{b\})$ is strongly FO$(\Sigma)$-separable;
		\item $\varphi_{\Kmc,\Sigma,a}(x) \models \neg \varphi_{\Kmc,\Sigma,b}(x)$.
	\end{enumerate}
Strong $(\Lmc,\text{FO})$-separability with signature is \ExpTime-complete.
\end{restatable}
The \ExpTime upper bound follows from the fact that the 
complement of the problem to decide $\varphi_{\Kmc,\Sigma,a}(x) \models \neg \varphi_{\Kmc,\Sigma,b}(x)$ can be equivalently formulated as a concept
satisfiability problem in the extension $\mathcal{ALCIO}^{u}$ of $\mathcal{ALCIO}$ with the universal 
role $u$. The lower bound can be proved by reduction of $\ALC$-KB satisfiability.

It follows from Theorem~\ref{thm:fostrongsep} that one can use FO
theorem provers such as Vampire to compute strongly separating
formulas~\cite{DBLP:conf/aplas/HoderHKV12}. FO is arguably too
powerful, however, to serve as a useful separation language for
labeled description logic KBs. Thus, two important questions arise:
(1) which fragment of FO is needed to obtain a strongly separating
formula in case that there is a strongly separating formula in FO? (2)
What happens if the languages $\Lmc\in \DLS$ are used as separation
languages? For~(1), one can show that none of the
languages in $\text{UCQ}^{\Lmc}$, $\Lmc\in \DLS$, is sufficient
to separate $a$ and $b$ in Example~\ref{exm:strong1}.
We next show that the need for the negation of a CQ in that
example is no accident. Indeed, by taking the closure BoCQ$^{\mathcal{ALCIO}}(\Sigma)$ of 
CQ$^{\mathcal{ALCIO}}(\Sigma)$ under negation, conjunction, and disjunction
one obtains a sufficiently powerful
language for (1), at least if the KB does not admit nominals. 
\begin{restatable}{theorem}{thmstrongFOequivalent}
	\label{thm:strongFOequivalent}
	The following conditions are equivalent for any labeled $\mathcal{ALCI}$-KB
	$(\Kmc,P,N)$ and signature $\Sigma \subseteq \text{sig}(\Kmc)$.
	\begin{enumerate}
		\item $(\Kmc,P,N)$ is strongly FO$(\Sigma)$-separable;
		\item $(\Kmc,P,N)$ is strongly BoCQ$^{\mathcal{ALCIO}}(\Sigma)$-separable.
	\end{enumerate}
\end{restatable}
The proof of Theorem~\ref{thm:strongFOequivalent} uses the model-theoretic characterization given in Lemma~\ref{lem:equivalence2} and techniques introduced in \cite{DBLP:journals/corr/SegoufinC13}. Problem~(2) can be comprehensively 
solved by using recent results about the complexity of deciding
the existence of interpolants in DLs with nominals~\cite{AJMOW-AAAI21}. Rather
surprisingly,
strong separability becomes one exponential harder than for FO. While the upper bounds are direct consequences of the results in~\cite{AJMOW-AAAI21}, for the lower bounds
one has to adapt the proofs.
\begin{restatable}{theorem}{thmtwoexpdl}
	\label{thm:thmtwoexpdl}
	Let $\Lmc\in \DLS$. Then strong $\Lmc$-separability with signature is \TwoExpTime-complete.
\end{restatable}

\section{Separability with Signature in GF and FO$^{2}$}
\label{sec:GF}
In the guarded fragment, GF, of FO quantification takes the form
$$ 
\forall
\vec{y}(\alpha(\vec{x},\vec{y})\rightarrow \varphi(\vec{x},\vec{y}))
\text{ and } \exists \vec{y}(\alpha(\vec{x},\vec{y})\wedge
\varphi(\vec{x},\vec{y})) $$ 
where $\alpha(\vec{x},\vec{y})$ is an atomic formula or an equality
$x=y$ that contains all variables in
$\vec{x},\vec{y}$~\cite{ANvB98,DBLP:journals/tocl/HernichLPW20}. The
two-variable fragment, FO$^{2}$, is the fragment of FO with only two
individual variables. For GF we admit relation symbols of arbitrary
arity and equality, but no constant symbols. For FO$^{2}$ we make the
same assumptions except that we admit relation symbols of arity one
and two only. The definitions of weak projective and non-projective
separability and of strong separability are the obvious extensions of
the definitions given for description logics. Our results do not
depend on whether one admits examples that are sets of tuples of
constants of fixed but arbitary length or still only considers sets of constants.

Weak FO$^{2}$-separability is undecidable already with full relational signature, in both the projective and the non-projective case~\cite{KR}. For GF the situation is different: in both cases
weak GF-separability is \TwoExpTime-complete, thus not harder than
satisfiability. This result does not generalize to restricted
signatures. In fact, by adapting the undecidability proof for
conservative extensions given in~\cite{JLMSW17}, one can show the
following.
%
\begin{restatable}{theorem}{thmmaingf}
\label{thm:main-gf}
	Projective and non-projective $(\Lmc,\Lmc_S)$-separability with signature are
	undecidable for all $(\Lmc,\Lmc_S)$ such that $\Lmc$ contains
	GF$^3$ and $\Lmc_S$ contains \ALC.
\end{restatable}
We now consider strong separability. For both FO$^{2}$ and GF the
complexity of deciding strong separability with full relational signature is the same as
validity, thus {\sc coNExpTime}-complete and, respectively, \TwoExpTime-complete~\cite{KR}. With restricted signatures, the situation is different, and can again be analyzed in terms of interpolant existence. The formula $\varphi_{\Kmc,\Sigma,a}(x)$ constructed in Section~\ref{sec:dfstrong} 
is not guaranteed to be in GF or FO$^{2}$ even if $\Kmc$ is a GF or, respectively, FO$^{2}$-KB. It is, however, straightforward to construct formulas in the respective fragments that can serve the same purpose (either by using constants or by introducing a fresh relation symbol as a guard for $\Dmc$ (for GF) and re-using variables (for FO$^{2}$)). Thus, strong separability in GF and FO$^{2}$-KBs is
again equivalent to interpolant existence. Points~1 and 2 of the following theorem then follow from the CIP of FO and the complexity of GF and FO$^{2}$~\cite{DBLP:journals/jsyml/Gradel99,DBLP:journals/bsl/GradelKV97}. Neither FO$^{2}$ nor GF have the CIP~\cite{comer1969,Pigozzi71,DBLP:journals/sLogica/HooglandM02}, thus separating in FO$^{2}$ and GF is less powerful than separating using FO. The complexity of interpolant existence for GF and FO$^{2}$ has recently been studied in~\cite{jung2020living} and the upper bounds in Points~3 and 4 follow directly from the complexity upper bounds for interpolant existence. The lower bounds are obtained by adapting the proofs.
\begin{restatable}{theorem}{thmgffo}
	\label{thm:gffotwo}
	\begin{enumerate}
		\item Strong $(\text{GF,FO})$-separability with signature is \TwoExpTime-complete;
		\item Strong $(\text{FO}^{2},\text{FO})$-separability with signature is {\sc coNExpTime}-complete;
		\item Strong GF-separability is {\sc 3ExpTime}-complete, for relational signatures;
		\item Strong FO$^{2}$-separability with signature is in {\sc coN2ExpTime} and \TwoExpTime-hard, for relational signatures.
	\end{enumerate}
\end{restatable}

\section{Discussion}
We have started investigating separability of data examples under
signature restrictions. Our main contributions are an analysis of the
separating power of several important languages and the computational
complexity of deciding separability. 
The following table gives an overview of the complexity of separability for expressive fragments of FO with and without signature restrictions. For Horn-DLs we refer the reader
	to~\cite{DBLP:conf/ijcai/FunkJLPW19,aaaithis}. The results in the gray
	columns (weak, projective, with signature restriction and strong
	with signature restriction, respectively) are shown in this article, the
	results of the first (weak, projective, and full signature), second
	(weak, non-projective, and full signature), and fourth (strong and
	full signature) column are shown
	in~\cite{DBLP:conf/ijcai/FunkJLPW19,KR}. 
	
	\setlength{\tabcolsep}{4pt}
	\newcolumntype{a}{>{\columncolor{Gray}}c|}
	\begin{center}
		\begin{tabular}{ |c | c c a c a}
			\hline
			& \multicolumn{3}{c|}{Weak Separability} &
			\multicolumn{2}{c|}{Strong Separability} \\
			$\mathcal{L}$ & prj+full & full & prj+rstr & full & rstr \\
			\hline                     
			$\mathcal{ALC}$ & \textsc{NExp} & ? & \textsc{2Exp} & \textsc{Exp} & \textsc{2Exp}\\
			$\mathcal{ALCI}$ & \textsc{NExp} & \textsc{NExp} & \textsc{2Exp} & \textsc{Exp} & \textsc{2Exp}\\
			$\mathcal{ALCO}$ &   ?        & ? & \textsc{3Exp} & ? & \textsc{2Exp}\\
			GF & \textsc{2Exp} & \textsc{2Exp} & Undec & \textsc{2Exp} & \textsc{3Exp} \\
			FO$^2$ & Undec & Undec & Undec & \textsc{NExp} &
			\makecell{$\leq$\textsc{coN}2\textsc{Exp} \\ $\geq$\textsc{2Exp}} \\
			\hline
		\end{tabular}
	\end{center}
	The missing entries for $\mathcal{ALCO}$ are due to the fact that
	nominals are considered for the first time in this article in the
	context of separability.  We conjecture that the complexity is the
	same as for $\mathcal{ALC}$; note, however, that one has to be careful
	when defining separability problems in \ALCO under the full signature as the
	individuals that provide positive and negative counterexamples should
	be disallowed from separating concepts.
	

Further interesting theoretical problems include: what is the complexity
of weak projective separability with signature for $\mathcal{ALCIO}$,
where the bisimulation characterization given in
Theorem~\ref{thm:L-modeltheory0} does not hold? What is the complexity
of non-projective weak separability with signature 
(and conservative extensions) 
for the DLs in $\DLS$? 
From a practical viewpoint, it would be of interest to investigate
systematically the size of separating concepts and to develop
algorithms for computing them, if they exist. Recall that such an algorithm is
already provided (by the relation of separating formulas to Craig interpolants) in the case of strong separability and it would be of interest to evaluate
empirically the shape and size of Craig interpolants in FO in that
case.

\section*{Acknowledgements}
Carsten Lutz was supported by DFG CRC 1320 Ease.
Frank Wolter was supported by EPSRC grant EP/S032207/1.


\cleardoublepage

\appendix

\section{Further Preliminaries}
We remind the reader of
different kinds of bisimulations that characterize the expressive
power of the languages in $\DLS$~\cite{TBoxpaper,goranko20075}.  Let
$\Amf$ and $\Bmf$ be structures and $\Sigma$ a signature. A relation
$S \subseteq \text{dom}(\Amf) \times \text{dom}(\Bmf)$ is an
\emph{$\mathcal{ALCO}(\Sigma)$-bisimulation between $\Amf$ and $\Bmf$}
if the following conditions hold:
\begin{enumerate}
	\item if $(d,e)\in S$ and $A\in \Sigma$, then
	$d \in A^{\Amf}$ iff $e\in A^{\Bmf}$; 
	\item if $(d,e)\in S$ and $c\in \Sigma$, then
	$d=c^{\Amf}$ iff $e=c^{\Bmf}$;
	\item if $(d,e)\in S$, $r\in \Sigma$, and $(d,d')\in r^{\Amf}$, 
	then there is an $e'$ with $(e,e')\in r^{\Bmf}$ and
	$(d',e')\in S$;
	\item if $(d,e)\in S$, $r\in \Sigma$, $(e,e')\in r^{\Bmf}$, 
	then there is a $d'$ with $(d,d')\in r^{\Amf}$ and
	$(d',e')\in S$,
\end{enumerate}	
$S$ is an \emph{$\mathcal{ALCIO}(\Sigma)$-bisimulation between $\Amf$ and $\Bmf$}  
if Points~3 and 4 also hold for inverse roles over $\Sigma$. If $\Sigma$ is relational then we speak about $\mathcal{ALC}(\Sigma)$ and $\mathcal{ALCI}(\Sigma)$-bisimulations, respectively. 

Let $\Sigma$ be a signature. A \emph{$\Sigma$-homomorphism} $h$ from a
structure $\Amf$ to a structure $\Bmf$ is a function
$h:\text{dom}(\Amf)
\rightarrow \text{dom}(\Bmf)$ such that $a\in A^{\Amf}$ implies $h(a)\in A^{\Bmf}$
for all $A\in \Sigma$, $(a,b)\in r^{\Amf}$ implies $(h(a),h(b))\in r^{\Bmf}$
for all $r\in \Sigma$, and $h(c^{\Amf})=c^{\Bmf}$ for all $c\in \Sigma$.
\section{Proofs for Section~\ref{sec:prelim}}
\lemequivalencetwo*
\begin{proof} \
	Assume first that $\Amf,a \Rightarrow^{\text{mod}}_{\text{CQ}_{r}^{\Lmc},\Sigma} \Bmf,b$
	and let $\varphi(x)$ be a formula in $\text{CQ}_{r}^{\Lmc}(\Sigma)$
	such that $\Amf\models \varphi(a)$. Then there exists a mapping $h$ from the 
	set $\text{var}(\varphi)$ of variables in $\varphi(x)$ to $\Amf$ such that $h(x)=a$ and 
	\begin{itemize}
		\item If $r(y,z)$ is a conjunct of $\varphi(x)$, then $(h(y),h(z))\in r^{\Amf}$;
		\item If $C(y)$ is a conjunct of $\varphi(x)$, then $h(y)\in C^{\Amf}$.
	\end{itemize}
	Let $D$ be the image of $\text{var}(\varphi)$ under $h$. Then the $\Sigma$-reduct of $\Amf_{|D}$ is $\Lmc$-rooted in $a$ and, by definition of
	$\Amf,a \Rightarrow^{\text{mod}}_{\text{CQ}^{\Lmc},\Sigma} \Bmf,b$,
	we have a $\Sigma$-homomorphism $h'$ from $\Amf_{|D}$ to $\Bmf$ such that $h'(a)=b$ and
	$\Amf,c \sim_{\Lmc,\Sigma} \Bmf,h'(c)$ for all $c\in D$. Take the composition
	$h'\circ h$ and observe that by Lemma~\ref{lem:equivalence}, $h'\circ h(y) \in C^{\Bmf}$ if 
	$C(y)$ is a conjunct of $\varphi$. Thus, $\Bmf \models \varphi(b)$, as required. The proof for $\text{CQ}^{\Lmc}$ is the same except that the one
	does not need to observe that the $\Sigma$-reduct of $\Amf_{|D}$ is $\Lmc$-rooted in $a$.
	
	Conversely, assume that $\Amf,a \Rightarrow_{\text{CQ}_{r}^{\Lmc},\Sigma} \Bmf,b$. To show that $\Amf,a \Rightarrow^{\text{mod}}_{\text{CQ}_{r}^{\Lmc},\Sigma} \Bmf,b$,
	let $D$ be such that the $\Sigma$-reduct of $\Amf_{|D}$ is $\Lmc$-rooted at $a$. 
	Consider the set for formulas $q_{D}^{\Amf}$ that is obtained by regarding the nodes $d$ in $D$ as variables $x_{d}$ and taking $(x_{d_{1}},x_{d_{2}})$ if $(d_{1},d_{2})\in r^{\Amf}$, $r\in \Sigma$, and $C(x_{d})$ if $d\in C^{\Amf}$ for $C\in \Lmc(\Sigma)$. If follows from $\Amf,a \Rightarrow_{\text{CQ}_{r}^{\Lmc},\Sigma} \Bmf,b$ that every finite subset of
	$q_{D}^{\Amf}$ is satisfied in $\Bmf$ under an assignment mapping $x_{a}$ to $b$.
	If $\Bmf$ is $\omega$-saturated, then $q_{D}^{\Amf}$ is satisfied in $\Bmf$ by definition of $\omega$-saturatedness (and also holds if the $\Sigma$-reduct of $\Amf_{|D}$ is not rooted in $a$). If $\Bmf$ has finite outdegree then this can be shown directly using the condition that the $\Sigma$-reduct of $\Amf_{|D}$ is rooted in $a$. Let $v$ be the satisfying assignment. Then $h:D \rightarrow \Bmf$ defined by setting $h(d) = v(x_{d})$ is a
	$\Sigma$-homomorphism, $h(a)=b$, and $\Amf,c \sim_{\Lmc,\Sigma} \Bmf,h(c)$ for all $c\in D$, as required. The implication for $\text{CQ}^{\Lmc}$ follows using the comment above.
\end{proof}

We slightly extend Lemma~\ref{lem:forestm} as required later. The proof is by a standard selective unraveling procedure.

\begin{restatable}{lemma}{lemforestmodelcompl}
	\label{lem:forestmodelcompleteness}
	Let $\Lmc\in \{\ALC,\ALCI,\mathcal{ALCO}\}$	and let $\Kmc$ be an $\Lmc$-KB and $C$ an $\Lmc$-concept.
	If $\Kmc\not\models C(a)$, then there exists an $\Lmc$-forest model 
	$\Amf$ of $\Kmc$ of finite $\Lmc$-outdegree with $a\not\in C^{\Amf}$.
	
	For every model $\Amf$ of $\Kmc$ there exists an $\Lmc$-forest model $\Amf'$ of $\Kmc$ such that $\Amf',a^{\Amf'} \sim_{\Lmc}^{f} \Amf,a^{\Amf}$. If $\Amf$ is finite, then there exists such a model of finite $\Lmc$-outdegree.
\end{restatable}
We now show that Lemma~\ref{lem:forestm} does not hold for $\mathcal{ALCIO}$.
Note that an $\mathcal{ALCIO}$-forest model of $\Kmc$ is a model $\Amf$ of $\Kmc$
such that $\Amf$ with all $R(a,b^{\Amf})$, $R$ a role, $a\in
\text{dom}(\Amf)$, $b\in \text{ind}(\Dmc)$ is a disjoint union of
$\mathcal{ALCI}$-trees rooted at $a$, $a\in \text{ind}(\Dmc)$. Then the concept
$$
\{a\} \sqcap A \sqcap \exists s.\top \sqcap \forall s.(\neg A \sqcap \exists r.\exists s^{-}.\{a\})
$$
is satisfiable in a model of $\Kmc=(\emptyset,\{A(a)\})$, but not in any $\mathcal{ALCIO}$-forest model of $\Kmc$ of finite $\mathcal{ALCIO}$-outdegree.

\section{Proofs for Section~\ref{sec:weaksep}}

\thmnominalhelper*

The
proof is by a reduction of the following undecidable tiling problem.
%
\begin{definition}
	A \emph{tiling system} $S=(\mathcal{T},H,V,R,L,T,B)$ consists of a finite set
	$\Tmc$ of \emph{tiles}, horizontal and vertical \emph{matching
		relations} $H,V \subseteq \Tmc \times \Tmc$, and sets $R,L,T,B \subseteq
	\Tmc$ of \emph{right} tiles, \emph{left} tiles, \emph{top} tiles, and
	\emph{bottom} tiles. A \emph{solution} to $S$ is a triple $(n,m,\tau)$
	where $n,m \geq 1$ and $\tau: \{0,\ldots,n\} \times \{0,\ldots,m\} \rightarrow \Tmc$ such
	that the following hold:
	\begin{enumerate}
		
		\item $(\tau(i,j),\tau(i+1,j)) \in H$, for all $i<n$ and $j \leq m$;
		
		\item $(\tau(i,j),\tau(i,j+1)) \in V$, for all $i\leq n$ and $j<m$;
		
		\item $\tau(0,j) \in L$ and $\tau(n,j) \in R$, for all $0 \leq j \leq m$;
		
		\item $\tau(i,0) \in B$ and $\tau(i,m) \in T$, for all $0 \leq i \leq n$.
		
	\end{enumerate}
	\vspace*{-\medskipamount}
\end{definition}
We show how to convert a tiling system $S$ into a labeled
$\mathcal{ALC}$-KB $(\Kmc,P,N)$ and signature $\Sigma$ such that
$S$ has a solution iff $(\Kmc,P,N)$ is projectively
$\mathcal{ALCO}(\Sigma)$-separable with individual names as additional helper 
symbols.

Let $S=(\Tmc,H,V,R,L,T,B)$ be a tiling system. Define an ontology \Omc containing the following inclusions.
\begin{itemize}
	
	%
	\item Every grid node is labeled with exactly
	one tile and the matching conditions are satisfied:
	$$
	\begin{array}{rcl}
		\top &\sqsubseteq& \bigsqcup_{t\in T}(t \sqcap \bigsqcap_{t' \in T,\; t'\not=t} \neg t') \\[4mm]
		\top &\sqsubseteq& \bigsqcap_{t\in T}(t \rightarrow (\bigsqcup_{(t,t') \in H} \forall r_x . t' \sqcap \bigsqcup_{(t,t') \in V} \forall r_y . t'))
	\end{array}
	$$
	
	\item The concepts \mn{left}, \mn{right}, \mn{top}, \mn{bottom} mark the borders of
	the grid in the expected way: 
	\begin{eqnarray*}	
		\mn{bottom} & \sqsubseteq & \neg \mn{top} \sqcap \forall r_{x}. \mn{bottom} \\	
		\mn{right} & \sqsubseteq & \forall r_{y}. \mn{right}\\
		\mn{left} & \sqsubseteq & \neg \mn{right} \sqcap \forall r_{y}.\mn{left}\\
		\mn{top} & \sqsubseteq & \forall r_{x}.\mn{top}\\
		\neg \mn{top} & \equiv & \exists r_{y}.\top \\
        \neg \mn{right} & \equiv & \exists r_{x}.\top	
\end{eqnarray*}
	and $\mn{bottom} \sqsubseteq \bigsqcup_{t\in B}t$, $\mn{right} \sqsubseteq \bigsqcup_{t\in R}t$, $\mn{left} \sqsubseteq \bigsqcup_{t\in L}t$, $\mn{top} \sqsubseteq \bigsqcup_{t\in T}t$.

	%
	\item There is an infinite outgoing $r_x$/$r_y$-path starting at
	$Q$ or some grid cell does not close in the part of models 
	reachable from $Q$:
	$$	
	Q \sqsubseteq \exists r_{x}. Q \sqcup \exists r_{y}.Q \sqcup
	(\exists r_{x}.\exists r_{y}.P \sqcap \exists r_{y}. \exists r_{x}.\neg P) 
	$$
	\item $Q$ is triggered by $A_{1}\sqcap D$:
	$$
	A_1 \sqcap D \sqsubseteq Q
	$$
\end{itemize}
Now let $\Dmc = \{A_{1}(a),Y(b), D(o), \mn{left}(o),\mn{bottom}(o)\}$, set
$$
\Sigma = \{ o,r_x,r_y, \mn{left}, \mn{right}, \mn{top}, \mn{bottom}\}
$$
and consider the labeled KB $(\Kmc,\{a\},\{b\})$ where $\Kmc=(\Omc,\Dmc)$.

\begin{lemma}\label{lemfisrtone}
	If  $S$ has a solution, then there is an
	$\mathcal{ALCO}(\Sigma\cup \Sigma_{\text{help}})$-concept that separates $(\Kmc,\{a\},\{b\})$, where $\Sigma_{\text{help}}$ is a set of fresh individual names.
\end{lemma}
\begin{proof} \
	Assume that $S$ has a solution consisting of a properly tiled $n\times m$ grid. We design an $\mathcal{ALCO}(\Sigma\cup \Sigma_{\text{help}})$-concept $G$ so that any model of $G$ and $\Kmc$ includes a properly tiled $n \times m$-grid with lower
	left corner $o$, where $\Sigma_{\text{help}}$ is a set of fresh individual names. The individual names in $\Sigma_{\text{help}}$ are $a_{i,j}$, $0 \leq i \leq n$, $0\leq j \leq m$. Then let $G$ be the $\mathcal{ALCO}(\Sigma\cup \Sigma_{\text{help}})$-concept that states that $G$ is true at the bottom left corner of a $r_{x}/r_{y}$ grid in which the nodes are given by the 
	interpretation of the individual names $a_{i,j}$ and such that
\begin{itemize}
	\item $a_{i+1,j}$ is the only $r_{x}$-successor of $a_{i,j}$;
	\item $a_{i,j+1}$ is the only $r_{y}$-successor of $a_{i,j}$;
	\item the borders of the grid satisfy the respective concepts $\mn{left}, \mn{right}$, $\mn{top}$, $\mn{bottom}$;
	\item $o=a_{0,0}$.
\end{itemize}
Thus, the grid can be depicted as follows:

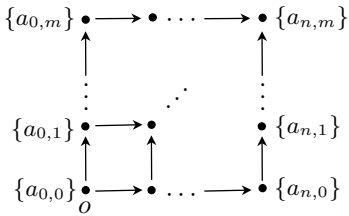
\begin{center}
	\tikzset{every picture/.style={line width=0.5pt}} 

\begin{tikzpicture}[x=0.75pt,y=0.75pt,yscale=-1,xscale=1]
	
	\draw  [fill={rgb, 255:red, 0; green, 0; blue, 0 }  ,fill opacity=1 ] (262.1,270.25) .. controls (262.1,269.28) and (261.32,268.5) .. (260.35,268.5) .. controls (259.38,268.5) and (258.6,269.28) .. (258.6,270.25) .. controls (258.6,271.22) and (259.38,272) .. (260.35,272) .. controls (261.32,272) and (262.1,271.22) .. (262.1,270.25) -- cycle ;
	\draw  [fill={rgb, 255:red, 0; green, 0; blue, 0 }  ,fill opacity=1 ] (261.9,238.05) .. controls (261.9,237.08) and (261.12,236.3) .. (260.15,236.3) .. controls (259.18,236.3) and (258.4,237.08) .. (258.4,238.05) .. controls (258.4,239.02) and (259.18,239.8) .. (260.15,239.8) .. controls (261.12,239.8) and (261.9,239.02) .. (261.9,238.05) -- cycle ;
	\draw    (260.38,265.52) -- (260.28,246.05) ;
	\draw [shift={(260.27,243.05)}, rotate = 449.7] [fill={rgb, 255:red, 0; green, 0; blue, 0 }  ][line width=0.08]  [draw opacity=0] (5,-2.5) -- (0,0) -- (5,2.5) -- (3.5,0) -- cycle    ;
	\draw    (265.18,270.72) -- (286,270.27) ;
	\draw [shift={(289,270.2)}, rotate = 538.76] [fill={rgb, 255:red, 0; green, 0; blue, 0 }  ][line width=0.08]  [draw opacity=0] (5,-2.5) -- (0,0) -- (5,2.5) -- (3.5,0) -- cycle    ;
	\draw  [fill={rgb, 255:red, 0; green, 0; blue, 0 }  ,fill opacity=1 ] (295.3,270.25) .. controls (295.3,269.28) and (294.52,268.5) .. (293.55,268.5) .. controls (292.58,268.5) and (291.8,269.28) .. (291.8,270.25) .. controls (291.8,271.22) and (292.58,272) .. (293.55,272) .. controls (294.52,272) and (295.3,271.22) .. (295.3,270.25) -- cycle ;
	\draw    (265.18,237.72) -- (286,237.27) ;
	\draw [shift={(289,237.2)}, rotate = 538.76] [fill={rgb, 255:red, 0; green, 0; blue, 0 }  ][line width=0.08]  [draw opacity=0] (5,-2.5) -- (0,0) -- (5,2.5) -- (3.5,0) -- cycle    ;
	\draw  [fill={rgb, 255:red, 0; green, 0; blue, 0 }  ,fill opacity=1 ] (295.3,237.25) .. controls (295.3,236.28) and (294.52,235.5) .. (293.55,235.5) .. controls (292.58,235.5) and (291.8,236.28) .. (291.8,237.25) .. controls (291.8,238.22) and (292.58,239) .. (293.55,239) .. controls (294.52,239) and (295.3,238.22) .. (295.3,237.25) -- cycle ;
	\draw  [fill={rgb, 255:red, 0; green, 0; blue, 0 }  ,fill opacity=1 ] (351.3,269.25) .. controls (351.3,268.28) and (350.52,267.5) .. (349.55,267.5) .. controls (348.58,267.5) and (347.8,268.28) .. (347.8,269.25) .. controls (347.8,270.22) and (348.58,271) .. (349.55,271) .. controls (350.52,271) and (351.3,270.22) .. (351.3,269.25) -- cycle ;
	\draw    (293.58,265.05) -- (293.48,245.58) ;
	\draw [shift={(293.47,242.58)}, rotate = 449.7] [fill={rgb, 255:red, 0; green, 0; blue, 0 }  ][line width=0.08]  [draw opacity=0] (5,-2.5) -- (0,0) -- (5,2.5) -- (3.5,0) -- cycle    ;
	\draw  [fill={rgb, 255:red, 0; green, 0; blue, 0 }  ,fill opacity=1 ] (260.36,182.44) .. controls (259.39,182.43) and (258.59,183.19) .. (258.58,184.16) .. controls (258.56,185.13) and (259.33,185.92) .. (260.3,185.94) .. controls (261.26,185.96) and (262.06,185.19) .. (262.08,184.22) .. controls (262.09,183.26) and (261.32,182.46) .. (260.36,182.44) -- cycle ;
	\draw  [fill={rgb, 255:red, 0; green, 0; blue, 0 }  ,fill opacity=1 ] (351.3,183.25) .. controls (351.3,182.28) and (350.52,181.5) .. (349.55,181.5) .. controls (348.58,181.5) and (347.8,182.28) .. (347.8,183.25) .. controls (347.8,184.22) and (348.58,185) .. (349.55,185) .. controls (350.52,185) and (351.3,184.22) .. (351.3,183.25) -- cycle ;
	\draw    (265.18,183.72) -- (285.35,183.56) ;
	\draw [shift={(288.35,183.53)}, rotate = 539.55] [fill={rgb, 255:red, 0; green, 0; blue, 0 }  ][line width=0.08]  [draw opacity=0] (5,-2.5) -- (0,0) -- (5,2.5) -- (3.5,0) -- cycle    ;
	\draw  [fill={rgb, 255:red, 0; green, 0; blue, 0 }  ,fill opacity=1 ] (350.9,238.05) .. controls (350.9,237.08) and (350.12,236.3) .. (349.15,236.3) .. controls (348.18,236.3) and (347.4,237.08) .. (347.4,238.05) .. controls (347.4,239.02) and (348.18,239.8) .. (349.15,239.8) .. controls (350.12,239.8) and (350.9,239.02) .. (350.9,238.05) -- cycle ;
	\draw    (349.38,265.52) -- (349.28,246.05) ;
	\draw [shift={(349.27,243.05)}, rotate = 449.7] [fill={rgb, 255:red, 0; green, 0; blue, 0 }  ][line width=0.08]  [draw opacity=0] (5,-2.5) -- (0,0) -- (5,2.5) -- (3.5,0) -- cycle    ;
	\draw  [fill={rgb, 255:red, 0; green, 0; blue, 0 }  ,fill opacity=1 ] (295.9,183.05) .. controls (295.9,182.08) and (295.12,181.3) .. (294.15,181.3) .. controls (293.18,181.3) and (292.4,182.08) .. (292.4,183.05) .. controls (292.4,184.02) and (293.18,184.8) .. (294.15,184.8) .. controls (295.12,184.8) and (295.9,184.02) .. (295.9,183.05) -- cycle ;
	\draw    (260.38,212.52) -- (260.28,193.05) ;
	\draw [shift={(260.27,190.05)}, rotate = 449.7] [fill={rgb, 255:red, 0; green, 0; blue, 0 }  ][line width=0.08]  [draw opacity=0] (5,-2.5) -- (0,0) -- (5,2.5) -- (3.5,0) -- cycle    ;
	\draw    (349.38,212.52) -- (349.28,193.05) ;
	\draw [shift={(349.27,190.05)}, rotate = 449.7] [fill={rgb, 255:red, 0; green, 0; blue, 0 }  ][line width=0.08]  [draw opacity=0] (5,-2.5) -- (0,0) -- (5,2.5) -- (3.5,0) -- cycle    ;
	\draw    (320.18,270.72) -- (341,270.27) ;
	\draw [shift={(344,270.2)}, rotate = 538.76] [fill={rgb, 255:red, 0; green, 0; blue, 0 }  ][line width=0.08]  [draw opacity=0] (5,-2.5) -- (0,0) -- (5,2.5) -- (3.5,0) -- cycle    ;
	\draw    (320.18,183.72) -- (340.35,183.56) ;
	\draw [shift={(343.35,183.53)}, rotate = 539.55] [fill={rgb, 255:red, 0; green, 0; blue, 0 }  ][line width=0.08]  [draw opacity=0] (5,-2.5) -- (0,0) -- (5,2.5) -- (3.5,0) -- cycle    ;
	
	\draw (262.67,214) node [anchor=north west][inner sep=0.75pt]  [rotate=-90] [align=left] {$\displaystyle \dotsc $};
	\draw (314.42,218.9) node [anchor=north west][inner sep=0.75pt]  [rotate=-136] [align=left] {$\displaystyle \dotsc $};
	\draw (317.67,274.2) node [anchor=north west][inner sep=0.75pt]  [rotate=-180] [align=left] {$\displaystyle \dotsc $};
	\draw (317.67,186.2) node [anchor=north west][inner sep=0.75pt]  [rotate=-180] [align=left] {$\displaystyle \dotsc $};
	\draw (351.67,214) node [anchor=north west][inner sep=0.75pt]  [rotate=-90] [align=left] {$\displaystyle \dotsc $};
	\draw (222.47,262.93) node [anchor=north west][inner sep=0.75pt]  [font=\footnotesize] [align=left] {$\displaystyle \{a_{0,0}\}$};
	\draw (354.27,262.33) node [anchor=north west][inner sep=0.75pt]  [font=\footnotesize] [align=left] {$\displaystyle \{a_{n,0}\}$};
	\draw (221.47,232.53) node [anchor=north west][inner sep=0.75pt]  [font=\footnotesize] [align=left] {$\displaystyle \{a_{0,1}\}$};
	\draw (217.67,176.53) node [anchor=north west][inner sep=0.75pt]  [font=\footnotesize] [align=left] {$\displaystyle \{a_{0,m}\}$};
	\draw (354.27,176.13) node [anchor=north west][inner sep=0.75pt]  [font=\footnotesize] [align=left] {$\displaystyle \{a_{n,m}\}$};
	\draw (354.27,230.93) node [anchor=north west][inner sep=0.75pt]  [font=\footnotesize] [align=left] {$\displaystyle \{a_{n,1}\}$};
	\draw (255.27,274.93) node [anchor=north west][inner sep=0.75pt]   [align=left] {$\displaystyle o$};

\end{tikzpicture}
\end{center}	
	We show that $\Kmc\models \neg G(a)$ and $\Kmc\not\models \neg G(b)$, thus $G$ separates $(\Kmc,\{a,\{b\})$.
	
	Assume first for a proof by contradiction that there is a model $\Amf$ of $\Kmc$ such that $\Amf \models G(a)$. Then $a^{\Amf}=o^{\Amf}$ and so $a^{\Amf}\in (A_{1}\sqcap D)^{\Amf}$. But then $a^{\Amf}\in Q^{\Amf}$.
	This contradicts the fact that $o^{\Amf}$ is the origin of an $n\times m$-grid in $\Amf$.
	
	Now for $\Kmc \not\models \neg G(b)$. We find a model \Amf of \Kmc with
	$b^\Amf \in G^\Amf$ since the concept name $Q$ is not triggered at $b$ as $A_{1}$ is not true for $b$.
\end{proof}
The following lemma implies that if $S$ has no solution, then
$(\Kmc,\{a\},\{b\})$ is not projectively FO$(\Sigma)$-separable.
\begin{lemma}
	If $S$ has no solution, then for every model \Amf of \Kmc, there is
	a model \Bmf of \Kmc such that $(\Amf,b^\Amf)$ is
	$\Sigma$-isomorphic to $(\Bmf,a^\Amf)$.
\end{lemma}
\begin{proof} \
	(sketch)  
	If $b^\Amf \not=o^{\Amf}$, then we can simply obtain \Bmf
	from \Amf by switching $a^\Amf$ and $b^\Amf$, making $A_1$ true at
	$a^\Amf$, and $D$ at exactly $o^{\Amf}$. If $b^\Amf = o^\Amf$, then after switching we
	additionally have to re-interpret $Q$ and $P$ in a suitable way. But
	$S$ has no solution and thus when following $r_x$/$r_y$-paths from
	$o^{\Amf}$ in \Amf, we must either encounter an infinite such path or a
	non-closing grid cell as otherwise we can extract from \Amf a solution
	for $S$. Thus we can re-interpret $Q$ and $P$ as required.
\end{proof}

\proprolehelp*
\begin{proof} \
	First assume that $\Lmc\in \{\ALC,\ALCI,\ALCO\}$. We employ the characterization of projective separability given below in 
	Theorem~\ref{thm:L-modeltheory0}. Observe that the following conditions are equivalent:
	\begin{itemize}
		\item there exists an $\Lmc(\Sigma\cup (\NC\cup \NR)\setminus \text{sig}(\Kmc))$-concept $C$ such that $\Kmc\models C(a)$ for all $a\in P$ and $\Kmc\not\models C(b)$;
		\item there exists an $\Lmc$-forest model $\Amf$ of $\Kmc$ of finite
		$\Lmc$-outdegree and a set $\Sigma'$ of concept and role names disjoint from $\text{sig}(\Kmc)$ such that for all models $\Bmf$ of
		$\Kmc$ and all $a\in P$: $\Bmf,a^{\Bmf}
		\not\sim_{\Lmc,\Sigma\cup \Sigma'} \Amf, b^{\Amf}$.
	\end{itemize}
	Thus, for $\Lmc\in \{\ALC,\ALCI\}$ it suffices to show that the second condition is equivalent to the third condition of Theorem~\ref{thm:L-modeltheory0}. For a proof by contradiction
	assume that there exists an $\Lmc$-forest model $\Amf$ of $\Kmc$ satisfying Condition~2 above for $\Sigma'$
	but not Condition~3 of Theorem~\ref{thm:L-modeltheory0}. Take a model $\Bmf$ of
	$\Kmc$ and a functional $\Sigma$-bisimulation $f$ witnessing $\Bmf,a^{\Bmf}\sim_{\Lmc,\Sigma}^{f}\Amf,b^{\Amf}$ for some $a\in P$.
	We modify $\Bmf$ to obtain a model $\Bmf'$ of $\Kmc$ such that
	$\Bmf',a^{\Bmf'}\sim_{\Lmc,\Sigma\cup \Sigma'}^{f} \Amf, b^{\Amf}$ and thus obtain a contradiction. $\Bmf'$ is obtained from $\Bmf$ by assuming first that $\Bmf$
	does not interpret any symbol in $\Sigma'$ and then 
	\begin{itemize}
		\item taking the disjoint union $\Bmf\cup \Amf'$ of $\Bmf$ and a copy $\Amf'$ of $\Amf$ that does not interpret any individual names nor symbols
		in $\Sigma'$.
		\item observing that the function $g =f \cup {\sf id}$, where ${\sf id}$ maps every node in $\Amf'$ to the node in $\Amf$ of which it is a copy, is a functional and surjective $\Lmc(\Sigma)$-bisimulation between $\Bmf\cup \Amf'$ and $\Amf$.		
		\item setting $A^{\Bmf'}= g^{-1}(A^{\Amf})$ for all concept names $A\in \Sigma'$ and $r^{\Bmf'}= g^{-1}(r^{\Amf})$ for all role names $r\in \Sigma'$.
	\end{itemize}
	It is easy to see that $g$ is a functional $\Lmc(\Sigma\cup\Sigma')$-bisimulation
	between $\Bmf'$ and $\Amf$, as required.
	
	Now assume that $\Lmc=\mathcal{ALCO}$. As Condition~2 trivially implies Condition~1, it suffices to show that Condition~1 implies Condition~2.
	Assume that Condition~1 holds. Take an $\Lmc$-forest model $\Amf$ of $\Kmc$ of finite
	$\Lmc$-outdegree and a set $\Sigma'$ of concept and role names disjoint from $\text{sig}(\Kmc)$ such that for all models $\Bmf$ of
	$\Kmc$ and all $a\in P$: $\Bmf,a^{\Bmf}
	\not\sim_{\Lmc,\Sigma\cup \Sigma'} \Amf, b^{\Amf}$. Assume for a proof by contradiction that there does not exist any such model if $\Sigma'$ is replaced by $\Sigma\cup \{r_{I}\}$. Obtain $\Amf'$ from $\Amf$ dropping the interpretation of role names in $\Sigma'$ and instead setting 
	$$
	r_{I}^{\Amf'}= \{(b^{\Amf},c^{\Amf}) \mid c\in \text{ind}(\Dmc)\} \cup\bigcup_{r\in \NR}r^{\Amf}.
	$$
	Then $\Amf'$ is an $\Lmc$-forest model of $\Kmc$ of finite $\Lmc$-outdegree.
	Thus, by assumption there exists a model $\Bmf$ of $\Kmc$ and $a\in P$ such that $\Bmf,a^{\Bmf} \sim_{\Lmc,\Sigma\cup \{r_{I}\}}^{f} \Amf',b^{\Amf'}$.
	Let $f$ be the functional $\Lmc(\Sigma\cup\{r_{I}\})$-bisimulation witnessing this. Then $f$ is surjective. Now obtain $\Bmf'$ from $\Bmf$ by keeping the 
	interpretation of symbols not in $\Sigma'$ and setting $A^{\Bmf'}= f^{-1}(A^{\Amf})$ for all concept names $A\in \Sigma'$ and $r^{\Bmf'}= f^{-1}(r^{\Amf})$ for all role names $r\in \Sigma'$. Then $f$ is a functional
	$\Lmc(\Sigma\cup \Sigma')$-bisimulation between $\Bmf',a^{\Bmf'}$ and $\Amf,b^{\Amf}$ and we have derived a contradiction.
	
	Finally, assume that $\Lmc=\mathcal{ALCIO}$. Then we cannot use Theorem~\ref{thm:L-modeltheory0} as it does not hold for $\mathcal{ALCIO}$.
	However, if one replaces $\Lmc$-forest models of finite $\Lmc$-outdegree by $\omega$-saturated models, then Theorem~\ref{thm:L-modeltheory0} holds for $\mathcal{ALCIO}$. Now exactly the same proof can be done for $\mathcal{ALCIO}$ as for $\mathcal{ALCO}$ using $\omega$-saturated models instead of forest models.  	
\end{proof}

\thmfoundec*

The proof is by reduction of the same tiling problem as in the proof of Theorem~\ref{thm:nominalhelper}. If fact, given a tiling system $S$, the labeled 
KB $(\Kmc,\{a\},\{b\})$ is exactly the same KB as in the proof of Theorem~\ref{thm:nominalhelper}. The only difference is in  Lemma~\ref{lemfisrtone} about the construction
of a concept witnessing separability: this concept is now not a concept using individual names as helper symbols but a $\mathcal{ALCFIO}(\Sigma)$-concept
without helper symbols.
\begin{lemma}
	If  $S$ has a solution, then there is an
	$\mathcal{ALCFIO}(\Sigma)$-concept that non-projectively separates $(\Kmc,\{a\},\{b\})$.
\end{lemma}
\begin{proof} \
	Assume that $S$ has a solution consisting of a properly tiled $n\times m$ grid. We design an $\mathcal{ALCFIO}(\Sigma)$-concept $G$ so that any model of $G$ and $\Kmc$ includes a properly tiled $n \times m$-grid with lower
	left corner $o$. Let $F$ be the obvious concept stating that $(\leqslant 1\ r)$ holds for $r\in \{r_{x},r_{x},r_{y}^{-},r_{x}^{-}\}$ for all nodes reachable in no more than $2(n+m)$ steps along roles $r_{x},r_{x},r_{y},r_{x}^{-}$.
	For every word $w \in \{r_x,r_y\}^*$, denote by $\overleftarrow{w}$
	the word that is obtained by reversing $w$ and then adding
	$\cdot^{-}$ to each symbol. Let $|w|_r$ denote the number of
	occurrences of the symbol $r$ in $w$.  Now let $G= F \sqcap E$, where $E$ is the conjunction of
	$$
	\{o\} \sqcap \forall r^{n+1}_{x}.\bot \sqcap 
	\forall r_x^{\leq n}. \mn{bottom} \sqcap \forall r^{m+1}_{y}.\bot \sqcap \forall r_y^{\leq m}. \mn{left}
	$$
	and for every $w \in \{r_x,r_y\}^*$ such that $|w|_{r_x} < n$ and
	$|w|_{r_y} < m$, the concept 
	$$ 
	\exists (w \cdot r_x r_y r_x^- r_y^- \cdot \overleftarrow{w}) . \{ o
	\},
	$$
	where $\exists w . F$ abbreviates
	$\exists r_1 . \cdots \exists r_k .  F$ if $w=r_1 \cdots r_k$.  It is
	readily checked that $G$ indeed enforces a grid, as announced. 
	
	We show that $\Kmc\models \neg G(a)$ and $\Kmc\not\models \neg G(b)$, thus $G$ separates $(\Kmc,\{a,\{b\})$.
	
	Assume first for a proof by contradiction that there is a model $\Amf$ of $\Kmc$ such that $\Amf \models G(a)$. Then $a^{\Amf}=o^{\Amf}$ and so $a^{\Amf}\in (A_{1}\sqcap D)^{\Amf}$. But then $a^{\Amf}\in Q^{\Amf}$.
	This contradicts the fact that $o^{\Amf}$ is the origin of an $n\times m$-grid in $\Amf$.
	
	Now for $\Kmc \not\models \neg G(b)$. We find a model \Amf of \Kmc with
	$b^\Amf \in G^\Amf$ since the concept name $Q$ is not triggered at $b$ as $A_{1}$ is not true for $b$.
\end{proof}
\section{Proofs for Section~\ref{sec:modelandequi}}

\thmLmodeltheorynull*

\begin{proof} \
	``1. $\Rightarrow$ 2''. Assume that $(\Kmc,P,\{b\})$ is projectively $\Lmc(\Sigma)$-separable. Take an $\Lmc$-concept $C$ that separates $(\Kmc,P,\{b\})$ and uses symbols from $\Sigma\cup \Sigma_{\text{help}}$,
	where $\Sigma_{\text{help}}$ is a set of concept names disjoint from $\text{sig}(\Kmc)$. By Lemma~\ref{lem:forestm}, there exists a $\Lmc$-forest model $\Amf$ of $\Kmc$ of finite $\Lmc$-outdegree
	such that $b^{\Amf}\in (\neg C)^{\Amf}$. Then $\Amf$ is as required for 
	Condition~2, by Lemma~\ref{lem:equivalence}.
	
	``2 $\Rightarrow$ 1''. Assume Condition~2 holds for $\Amf$ and $\Sigma_{\text{help}}$. Let 
	$$
	t_{\Amf}(b) = \{C \in \Lmc(\Sigma\cup \Sigma_{\text{help}}) \mid b^{\Amf}\in C^{\Amf}\}.
	$$
	It follows from Lemma~\ref{lem:equivalence} that 
	$$
	\Gamma_{a}= \Kmc \cup \{C(a) \mid C \in t_{\Amf}(b)\}
	$$
	is not satisfiable, for any $a\in P$. (Otherwise an $\omega$-saturated satisfying model would contradict Condition~2.) By compactness (and closure under conjunctions) we find for every $ a\in P$ a concept $C_{a}\in t_{\Amf}(b)$ such that $\Kmc\models \neg C(b)$. Thus, the concept $\neg (\bigsqcap_{a\in P}C_{a})$ separates $(\Kmc,P,\{b\})$, as required.
	
	``2 $\Rightarrow$ 3''. Take an $\Lmc$-forest model $\Amf$ and $\Sigma_{\text{help}}$ such that Condition 2 holds. 
	We show that Condition~3 holds for $\Amf$ as well. 
	Suppose for a proof by contradiction that there exists a model $\Bmf$ of $\Kmc$ and an $a\in P$ and a functional $\Sigma$-bisimulation $f$ witnessing $\Bmf,a^{\Bmf} \sim_{\Lmc,\Sigma}^{f} \Amf,b^{\Amf}$.
	We may assume that $\Bmf$ does not interpret any symbols in $\Sigma_{\text{help}}$  
	and define $\Bmf'$ by expanding $\Bmf$ as follows: 
	for every concept name $A\in \Sigma_{\text{help}}$
	and $d\in \text{dom}(f)$, let $d\in A^{\Bmf'}$ if $f(d) \in A^{\Amf}$.
	It is easy to see that $f$ witnesses $\Bmf',a^{\Bmf'} \sim_{\Lmc,\Sigma\cup 
	\Sigma_{\text{help}}} \Amf,b^{\Amf}$, and we have derived a contradiction.
    
    ``3 $\Rightarrow$ 2''. Take an $\Lmc$-forest model $\Amf$ of $\Kmc$ of finite $\Lmc$-outdegree such that Condition~3 holds. We may assume that $\Amf$ only interprets the symbols in $\text{sig}(\Kmc)$.  
    Define $\Amf'$ by expanding $\Amf$ as follows. Take for any $d\in \text{dom}(\Amf)$ a fresh concept name $A_{d}$ and set
    $A_{d}^{\Amf'} =\{d\}$. 
    Then Condition~2 holds for $\Amf'$ and $\Sigma_{\text{help}} = \{A_{d}\mid d\in \text{dom}(\Amf)\}$. 
\end{proof}
We next show that Theorem~\ref{thm:L-modeltheory0} does not hold for $\mathcal{ALCIO}$. To this end we define a labeled $\mathcal{ALCI}$-KB
$(\Kmc,\{a\},\{b\})$ and signature $\Sigma$ such that $(\Kmc,\{a\},\{b\})$ is  weakly $\mathcal{ALCIO}(\Sigma)$-separable but there does not exist an $\mathcal{ALCIO}$-forest model $\Amf$ of $\Kmc$ of finite $\mathcal{ALCIO}$-outdegree
such that for all models $\mathfrak{B}$ of $\mathcal{K}$: $\mathfrak{A}, b^\mathfrak{A}
\not\sim^f_{\mathcal{ALCIO},\Sigma}
\mathfrak{B}, a^\mathfrak{B}$.

Let $\mathcal{K} = (\mathcal{O,D})$ with
\begin{align*}
	\mathcal{D} = \{& A(a), B(b), C(c), r_0(b,c)\}\\
	\mathcal{O} = \{& C \sqsubseteq \exists r_0^-. A \rightarrow
	A_0,\\ 
	&C \sqsubseteq \exists s.\top \sqcap \forall s.(E \sqcap \exists r.(E\sqcap \exists s^{-}.\top))\\ & A_0 \sqsubseteq \exists s.\exists s^-.\neg A_{0} \sqcup \exists s.\exists r.\exists s^{-}.\neg A_{0}, \\ 
	&B \sqcup A \sqsubseteq \neg C\}
\end{align*}
where $E$ stands for $\neg C \sqcap \neg A \sqcap \neg B$.
Let $\Sigma = \{c, r_0, s,r\}$.
\begin{lemma}
	The $\mathcal{ALCIO}(\Sigma)$-concept 
	$$D=\neg \exists r_0 (\{c\} \sqcap \forall s. (\forall s^-. \{c\} \sqcap \forall r. \forall s^-. \{c\}))
	$$ 
weakly separates $(\mathcal{K},\{a\},\{b\})$.
\end{lemma}
\begin{proof} \
	We first show that $\mathcal{K} \vDash D(a)$.
	Assume there is a model $\Amf$ of $\Kmc$ with $a^{\Amf}\not\in D^{\Amf}$.
	Then $(a^{\Amf},c^{\Amf}) \in
	r^\mathfrak{A}_0$. Then by
	definition of $\mathcal{K}$ we have
	$c^\mathfrak{A} \in
	A^\mathfrak{A}_0$ thus
	$c^\mathfrak{A} \in (\exists s.
	\exists s^-. \neg A_0 \sqcup \exists
	s.\exists r.\exists s^{-}.\neg
	A_{0})^\mathfrak{A}$, contradicting $c^\mathfrak{A} \in
	(\forall s. (\forall s^-. \{c\}
	\sqcap \forall r. \forall s^-.
	\{c\}))^\mathfrak{A}$. On the other hand, the model depicted below
	
	\begin{center}

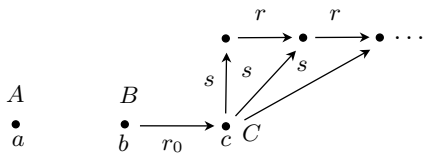
	\tikzset{every picture/.style={line width=0.5pt}} 
	
	\begin{tikzpicture}[x=0.75pt,y=0.75pt,yscale=-1,xscale=1]
		
		\draw  [fill={rgb, 255:red, 0; green, 0; blue, 0 }  ,fill opacity=1 ] (297.1,249.25) .. controls (297.1,248.28) and (296.32,247.5) .. (295.35,247.5) .. controls (294.38,247.5) and (293.6,248.28) .. (293.6,249.25) .. controls (293.6,250.22) and (294.38,251) .. (295.35,251) .. controls (296.32,251) and (297.1,250.22) .. (297.1,249.25) -- cycle ;
		\draw    (303,250) -- (337.56,250.05) ;
		\draw [shift={(340.56,250.06)}, rotate = 180.08] [fill={rgb, 255:red, 0; green, 0; blue, 0 }  ][line width=0.08]  [draw opacity=0] (5,-2.5) -- (0,0) -- (5,2.5) -- (3.5,0) -- cycle    ;
		\draw  [fill={rgb, 255:red, 0; green, 0; blue, 0 }  ,fill opacity=1 ] (348.1,250.25) .. controls (348.1,249.28) and (347.32,248.5) .. (346.35,248.5) .. controls (345.38,248.5) and (344.6,249.28) .. (344.6,250.25) .. controls (344.6,251.22) and (345.38,252) .. (346.35,252) .. controls (347.32,252) and (348.1,251.22) .. (348.1,250.25) -- cycle ;
		\draw  [fill={rgb, 255:red, 0; green, 0; blue, 0 }  ,fill opacity=1 ] (346.29,204.17) .. controls (345.32,204.16) and (344.53,204.94) .. (344.53,205.91) .. controls (344.52,206.87) and (345.3,207.66) .. (346.26,207.67) .. controls (347.23,207.67) and (348.02,206.9) .. (348.03,205.93) .. controls (348.03,204.96) and (347.25,204.17) .. (346.29,204.17) -- cycle ;
		\draw  [fill={rgb, 255:red, 0; green, 0; blue, 0 }  ,fill opacity=1 ] (385.29,203.77) .. controls (384.33,203.76) and (383.54,204.54) .. (383.53,205.5) .. controls (383.52,206.47) and (384.3,207.26) .. (385.27,207.27) .. controls (386.23,207.27) and (387.02,206.49) .. (387.03,205.53) .. controls (387.04,204.56) and (386.26,203.77) .. (385.29,203.77) -- cycle ;
		\draw  [fill={rgb, 255:red, 0; green, 0; blue, 0 }  ,fill opacity=1 ] (423.96,203.7) .. controls (422.99,203.69) and (422.2,204.47) .. (422.2,205.43) .. controls (422.19,206.4) and (422.97,207.19) .. (423.94,207.2) .. controls (424.9,207.2) and (425.69,206.42) .. (425.7,205.46) .. controls (425.7,204.49) and (424.93,203.7) .. (423.96,203.7) -- cycle ;
		\draw    (352.39,205.55) -- (376.8,205.42) ;
		\draw [shift={(379.8,205.4)}, rotate = 539.69] [fill={rgb, 255:red, 0; green, 0; blue, 0 }  ][line width=0.08]  [draw opacity=0] (5,-2.5) -- (0,0) -- (5,2.5) -- (3.5,0) -- cycle    ;
		\draw    (391.39,205.43) -- (415.81,205.3) ;
		\draw [shift={(418.81,205.29)}, rotate = 539.69] [fill={rgb, 255:red, 0; green, 0; blue, 0 }  ][line width=0.08]  [draw opacity=0] (5,-2.5) -- (0,0) -- (5,2.5) -- (3.5,0) -- cycle    ;
		\draw    (346.23,243.76) -- (346.62,216.08) ;
		\draw [shift={(346.67,213.08)}, rotate = 450.81] [fill={rgb, 255:red, 0; green, 0; blue, 0 }  ][line width=0.08]  [draw opacity=0] (5,-2.5) -- (0,0) -- (5,2.5) -- (3.5,0) -- cycle    ;
		\draw  [fill={rgb, 255:red, 0; green, 0; blue, 0 }  ,fill opacity=1 ] (242.1,249.25) .. controls (242.1,248.28) and (241.32,247.5) .. (240.35,247.5) .. controls (239.38,247.5) and (238.6,248.28) .. (238.6,249.25) .. controls (238.6,250.22) and (239.38,251) .. (240.35,251) .. controls (241.32,251) and (242.1,250.22) .. (242.1,249.25) -- cycle ;
		\draw    (351.17,245.08) -- (379.63,214.29) ;
		\draw [shift={(381.67,212.08)}, rotate = 492.75] [fill={rgb, 255:red, 0; green, 0; blue, 0 }  ][line width=0.08]  [draw opacity=0] (5,-2.5) -- (0,0) -- (5,2.5) -- (3.5,0) -- cycle    ;
		\draw    (355.67,247.08) -- (416.55,213.05) ;
		\draw [shift={(419.17,211.58)}, rotate = 510.79] [fill={rgb, 255:red, 0; green, 0; blue, 0 }  ][line width=0.08]  [draw opacity=0] (5,-2.5) -- (0,0) -- (5,2.5) -- (3.5,0) -- cycle    ;
		
		\draw (290.74,253.1) node [anchor=north west][inner sep=0.75pt]  [font=\small] [align=left] {$\displaystyle {b}$};
		\draw (342.24,253.6) node [anchor=north west][inner sep=0.75pt]  [font=\small] [align=left] {$\displaystyle c$};
		\draw (429.63,203.54) node [anchor=north west][inner sep=0.75pt]  [rotate=-0.13] [align=left] {$\displaystyle \dotsc $};
		\draw (312.62,255.57) node [anchor=north west][inner sep=0.75pt]  [font=\footnotesize] [align=left] {$\displaystyle r_{0}$};
		\draw (237.07,253.76) node [anchor=north west][inner sep=0.75pt]  [font=\small] [align=left] {$\displaystyle {a}$};
		\draw (333.67,223.33) node [anchor=north west][inner sep=0.75pt]  [font=\footnotesize] [align=left] {$\displaystyle s$};
		\draw (352.67,218.83) node [anchor=north west][inner sep=0.75pt]  [font=\footnotesize] [align=left] {$\displaystyle s$};
		\draw (380.67,215.83) node [anchor=north west][inner sep=0.75pt]  [font=\footnotesize] [align=left] {$\displaystyle s$};
		\draw (359.67,190.33) node [anchor=north west][inner sep=0.75pt]  [font=\footnotesize] [align=left] {$\displaystyle r$};
		\draw (397.17,190.33) node [anchor=north west][inner sep=0.75pt]  [font=\footnotesize] [align=left] {$\displaystyle r$};
		\draw (234.17,228.33) node [anchor=north west][inner sep=0.75pt]  [font=\small] [align=left] {$\displaystyle A$};
		\draw (291.17,228.83) node [anchor=north west][inner sep=0.75pt]  [font=\small] [align=left] {$\displaystyle B$};
		\draw (353.17,248.08) node [anchor=north west][inner sep=0.75pt]  [font=\small] [align=left] {$\displaystyle C$};

	\end{tikzpicture}
\end{center}		
	 
	clearly satisfies $D(b)$. The fact that it is a model
	of $\mathcal{K}$ is also
	straightforward, as its extension of
	$A_0$ is empty.

\end{proof}

\begin{lemma}
	For every $\mathcal{ALCIO}$-forest model $\mathfrak{A}$ of
	$\mathcal{K}$ of finite $\mathcal{ALCIO}$-outdegree there exists a model
	$\mathfrak{B}$ of $\mathcal{K}$ such
	that $\mathfrak{B}, a^\mathfrak{B}
	\sim^f_{\mathcal{ALCIO},\Sigma}
	\mathfrak{A}, b^\mathfrak{A}$.
\end{lemma}
\begin{proof} \
	Assume $\Amf$ is given. We construct $\Bmf$ as follows. Let $\text{dom}(\mathfrak{B}) =
	\text{dom}(\mathfrak{A})$,
	$a^\mathfrak{B} = b^\mathfrak{B} =
	b^\mathfrak{A}$, $c^\mathfrak{B} =
	c^\mathfrak{A}$, $A^\mathfrak{B} =
	\{a^\mathfrak{B}\}$,
	$A_0^\mathfrak{B} = C^\mathfrak{B}=
	\{c^\mathfrak{B}\}$, $C^\mathfrak{B}
	= \{c^{\Amf}\}$, and $B^{\Bmf}= B^{\Amf}$. Let
	$\rho^\mathfrak{B} =
	\rho^\mathfrak{A}$ for all role
	names $\rho$. There is a $\Sigma$-isomorphism between $\mathfrak{A},b^\mathfrak{A}$ and $\mathfrak{B},a^\mathfrak{B}$, as $\mathfrak{B}$ only differs from $\mathfrak{A}$ with respect to symbols outside of $\Sigma$. 
	It is clear that $\mathfrak{B}$ is a model of $\mathcal{D}$. We then check that $\mathfrak{B}$ satisfies each inclusion of $\mathcal{O}$.
	The first inclusion is clearly satisfied as $C^\mathfrak{B} = A_0^\mathfrak{B} = \{c^\mathfrak{B}\}$.
	The second and fourth inclusion are clearly satisfied by $\mathfrak{B}$ as they are by $\mathfrak{A}$.
	The third inclusion is satisfied: we
	have $A_0^\mathfrak{B} =
	\{c^\mathfrak{B}\}$. Assume the third inclusion is not satisfied. Then, by the second
	inclusion and $C(c)\in \Dmc$, there is an infinite $r^{\Amf}$-chain of nodes distinct from $c^{\Amf},a^{\Amf},b^{\Amf}$ all of which are in relation $(s^{-})^{\Amf}$ to $c^{\Amf}$. Then either $\Amf$ is not an $\mathcal{ALCIO}$-forest model as it contains an $r^{\Amf}$-cycle of nodes distinct from interpretations of individual names or it does not have finite $\mathcal{ALCIO}$-outdegree as the outdegree of $c^{\Amf}$ is infinite.
\end{proof}
We now show that Theorem~\ref{thm:L-modeltheory0} cannot be repaired 
for $\mathcal{ALCIO}$ by admitting infinite outdegree forest models.
To this end we define a labeled $\mathcal{ALCI}$-KB
$(\Kmc,\{a\},\{b\})$ and signature $\Sigma$ such that $(\Kmc,\{a\},\{b\})$ is  not projectively weakly $\mathcal{ALCIO}(\Sigma)$-separable but there 
exists an $\mathcal{ALCIO}$-forest model $\Amf$ of $\Kmc$ of such that for all models $\mathfrak{B}$ of $\mathcal{K}$: $\mathfrak{B}, a^\mathfrak{B}
\not\sim^f_{\mathcal{ALCIO},\Sigma}
\mathfrak{A}, b^\mathfrak{A}$.

Let $\mathcal{K} = (\mathcal{O,D})$ with 
\begin{align*}
	\mathcal{D} = \{& A(a), B(b), C(c), r_0(b,c)\}\\
	\mathcal{O} = \{& C \sqsubseteq \exists r_0^-. A \rightarrow
	A_0,\\ & A_0 \sqsubseteq (\exists s. \top \sqcap \forall s.
	\exists r. \exists s^-. A_0) \rightarrow \exists s. B', \\ &
	B'
	\sqsubseteq \exists r^-.B'\\
	&B \sqcup A \sqsubseteq \neg C\}
\end{align*}
Let $\Sigma = \{r_0, s, r, c\}$. 

\begin{lemma} $(\mathcal{K}, \{a\}, \{b\})$ is not weakly projectively
	$\mathcal{ALCIO}(\Sigma)$-separable.
\end{lemma}
\begin{proof} \
	Let $\mathfrak{A}$ be a model of $\mathcal{K}$. We show
	there exists a model 
	$\mathfrak{B}$ of $\mathcal{K}$ such that
	$\mathfrak{B},a^\mathfrak{B} \equiv_{\mathcal{ALCIO},\Sigma\cup \Sigma'}
	\mathfrak{A},b^\mathfrak{A}$, where $\Sigma'$ is the set of all concept names that do not occur in $\Kmc$.
	Let the model $\mathfrak{B}_0$ be defined in the same way as $\Amf$ except
	that $a^{\Bmf_{0}}=b^{\Amf}$, $A^{\mathfrak{B}_0} = \{a^{\mathfrak{B}_0}\}$
	$A_0^{\mathfrak{B}_0} = C^{\mathfrak{B}_0} = \{c^{\mathfrak{B}_0} \}$, and
	$B^{\mathfrak{B}_0} = \{b^{\mathfrak{B}_0}\}$, and we define the extension of $B'$ according to a case distinction.
	
	Case 1. $c^{\mathfrak{B}_0} \notin (\exists s. \top \sqcap \forall s. \exists r. \exists s^-. A_0)^{\Bmf_{0}}$. Then set $B'^{\Bmf_0}=\emptyset$. Then $\Bmf_{0}$ is a model of $\Kmc$ and the identity is a $\Sigma$-isomorphism between $\mathfrak{B}_0$ and $\mathfrak{A}$ mapping $a^{\Bmf_{0}}$ to $b^\mathfrak{A}$ and we are done.
	
	Case 2. Otherwise. As $A_0^{\mathfrak{B}_0} =
	\{c^{\mathfrak{B}_0}\}$, the set
	$$
	t(x)=\{s(c,x)\} \cup \{r(y_{1},x), r(y_{2},y_{1}), r(y_{3},y_{2}),\ldots\}
	$$ 
	is finitely satisfiable in
	$\mathfrak{B}_0$, so it is realized in an elementary
	extension $\mathfrak{B}_1$ of $\mathfrak{B}_0$. That
	implies there exists an infinite $r^-$-chain $a'_1, a'_2,
	\dots$ in $\mathfrak{B}_1$ with $(c^{\mathfrak{B}_1},a'_1)
	\in s^{\mathfrak{B}_1}$. Let $\mathfrak{B}$ be obtained from $\mathfrak{B}_1$ by defining the extension of $B'$ as $\{a'_i : i \geq
	1\}$. Then $\mathfrak{B}$ is a model of $\mathcal{K}$ and $\mathfrak{A},b^\mathfrak{A} \equiv_{\mathcal{ALCIO},\Sigma\cup \Sigma'} \mathfrak{B},a^\mathfrak{B}$.
\end{proof}

\begin{lemma} There exists a model $\mathfrak{A}$ of $\mathcal{K}$ such that $\mathfrak{B},a^\mathfrak{B} \nsim_{\mathcal{ALCIO}, \Sigma} \mathfrak{A},b^\mathfrak{A}$ for all models $\mathfrak{B}$ of $\mathcal{K}$.
\end{lemma}

\begin{proof} \
	Consider the model $\mathfrak{A}$ depicted above. An explicit definition is given by setting 
	$$
	\text{dom}(\mathfrak{A})= \{a,b,c\} \cup \{a_i : i \geq 0\}
	$$
	and 
	\begin{align*}
		A_0^\mathfrak{A} &= \emptyset & (B')^\mathfrak{A} &= \emptyset\\
		A^\mathfrak{A} &= \{a\} = a^\mathfrak{A} &  r_0^\mathfrak{A} &= \{b,c\}\\
		B^\mathfrak{A} &= \{b\} = b^\mathfrak{A} &  r^\mathfrak{A} &= \{(a_i, a_{i+1}) : i \geq 0\}\\
		C^\mathfrak{A} &= \{c\} = c^\mathfrak{A} &  s^\mathfrak{A} &= \{(c,a_i) : i \geq 0\} 
	\end{align*}
	It is immediate that $\Amf$ is a model of $\Kmc$.
	If $\mathfrak{B},a^\mathfrak{B} \sim_{\mathcal{ALCIO}, \Sigma}
	\mathfrak{A},b^\mathfrak{A}$, then $b^\mathfrak{A} \in
	(\exists r_0. (\{c\} \sqcap \exists s. \top \sqcap \forall s.
	\exists r. \exists s^-. \{c\}))^\mathfrak{A}$ implies
	$a^\mathfrak{B} \in (\exists r_0. (\{c\} \sqcap \exists s. \top
	\sqcap \forall s. \exists r. \exists s^-. \{c\}))^\mathfrak{B}$
	as $c,s,r,r_0 \in \Sigma$. The latter implies $a^\mathfrak{B}
	\in (\exists r_0. (A_0 \sqcap \exists s. \top \sqcap \forall s.
	\exists r. \exists s^-. A_0))^\mathfrak{B}$ as $c^\mathfrak{B}
	\in A_0^\mathfrak{B}$ in virtue of $\{A(a), C(c),r_0(a,c)\}
	\subseteq \mathcal{D}$ and the first inclusion $C \sqsubseteq \exists r_0^-. A
	\rightarrow A_0$ of $\mathcal{O}$. Then, by the second inclusion 
	we get that $a^\mathfrak{B} \in (\exists r_0. \exists s.
	B')^\mathfrak{B}$ while $b^\mathfrak{A} \notin (\exists r_0.
	\exists s. B')^\mathfrak{A}$. By the third inclusion,
	$a^\mathfrak{B}$ then has a $r_0$ successor with an $s$
	successor from which starts an infinite $r^-$ chain, while
	$b^\mathfrak{A}$ does not. A $\Sigma$-bisimulation including
	$(a^\mathfrak{B}, b^\mathfrak{A})$ is then impossible, as
	$\{r_0,s,r\} \subseteq \Sigma$.
\end{proof}

The following lemma provides a model-theoretic characterization of
$(\Lmc,\text{UCQ}_{r}^{\Lmc_{S}})$-separability, for some pairs
$(\Lmc,\Lmc_{S})$.
\begin{lemma}\label{lem:helps}
Let $(\Lmc,\Lmc_{S})$ be either $(\ALCI,\ALCI)$ or $(\ALC, \mathcal{ALCO})$ and let $(\Kmc,P,\{b\}\})$ be a labeled $\Lmc$-KB and $\Sigma\subseteq \text{sig}(\Kmc)$ a signature. Then the following conditions are equivalent:
	\begin{enumerate}		
		\item $(\Kmc,P,\{b\})$ is non-projectively UCQ$_{r}^{\Lmc_{S}}(\Sigma)$-separable; 
        \item there exists an $\Lmc_{S}$-forest model $\Amf$ of $\Kmc$ of finite
        $\Lmc_{S}$-outdegree such that there exists an $n$ such that for all models $\Bmf$ of 
        $a\in P$: there exist $D\subseteq \text{dom}(\Bmf)$ of cardinality not
        exceeding $n$ such that the $\Sigma$-reduct of $\Amf_{|D}$ is $\Lmc_{S}$-rooted in $a^{\Bmf}$ and $\Bmf,a^{\Bmf}\not\rightarrow_{D,\mathcal{ALCO},\Sigma}
        \Amf,b^{\Amf}$. 	
\end{enumerate} 
\end{lemma} 
\begin{proof} \
``1 $\Rightarrow$ 2''. Assume that Condition~1 holds and take a formula $\varphi(x)$ in UCQ$_{r}^{\Lmc_{S}}(\Sigma)$ such that $\Kmc\models \varphi(a)$ for all $a\in P$ and $\Kmc\not\models \varphi(b)$. Then there exists a model $\Amf$ of
$\Kmc$ such that $\Amf\not\models \varphi(b)$. By Lemma~\ref{lem:forestmodelcompleteness}, we may assume that $\Amf$ is an $\Lmc_{S}$-forest model of $\Kmc$ of finite $\Lmc_{S}$-outdegree. Then Condition~2 follows for $n$ the number of variables in $\varphi$, using Lemma~\ref{lem:equivalence2}.	

``2 $\Rightarrow$ 1''. The proof is indirect. Assume that Condition~1 does not hold.
Let $\Amf$ be any $\Lmc_{S}$-forest model of $\Kmc$ of finite $\Lmc_{S}$-outdegree and set
$$
\Gamma = \Kmc \cup \{ \neg \varphi(x) \mid \varphi(x)\in \text{UCQ}_{r}^{\Lmc_{s}}(\Sigma),\Amf\models \neg \varphi(b)\}.
$$ 
Then, by compactness, $\Gamma$ is satisfiable with $x=a$ for some $a\in P$. To show this, assume that it is not the case. Then for any $a\in P$ there exists a finite subset $\Gamma_{a}'$ of $\Gamma$ such that $\Gamma_{a}'$ is not satisfiable with $x=a$. Then $\Gamma'= \bigcup_{a\in P}\Gamma_{a}'$ is not satisfiable with $x=a$, for any $a\in P$. We may assume that $\Gamma'= \Kmc \cup \{\neg \varphi_{1}(x),\cdots,\neg \varphi_{n}(x)\}$. Then $\Kmc\models \varphi_{1}\vee \cdots \vee \varphi_{n}(a)$ for all $a\in P$. Observe that $\varphi_{1}\vee \cdots \vee \varphi_{n}\in \text{UCQ}_{r}^{\Lmc_{s}}$. Thus, as we assume that Condition~1 does not hold,
$\Kmc\models \varphi_{1}\vee \cdots \vee \varphi_{n}(b)$. Hence $\Amf\models  \varphi_{1}\vee \cdots \vee \varphi_{n}(b)$ and so there exists $i$
such that $\Amf\models \varphi_{i}(b)$. We have derived a contradiction.

Take an $\omega$-saturated model $\Bmf$ of $\Kmc$ satisfying $\Gamma$ in some $a\in P$. We have by definition $\Bmf,a^{\Bmf} \Rightarrow_{\text{CQ}_{r}^{\Lmc_{S}},\Sigma} \Amf,b^{\Amf}$.
By Lemma~\ref{lem:equivalence2},
$\Bmf,a^{\Bmf} \Rightarrow^{\text{mod}}_{\text{CQ}_{r}^{\Lmc_{S}},\Sigma} \Amf,b^{\Amf}$. This contradicts Condition~2, as required.
\end{proof}

\thmequival*

\begin{proof} \
	We first assume that $(\Lmc,\Lmc_{S})=(\ALCI,\ALCI)$. It suffices to show
	that Point~3 of Theorem~\ref{thm:L-modeltheory0} and Point~2 of Lemma~\ref{lem:helps} are equivalent.
	Assume first that Point~3 of Theorem~\ref{thm:L-modeltheory0} holds for $\Amf$. 
	For a model $\Cmf$ of $\Kmc$ we denote by $\Dmc_{\Sigma}^{\Cmf,a}$
the maximal connected component of $a^{\Cmf}$ in $\Cmf_{|\Dmc^{\Cmf}}$, 
where $\Dmc^{\Cmf}=\{c^{\Cmf} \mid c \in \text{ind}(\Dmc)\}$.

We show that for all models $\Bmf$ of $\Kmc$ and all $a\in P$: $\Bmf,a^{\Bmf}\not\rightarrow_{\Dmc_{\Sigma}^{\Bmf,a},\ALCI,\Sigma}\Amf,b^{\Amf}$. Then Point~2 of Lemma~\ref{lem:helps} holds for $n=|\Dmc|$.
For a proof be contradiction assume that there is a model $\Bmf$ of $\Kmc$ and an $a\in P$ such that $h:\Bmf,a^{\Bmf}\rightarrow_{\Dmc_{\Sigma}^{\Bmf,a},\ALCI,\Sigma}\Amf,b^{\Amf}$.
We aim to convert $\Bmf$ and $h$ into a new model $\Bmf'$ of $\Kmc$ and a functional bisimulation witnessing $\Bmf',a^{\Bmf'} \sim_{\Lmc,\Sigma}^{f} \Amf,b^{\Amf}$ and thus derive a contradiction.
To this end, we require the notion of a $k$-unfolding of a structure in which we do not only unfold into a tree-like structure but also take $k$ copies of every successor. In detail, we define the \emph{$k$-unfolding} $\Bmf^{\leq k}_{d}$ of a structure $\Bmf$ at $d\in \text{dom}(\Bmf)$ as follows, for any $k>0$. The domain of
$\Bmf^{\leq k}_{d}$ is the set $W$ of all words $w=d_{0}R_{0}(d_{1},i_{1}) \cdots R_{n-1}(d_{n},i_{n})$ such that $d_{0}=d$, $(d_{i},d_{i+1}) \in R_{i}^{\Bmf}$ for all $i<n$, and $i_{j}\leq k$ for all $j\leq n$, where all $R_{i}$ are roles. Let $\text{tail}(w)= d_{n}$. The interpretation of concept names and role names is as expected: we set $w\in A^{\Bmf^{\leq k}_{d}}$ if $\text{tail}(w)\in A^{\Bmf}$ and we set for $w_{1},w_{2}\in W$, $(w_{1},w_{2})\in R^{\Bmf^{\leq k}_{d}}$ if $w_{2}$ is obtained from $w_{1}$ by concatenating $w_{1}$ and some $R_{n+1}(d_{n+1},i_{n+1})$. 
	
	Now let $k$ be the maximum over the $\ALCI$-outdegrees of the
	nodes in $\Amf$. We define a new model $\Bmf'$ of $\Kmc$ by taking $\Bmf$, removing all nodes $d$ not in $\Dmc^{\Bmf}$ from it, and instead attaching $\Bmf_{d}^{\leq k}$ to $d$,
	for any $d\in \Dmc^{\Bmf}$. Now one can easily show that there is a functional $\ALCI(\Sigma)$-bisimulation $f$ between $\Bmf',a^{\Bmf'}$ and $\Amf,b^{\Amf}$:
	to define $f$ take the homomorphism $h$ and extend it with functional bisimulations witnessing $\Bmf_{d}^{\leq k},d \sim^{f}_{\ALCI,\Sigma} \Amf, h(d)$ for every $d\in \Dmc_{\Sigma}^{\Bmf,a}$.  

Assume now that Point 2 of Lemma~\ref{lem:helps} holds
for $\Amf$. We show that Point~3 of Theorem~\ref{thm:L-modeltheory0} holds for $\Amf$.
The proof is indirect. Assume Point~3 does not hold for $\Amf$. Thus, there exists a model $\Bmf$ of $\Kmc$ and $a\in P$ such that $\Bmf,a^{\Bmf} \sim_{\ALCI,\Sigma}^{f} \Amf,
b^{\Amf}$. Then we can regard the restriction of $f$ to any subset $D$ of $\text{dom}(\Bmf)$ as a $\Sigma$-homomorphism
$h$ for which clearly $\Bmf,c \sim_{\ALCI,\Sigma} \Amf,h(c)$ for all $c\in D$. 
Thus, $\Bmf,a^{\Bmf}\rightarrow_{D,\ALCI,\Sigma}\Amf,b^{\Amf}$. But then
Point~2 of Lemma~\ref{lem:helps} does not hold for $\Amf$.

\medskip

We now assume that $(\Lmc,\Lmc_{S})=(\ALC,\ALCO)$. Again it suffices to show
that Point~3 of Theorem~\ref{thm:L-modeltheory0} and Point~2 of Lemma~\ref{lem:helps} are equivalent.

Assume first that Point~3 of Theorem~\ref{thm:L-modeltheory0} holds for $\Amf$. 
We show that $\Amf$ witnesses Point~2 of Lemma~\ref{lem:helps}. The proof is indirect. Assume that for all $n>0$ there exists a model $\Bmf$ of $\Kmc$ and $a\in P$ such that for all $D\subseteq \text{dom}(\Bmf)$ of cardinality not
	exceeding $n$ such that the $\Sigma$-reduct of $\Amf_{|D}$ is $\mathcal{ALCO}$-rooted in $a^{\Bmf}$ we have  $\Bmf,a^{\Bmf}\rightarrow_{D,\mathcal{ALCO},\Sigma}
	\Amf,b^{\Amf}$. 
	
	Let $R$ denote the set of individuals $c\in \text{ind}(\Dmc) \cap \Sigma$ such that there is an $\ALC(\Sigma)$-path from $b^{\Amf}$ to $c^{\Amf}$ in $\Amf$.
	For any $c\in R$ let $n_{c}$ be the length of the shortest such path and let $m= \sum_{c\in R}n_{c}|\text{ind}(\Dmc)|$. Let $\Bmf$ 
	be a model of $\Kmc$ and $a\in P$ such that for all $D\subseteq \text{dom}(\Bmf)$ of cardinality not
	exceeding $m$ and such that the $\Sigma$-reduct of $\Bmf_{|D}$ is $\mathcal{ALCO}$-rooted in $a^{\Bmf}$ we have $\Bmf,a^{\Bmf}\rightarrow_{D,\mathcal{ALCO},\Sigma}
	\Amf,b^{\Amf}$.
	
	Now choose $D_{0}\subseteq \text{dom}(\Bmf)$ minimal such that 
	the $\Sigma$-reduct of $\Bmf_{|D}$ is rooted in $a^{\Bmf}$ and 
	$D_{0}$ contains all $c^{\Bmf}$ with $c\in \text{ind}(\Dmc)\cap \Sigma$ such that 
	there is an $\ALC(\Sigma)$-path from $a^{\Bmf}$ to $c^{\Bmf}$. Note that
	as $\Bmf,a^{\Bmf} \sim_{\mathcal{ALCO},\Sigma}\Amf,b^{\Amf}$
	the individuals we obtain are exactly those in $R$ and the cardinality
    of $D_{0}$ does not exceed $\sum_{c\in R}n_{c}$. Obtain $D$ from $D_{0}$
    by adding all nodes $c^{\Bmf}$ with $c\in \text{ind}(\Dmc)$ such that there
    exists an $\mathcal{ALC}(\Sigma)$-path from a node in $D_{0}$ through
    $\Dmc^{\Bmf}$ to $c^{\Bmf}$. Then the cardinality of $D$ does not exceed $m$.
    Thus, we have $\Bmf,a^{\Bmf}\rightarrow_{D,\mathcal{ALCO},\Sigma}
    \Amf,b^{\Amf}$. Let $h$ be the $\Sigma$-homomorphism witnessing this.
	
We define the \emph{directed $k$-unfolding omitting $\Sigma$-individuals}, $\Bmf^{d,\leq k}_{d}$, of a structure $\Bmf$ at $d\in \text{dom}(\Bmf)$ as follows, for any $k>0$. The domain of
$\Bmf_{d}$ is the set $W$ of all words $w=d_{0}r_{0}(d_{1},i_{1}) \cdots r_{n-1}(d_{n},i_{n})$ such that $d_{0}=d$, $(d_{i},d_{i+1}) \in r_{i}^{\Bmf}$ for all $i<n$, and $i_{j}\leq k$ for all $j\leq n$, where all $r_{i}$ are role names and $d_{n}$ does not interpret an individual name in $\Sigma$ if $r_{0},\ldots,r_{n-1}$ are all in $\Sigma$. Let $\text{tail}(w)= d_{n}$. The interpretation of concept names and role names is as expected: we set $w\in A^{\Bmf^{d,\leq k}_{d}}$ if $\text{tail}(w)\in A^{\Bmf}$ and we set for $w_{1},w_{2}\in W$, $(w_{1},w_{2})\in r^{\Bmf^{d,\leq k}_{d}}$ if $w_{2}$ is obtained from $w_{1}$ by concatenating $w_{1}$ and some $r_{n+1}(d_{n+1},i_{n+1})$.

Now let $k$ be the maximal $\mathcal{ALC}$-outdegree of a node in $\Amf$ and define a model $\Bmf'$ of $\Kmc$ by taking $\Bmf$, removing all nodes $d$ not in $\Dmc^{\Bmf}$ from it, and instead attaching $\Bmf_{d}^{d,\leq k}$ to $d$ for any $d\in \Dmc^{\Bmf}$. Moreover, add
$(w,c^{\Bmf})$ to the interpretation of $r$ if $(\text{tail}(w),c^{\Bmf})\in r^{\Bmf}$ and $c\in \Sigma$. 

Then one can show that there is a functional $\ALCO(\Sigma)$-bisimulation $f$ between $\Bmf',a^{\Bmf}$ and $\Amf,b^{\Amf}$ by taking the homomorphim $h$ and extend it with the functional bisimulations witnessing $\Bmf_{d}^{d,\leq k},d \sim^{f}_{\ALCO} \Amf, h(d)$ for every $d\in D$.  

The implication from Point~2 of Lemma~\ref{lem:helps} to Point~3 of Theorem~\ref{thm:L-modeltheory0} can be proved in the same way
as for $\mathcal{ALCI}$.
\end{proof}

\section{Proofs for Section~\ref{sec:complexity}}

\lemcereduction*

%
%
%

\begin{proof} 
  \ Assume $\Lmc$-ontologies $\Omc$ and $\Omc'$ are
  given. Let $\Sigma$ be the signature of $\Omc$. Let
  $\text{atom}_{\Sigma}$ denote the set of concepts $A$ with
  $A\in \Sigma\cap \NC$, $\{a\}$ with $a\in \Sigma\cap \NI$, and $\exists r.\top$ with
  $r\in \Sigma$. If $\Lmc$ admit inverse roles, then we also add
  $\exists r^{-}.\top$, for $r\in \Sigma$. We may assume that
  there exists a concept name $A\in \text{atom}_{\Sigma}$ such
  that $\Omc\models A \equiv \neg C$ for some $C\in
  \text{atom}_{\Sigma}$. Indeed, if no such $A$ exists, pick any
  $X\in \text{atom}_{\Sigma}$, add $A\sqsubseteq \neg X$, $X
  \sqsubseteq \neg A$ to $\Omc$ to obtain
  $\Omc_{1}$, and add $A$ to $\Sigma$. Then
  clearly $\Omc_{1}\cup \Omc'$ is a conservative extension of $\Omc_{1}$
  in $\Lmc$ (projectively or, respectively, non-projectively)
  iff  $\Omc\cup \Omc'$ is a conservative extension of $\Omc$ in
  $\Lmc$ (projectively or, respectively, non-projectively).

  We first consider the case $\Lmc=\ALCO$ which subsumes the case
  $\Lmc=\ALC$. The \emph{relativization} $C_{|A}$ of a concept $C$ to a concept name $A$ is defined by setting
  \begin{eqnarray*}
  	\top_{|A} & = & A\\
  	\bot_{|A} & = & \bot\\
   B_{|A} & = & B \sqcap A\\
   \{c\}_{|A} & = &\{c\} \sqcap A\\
   (\neg C)_{|A} & = & A \sqcap \neg (C_{|A})\\
   C \sqcap D)_{|A} & = & C_{|A} \sqcap D_{|A}\\
   (\exists R.C)_{|A} & = & A \sqcap \exists R.(A \sqcap C_{|A})
  \end{eqnarray*}
  The relativization of an inclusion $C \sqsubseteq D$ to $A$ is defined as
  $C_{|A} \sqsubseteq D_{|A}$. Observe that the relativization of an inclusion to a concept name $A$ is satisfied in $\Amf$ whenever $A^{\Amf}=\emptyset$. 
  Define the \emph{directed relativization}
  $\Omc^{A}$ of  $\Omc$ to a fresh concept name $A$ by
  relativizing all inclusions in $\Omc$ to $A$ and also adding $$
  \{c\} \sqsubseteq A, $$ for all $c\in \Sigma$, and $$ A \sqsubseteq
  \forall r.A $$ for all $r\in \Sigma$.
  Next define a database $\Dmc$ by taking fresh individual names $a$
  and $b$ used as the positive and negative example, a fresh concept
  name $D$, and a fresh role name $s$ and include in $\Dmc$: $A(b)$,
  $D(a)$, and $$ s(a,c), $$ for all individual names $c$ in
  $\Omc \cup \Omc'$. Next we take the directed relativization $(\Omc\cup \Omc')^{D'}$
  to a fresh concept name $D'$, but in this case instead of including
  the individual names in $\Omc\cup \Omc'$ in $D'$ we include $D'
  \sqsubseteq \forall s.D'$.

  Finally, we obtain $\Omc^{\ast}$ as the union of $\Omc^{A}$ and
  $(\Omc\cup \Omc')^{D'}$ and $$ D \sqcap E \sqsubseteq D', 
  $$ for all $E\in \text{atom}_{\Sigma}$. Let $\Kmc= (\Omc^{\ast},\Dmc)$. 

    \medskip \noindent\textit{Claim.} $\Omc \cup \Omc'$ is a conservative
  extension of $\Omc$ in $\mathcal{ALCO}$ iff $(\Kmc,\{a\},\{b\})$
  is not $\mathcal{ALCO}(\Sigma)$-separable, for both the projective
  and non-projective case.

  \medskip\noindent\textit{Proof of the Claim.} We consider the
  projective case. The non-projective case is similar and omitted.
  Consider an $\mathcal{ALCO}(\Sigma\cup
  \Sigma_{\text{help}})$-concept $C$, where $\Sigma_{\text{help}}$ is
  a set of fresh concept names.  We show the following equivalences:
  \begin{enumerate}[label=(\arabic*)]

    \item $C$ is satisfiable w.r.t.~$\Omc$ iff there exists a
      model $\Amf$ of $\Kmc$ such that $b^{\Amf}\in C^{\Amf}$.

    \item If $C$ is satisfiable
      w.r.t.~$\Omc\cup \Omc'$, then there exists a model $\Amf$ of $\Kmc$
      such that $a^{\Amf}\in C^{\Amf}$. 

    \item Let $E \in \text{atom}_{\Sigma}$. If there exists a model
      $\Amf$ of $\Kmc$ such that $a^{\Amf}\in (E \sqcap C)^{\Amf}$,
      then $E \sqcap C$ is satisfiable w.r.t.~$\Omc\cup \Omc'$.

  \end{enumerate}
  For (1), assume that $C$ is satisfiable w.r.t.~$\Omc$. Take a
  model $\Amf$ of $\Omc$ satisfying $C$ in $d$. We define a model
  $\Bmf$ of $\Kmc$ satisfying $C$ in $b^{\Amf}$: to define
  $\text{dom}(\Bmf)$, we add to $\text{dom}(\Amf)$ the individual $a$
  and all individuals $c\in \text{sig}(\Omc\cup \Omc')\setminus \Sigma$.
  Then we 
  set $b^{\Bmf}=d$ and interpret
  $a$ and the individuals in $\text{sig}(\Omc\cup \Omc') \setminus \Sigma$
  by themselves, interpret $A$ by the domain of $\Amf$, add the pairs
  $(a,c^{\Amf})$, $c\in \text{sig}(\Omc\cup \Omc')$ to $s^{\Bmf}$, and set
  $D^{\Bmf}=a$ and $D'^{\Bmf}=\emptyset$. Then $\Bmf$ is a model of
  $\Omc$ satisfying $C$ in $b^{\Bmf}$.	

  The converse direction of (1) is clear.

  For~(2), suppose that $C$ is satisfiable w.r.t.~$\Omc \cup \Omc'$.  Take a
  model $\Amf$ of $\Omc\cup \Omc'$ satisfying $C$ in $d$. We define a model
  $\Bmf$ of $\Kmc$ satisfying $C$ in $a^{\Bmf}$: take $\Amf$ and set
  $\text{dom}(\Bmf)= \text{dom}(\Amf)$. We interpret $b$ arbitrarily
  and set $A^{\Bmf}=D^{\Bmf}=D'^{\Bmf}= \text{dom}(\Bmf)$. Finally, we
  add the pairs  $(a^{\Amf},c^{\Amf})$, $c\in \text{sig}(\Omc\cup \Omc')$,
  to $s^{\Bmf}$.  Then $\Bmf$ is a model of $\Omc\cup \Omc'$ satisfying $C$
  in $a^{\Bmf}$.	

    For~(3), let $E\in \text{atom}_{\Sigma}$ and assume that $E \sqcap
    C$ is satisfied in a model $\Amf$ of $\Kmc$ at $a^{\Amf}$. Then
    $a^{\Amf}\in (D \sqcap E)^{\Amf}$ and therefore $a^{\Amf} \in
    D'^{\Amf}$ since $D \sqcap E \sqsubseteq D'\in \Omc\cup \Omc'$. Hence $E
    \sqcap C$ is satisfiable w.r.t.~$\Omc\cup \Omc'$ as it is satisfied in
    $\Amf$ w.r.t.~the directed relativization of $\Omc\cup \Omc'$ to $D'$.
    \medskip

    Now, we finish the proof of the Claim. Suppose first that
    $(\Kmc,\{a\},\{b\})$ is projectively
    $\mathcal{ALCO}(\Sigma)$-separable and let $C$ be a separating
    concept.  Then, there is a model \Amf of \Kmc such that $b\in
    (\neg C)^\Amf$. By Point~1, $\neg C$ is satisfiable w.r.t.\ $\Omc$.
    Moreover, there is no model $\Amf$ of \Kmc with $a\in (\neg
    C)^\Amf$. By Point~2, $\neg C$ is not satisfiable w.r.t.\
    $\Omc\cup \Omc'$. Hence, $\neg C$ is a witness concept for $\Omc,\Omc\cup \Omc'$.

    Conversely, assume that $\Omc\cup \Omc'$ is not a projective
    conservative extension of $\Omc$ in $\mathcal{ALCO}$ and let
    $C$ witness this. By our assumption on $\text{atom}_{\Sigma}$,
    there exists an $E\in \text{atom}_{\Sigma}$ such that $E \sqcap C$
    is also satisfiable w.r.t.~$\Omc$, but not satisfiable
    w.r.t.~$\Omc\cup \Omc'$. Thus, by Point~1, there exists a model $\Amf$
    of $\Kmc$ such that $b^{\Amf}\in (E\sqcap C)^{\Amf}$ and, by
    Point~3, there does not exist a model $\Amf$ of $\Kmc$ such that
    $a^{\Amf}\in (E\sqcap C)^{\Amf}$. Thus, $(\Kmc,\{a\},\{b\})$ is
    projectively $\mathcal{ALCO}(\Sigma)$-separable, namely by $\neg
    (E\sqcap C)$.

    \smallskip This finishes the proof of the Claim.

    \medskip
  The proof above is easily adapted for $\mathcal{ALCI}$ and
  $\mathcal{ALCIO}$. In fact, one only has to add to the directed
  relativization of $\Omc^{A}$ the inclusions $A \sqsubseteq
  \forall r^-.A$, for $r$ any role in $\Sigma$.  
\end{proof}

\subsubsection{Preliminaries for Tree Automata}

A \emph{tree} is a non-empty (and potentially infinite) set of words
$T \subseteq (\mathbb{N} \setminus 0)^*$ closed under prefixes.  A
node $w\in T$ is a \emph{successor} of $v\in T$ if $w=v\cdot i$ for
some $i\in\mathbb{N}$. Moreover, $w$ is an \emph{ancestor} of $v$ if
$w$ is a prefix of $v$. A tree is binary if every node has either zero
or two successors.  For an alphabet $\Theta$, a \emph{$\Theta$-labeled
tree} is a pair $(T,\tau)$ with $T$ a tree and $\tau:T \rightarrow
\Theta$ a node labeling function.

A {\em two-way alternating tree automaton (2ATA) over binary trees} is
a tuple $\Amc = (Q,\Theta,q_0,\delta,\Omega)$ where $Q$ is a finite
set of {\em states}, $\Theta$ is the finite {\em input alphabet},
$q_0\in Q$ is the {\em initial state}, $\delta$ is a {\em transition
function} as specified below, and $\Omega:Q\to \mathbb{N}$ is a {\em
priority function}. The transition function maps a state $q$ and some
input letter $\theta\in \Theta$ to a \emph{transition condition
$\delta(q,\theta)$} which is a positive Boolean formula over the truth
constants $\mn{true}$ and $\mn{false}$ and transitions of the form
$q$, $\langle-\rangle q$, $[-] q$, $\Diamond q$, $\Box q$ where $q \in
Q$.  Informally, the transition $q$ expresses that a copy of the
automaton is sent to the current node in state $q$, $\langle - \rangle
q$ means that a copy is sent in state $q$ to the predecessor node,
which is then required to exist, $[-] q$ means the same except that
the predecessor node is not required to exist, $\Diamond q$ means that
a copy is sent in state $q$ to some successor, and $\Box q$ that a
copy is sent in state $q$ to all successors. The semantics is defined
in terms of runs in the usual way~\cite{DBLP:conf/icalp/Vardi98}, see
below. We use $L(\Amc)$ to denote the set of all $\Theta$-labeled
binary trees accepted by \Amc.  The \emph{emptiness problem}, which
asks whether $L(\Amc)=\emptyset$ for a given 2ATA \Amc, can be decided
in time exponential in the number of states
of~\Amc~\cite{DBLP:conf/icalp/Vardi98}.

We make precise the semantics of 2ATAs. Let
$\Amc=(Q,\Theta,q_0,\delta,\Omega)$ be a 2ATA and $(T,L)$ a
$\Theta$-labeled tree. A {\em run for \Amc on $(T,L)$} is a $T\times
Q$-labeled tree $(T_r,r)$ such that:
\begin{itemize}

  \item $\varepsilon\in T_r$ and $r(\varepsilon)=(\varepsilon,q_0)$;

  \item For all $y\in T_r$ with $r(y)=(x,q)$ and $\delta(q,L(x))=\vp$,
    there is an assignment $v$ of truth values to the transitions in $\vp$
    such that $v$ satisfies $\vp$ and:
    \begin{itemize}

      \item if $v(p)=1$, then $r(y')=(x,p)$ for some successor
	$y'$ of $y$ in $T_r$;

      \item if $v(\langle - \rangle p)=1$, then $x \neq \varepsilon$
	and there is a successor $y'$ of $y$ in $T_r$ with
	$r(y')=(x\cdot -1,p)$;

      \item if $v([-] p)=1$, then $x=\varepsilon$ or there is a
	successor $y'$ of $y$ in $T_r$ such that 
	$r(y')=(x\cdot -1,p)$;

      \item if $v(\Diamond p)=1$, then there is a successor $x'$ of
	$x$ in $T$ and a successor $y'$ of $y$ in $T_r$ such that
	$r(y')=(x',p)$;

      \item if $v(\Box p)=1$, then for every successor $x'$ of
	$x$ in $T$, there is a successor $y'$ of $y$ in $T_r$ such that
	$r(y')=(x',p)$.

    \end{itemize}

\end{itemize}
Let $\gamma=i_0i_1\cdots$ be an infinite path in $T_r$ and denote, for
all $j\geq 0$, with $q_j$ the state such that $r(i_0\cdots
i_j)=(x,q_j)$. The path $\gamma$ is {\em accepting} if the largest
number $m$ such that $\Omega(q_j)=m$ for infinitely many $j$ is even.
A run $(T_r,r)$ is accepting, if all infinite paths in $T_r$ are
accepting. Finally, a tree is accepted if there is some accepting run
for it.

\medskip It is well-known that 2ATAs are closed under complementation,
intersection, and projection. 
\begin{lemma} \label{lem:closureproperties}
  Given 2ATAs $\Amc_1,\Amc_2$ over alphabet $\Theta$ and a mapping
  $h:\Theta\to\Theta'$, we can compute:
  \begin{itemize}

  \item in polynomial time a 2ATA $\overline{\Amc_1}$ such that
    $L(\overline{\Amc_1})=\overline{L(\Amc_1)}$ and the number of
    states of $\overline{\Amc_1}$ equals the number of states of
    $\Amc_1$;

  \item in polynomial time a 2ATA $\Amc$ such that
    $L(\Amc)=L(\Amc_1)\cap L(\Amc_2)$ and the number of states of
    $\Amc$ is $1+n_1+n_2$, $n_i$ the number of states of $\Amc_i$; 

  \item in exponential time a 2ATA $\Amc_h$ such that
    $L(\Amc_h)=\{(T,h(\tau))\mid (T,\tau) \in L(\Amc_1)\}$ and the
    number of states of $\Amc_h$ is exponential in the number of
    states of $\Amc_1$.

\end{itemize}
\end{lemma}

\subsubsection{Encoding of \Lmc-Forest Models}

Let $\Lmc\in \{\ALCI,\ALCO\}$ and fix an \Lmc-KB $\Kmc=(\Omc,\Dmc)$.
In order to work with tree automata, we need to encode \Lmc-forest
models of \Kmc of finite \Lmc-outdegree as input to the tree automata.
Since 2ATAs run over binary trees, we need to appropriately encode the
arbitrary outdegree, which is done similar to~\cite{JLZ-KR20}.

More precisely, we use the alphabet $\Theta$ defined by
\[\Theta = \{\circ\}\cup \text{Rol}(\Kmc)\times
2^{\text{ind}(\Dmc)\cup (\text{sig}(\Kmc)\cap \mn{N_C})}\times
2^\Fmc,\]
where $\text{Rol}(\Kmc)$ denotes the set of all role names that occur
in \Kmc and their inverses, and $\Fmc$ is the set of all pairs $(r,a)$
with $r$ a role name that in $\text{sig}(\Kmc)$ and $a\in
\text{ind}(\Dmc)$.  Intuitively, a node $w\in T$ with
$\tau(w)=(R,M,F)$ encodes an element that satisfies precisely the
concepts in $M$; moreover, $F$ describes its connections to 
elements in $\text{ind}(\Dmc)$ and $R$ is the ``incoming role''.
The symbol `$\circ$' is a label for dummy nodes that
we need for encoding arbitrary finite outdegree into binary trees: we simply
introduce as many intermediate $\circ$-labeled nodes as needed to
achieve the required outdegree at each node. 

More formally, let $(T,\tau)$ be a $\Theta$-labeled tree. 
For each $w \in T$ with
$\tau(w) \neq \circ$, let $w^\uparrow$ denote the unique ancestor $w'$
of $w$ in $T$
(if existing) such that $\tau(w') \neq
\circ$ and $\tau(w'')=\circ$ for all $w''$ between $w'$ and $w$.
We call $(T,\tau)$ \emph{well-formed} if 
\begin{enumerate}

  \item for every
    $a\in\text{ind}(\Dmc)$, there is a unique element $w_{a}
    \in T$ such that $\tau(w_a)=(R,M,F)$ for some $R,F$ and $a\in M$;

  \item for every $w\in T$ with $\tau(w)\neq \circ$, either $w=w_a$,
    for some $a$, and all ancestors of $w$ are labeled with $\circ$,
    or $w$ has an ancestor $w_a$, for some $a$.

\end{enumerate}
%
%
A well-formed $\Theta$-labeled tree $(T,\tau)$ gives rise to a
structure $\Amf_\tau$ with $\text{dom}(\Amf_\tau) = \{ w \in T
\mid \tau(w)\neq \circ\}$ as follows: 
$$ \begin{array}{r@{\;}c@{\;}l}
  a^{\Amf_\tau} &=& w_{a} \\[1mm]
  A^{\Amf_\tau} &=& \{ w\in T \mid \tau(w)=(S,M,F) \text{ and } A \in
M \} \\[1mm]
  r^{\Amf_\tau} &=& \{ (w^\uparrow,w) \mid
  \tau(w)=(r,M,F)\text{ and }
    w^\uparrow\text{ defined}\} \cup{} \\[1mm] 
    && \{ (w,w^\uparrow) \mid \tau(w)=(r^-,M,F)\text{ and }w^\uparrow\text{ defined}\} \cup{} \\[1mm] 
    && \{ (w,w_a) \mid \tau(w)=(S,M,F)\text{ and } (r,a)\in
  F\} 
    %
  %
\end{array} $$
for all $a \in \mn{N_I}$, $A \in \mn{N_C}$, and $r \in \mn{N_R}$ (and
for $a\in\mn{N_I}\setminus \text{ind}(\Dmc)$, $w_a$ denotes an
arbitrary element of $\Amf$).
 
Conversely, every finite outdegree forest model \Amf of \Kmc can be encoded
(up to isomorphism) as well-formed $\Theta$-labeled tree. 
To show this, we
start with a not-necessarily binary well-formed $\Theta$-labeled tree
$(T,\tau)$ that encodes $\Amf$; $(T,\tau)$ can easily be made binary
by introducing intermediate $\circ$-labeled nodes. Let $D=\{a^\Amf\mid
  a\in \text{ind}(\Dmc)\}$ and associate sets $M_d,F_d$ to every
  element $d\in \text{dom}(\Amf)$ by taking
$$ \begin{array}{r@{\;}c@{\;}l}
  M_d &=& \{A\in\text{sig}(\Kmc)\mid d\in A^\Amf\}\cup \{a\in
    \text{ind}(\Dmc) \mid d=a^\Amf\} \\[1mm]
  F_d &=& \{(r,e)\mid (d,e)\in r^\Amf, e\in D\} 
\end{array}$$
To start the construction of $(T,\tau)$, we set $\tau(\varepsilon)=\circ$ and
add a successor $w_d$ of $\varepsilon$ for every $d\in D$ and label it
with $\tau(w_d)=(S,M_d,F_d)$ for an arbitrary role name $S$.

For the rest of the construction, let $\Amf_d,d\in D$ be the \Lmc-trees
which exist since $\Amf$ is \Lmc-forest model of \Kmc.  Recall that
$\Amf_d$ is rooted at $d$. Now $(T,\tau)$ is obtained by exhaustively applying the
following rule:
\begin{itemize}
  \item[$(\ast)$] If $w_e$ is defined, for an
element $e$ of some $\Amf_d$ and $f$ is a successor of $e$ in $\Amf_d$
with $w_f$ undefined, add a fresh successor $w_f$ of $w_e$ to $T$ and
set $\tau(w_f)=(R,M_f,F_f)$ where $R$ is the unique role such
that $(e,f)\in R^\Amf$. 
\end{itemize}

It is not difficult to construct 2ATAs that accept precisely the
forest models of \Kmc, see for
example~\cite{JLMSW17} for full details of a
similar automata construction.

\begin{lemma}\label{lem:automata1}
  For $\Lmc\in\{\ALC,\ALCI,\ALCO\}$ and \Lmc-KBs \Kmc,
  we can construct in time polynomial in $||\Kmc||$
  2ATAs $\Amc_0,\Amc_\Kmc$ such that: 
  \begin{enumerate}

    \item $\Amc_0$ accepts precisely the well-formed $\Theta$-labeled
      trees;

    \item $\Amc_\Kmc$ accepts a well-formed $\Theta$-labeled tree
      $(T,\tau)$ iff $\Amf_\tau$ is a finite outdegree forest model of
      \Kmc.

  \end{enumerate}
\end{lemma}

\subsection{\ALCI and \ALC}

We concentrate on \ALCI, the case of \ALC is similar. 

\lemequivalci*
 
\begin{proof}
  \ 
  ``if''. Take a forest model $\Amf$ and some $a\in P$ with $\Dmc_{\text{con}(a)},a
  \rightarrow^{\Sigma}_{c} \Amf,b^{\Amf}$.
  Moreover, fix a $\Sigma$-homomorphism $h$ and
  \Kmc-types $t_{d}$, $d\in
  \text{ind}(\Dmc)$ that witness that.
  Take models $\Bmf_{d}$ of $\Omc$ such that
  $\Bmf_{d},d \sim_{\ALCI,\Sigma} \Amf,h(d)$. We may assume that the
  $\Bmf_{d}$ are tree-shaped with root $d$ and that the bisimulations
  are functions $f_{d}$.  Now attach to every $d\in \text{ind}(\Dmc)$
  the model $\Bmf_{d}$ and obtain $\Bmf$ by adding $(d,d')$ to
  $r^{\Bmf}$ if $r(d,d')\in \Dmc$. Then \[f= \bigcup_{d\in
  \text{ind}(\Dmc)}f_{d}\] is a functional
  $\ALCI(\Sigma)$-bisimulation between $\Bmf$ and $\Amf$.

  ``only if''. Take a forest model $\Amf$, some $a\in P$, and a model
  $\Bmf$ of \Kmc such that $\Bmf,a^{\Bmf}\sim_{\ALCI,\Sigma}^{f}
  \Amf,b^{\Amf}$.  Let $f$ be a functional
  $\ALCI(\Sigma)$-bisimulation witnessing that. The restriction
  $h$ of $f$ to $\text{ind}(\Dmc)$ and types $t_d=\text{tp}_\Kmc(\Bmf,d^\Bmf)$,
  for all $d\in\text{ind}(\Dmc)$ witness $\Dmc_{\text{con}(a)},a
  \rightarrow^{\Sigma}_{c} \Amf,b^{\Amf}$.
\end{proof}

\begin{lemma} \label{lem:automata-alci} 
  Let $\Lmc\in\{\ALC,\ALCI\}$.  For each $a\in P$, there is a 2ATA
  $\Amc_a$ such that $\Amc_a$ accepts a well-formed labeled tree
  $(T,\tau)$ iff $\Dmc_{\text{con}(a)},a \rightarrow^{\Sigma}_{c}
  \Amf,b^{\Amf}$.  Moreover, $\Amc_a$ can be constructed in double
  exponential time in $||\Kmc||$ and has exponentially many states.
\end{lemma}

\begin{proof}
  \ We sketch the construction of the automata $\Amc_a$, $a\in P$. Let
  $(T,\tau)$ be a well-formed input tree.
  As the first step, $\Amc_a$ non-deterministically guesses the following: 
  \begin{itemize}

    \item types $t_d$, $d\in\text{ind}(\Dmc_{\text{con}(a)})$, such
      that $(\Omc,\Dmc')$ is satisfiable where  $\Dmc'=\Dmc\cup \{
	C(d) \mid C \in t_d, \ d\in
	\text{ind}(\Dmc_{\text{con}(a)})\}$;

    \item a partition $\Dmc_0,\Dmc_1,\ldots,\Dmc_m$ of
      $\Dmc_{\text{con}(a)}$ such that
      $\text{ind}(\Dmc_i)\cap\text{ind}(\Dmc_j) = \emptyset$ for
      $1\leq i<j\leq m$;
	
    \item a mapping $h$ from $\text{ind}(\Dmc_0)$ to
      $\text{ind}(\Dmc)$ such that $h(a)=b$ and there are
      $a_1,\dots,a_m$ with $h(c)=a_i$
      for all $c \in \text{ind}(\Dmc_0 \cap \Dmc_i)$ and
      $1 \leq i \leq m$.


  \end{itemize}
  The entire guess is stored in the state of the automaton. Note that
  the first item above ensures that Item~(ii) from the definition of
  $\Dmc_{\text{con}(a)},a \rightarrow^{\Sigma}_{c}
  \Amf_\tau,b^{\Amf_\tau}$ is satisfied for the guessed types
  $t_d$. Moreover, elements $d^\Bmf$ for $d\in\text{ind}(\Dmc_0)$ are intuitively mapped to
  $h(d)^\Amf$ while elements $d^\Bmf$ for
  $d\in\text{ind}(\Dmc_i)\setminus\text{ind}(\Dmc_0)$ are mapped to
  the trees below $a_i$.

  After making its guess, the automaton verifies first $h$ is a
  $\Sigma$-homomorphism from $\Dmc_0$ to $\Amf_\tau$ by sending out
  copies for all facts in $\Dmc_0$ to the respective elements in
  $\Amf_\tau$. Then, it verifies that  this homomorphism can be
  extended to a homomorphism from $\Dmc_{\text{con}(a)}$ to
  $\Amf_\tau$ that satisfies Item~(i) from the definition of
  $\Dmc_{\text{con}(a)},a \rightarrow^{\Sigma}_{c}
  \Amf_\tau,b^{\Amf_\tau}$. To this end, it does a top-down traversal
  of $\Amf_\tau$ checking that each $\Dmc_i$ can be homomorphically
  mapped to the subtree of $\Amf_\tau$ below $h(a_i)^{\Amf_\tau}$. During the
  traversal, the automaton memorizes in its state the set of
  individuals from $\Dmc_i$ that are mapped to the currently visited
  element.

The automaton additionally makes sure that Item~(i) from the definition
of
$\Dmc_{\text{con}(a)},a \rightarrow^{\Sigma}_{c}
\Amf_\tau,b^{\Amf_\tau}$ is satisfied, in the following way. During
the top-down traversal, it spawns copies of itself to verify that,
whenever it has decided to map a
$d \in \text{ind}(\Dmc_{\text{con}(a)})$ to the current element, then
there is a tree-shaped model $\Bmc_d$ of \Omc with
$\text{tp}_{\Kmc}(\Bmf_{d},d) = t_{d}$ and a bisimulation that
witnesses $\Bmc_d,d \sim_{\ALCI,\Sigma} \Amf_\tau,c$. This is done by
`virtually' traversing $\Bmc_d$ elements-by-element, storing at each
moment only the type of the current element in a state. This is
possible because $\Bmf_d$ is tree-shaped. At the beginning, the
automaton is at an element of $\Bmc_d$ of type $t_d$ and knows that
the bisimulation maps this element to the node of $\Amf_\tau$
currently visited by the automaton. It then does two things to verify
the two main conditions of bisimulations.  First, it transitions to
every neighbor of the node of $\Amf_\tau$ currently visited, both
upwards and downwards, and carries out in its state the corresponding
transition in $\Bmc_d$, in effect guessing a new type. Second, it
considers the current type of $\Bmc_d$ and guesses successor types
that satisfy the existential restrictions in it. For every required
successor type,  it then guesses a neighbor of the currently visited
node in $\Amf_\tau$ to which the successor is mapped. The two steps
are alternated, exploiting the alternation capabilities of the
automaton. Some extra bookkeeping in states is needed for the root
node of the input tree as it represents more than one element
of $\Amf_\tau$.


It can be verified that only exponentially many states are required
and that the transition function can be computed in double exponential
time in $||\Kmc||$. 
\end{proof}

We can now finish the proof of the upper bound in
Theorem~\ref{thm:complexity-alci}. By
Lemmas~\ref{lem:automata1},~\ref{lem:equiv},
and~\ref{lem:automata-alci} and Theorem~\ref{thm:L-modeltheory0},
$(\Kmc,P,\{b\})$ is projectively $\Lmc(\Sigma)$-separable iff
\[L(\Amc_0)\cap L(\Amc_\Kmc)\cap \textstyle\bigcap_{a\in
P}\overline{L(\Amc_a)}\neq\emptyset\] where $\overline{L(\Amc_a)}$
denotes the complement of $L(\Amc_a)$. By Lemmas~\ref{lem:automata1}
and \ref{lem:automata-alci} all these automata can be constructed in
double exponential time and their number of states is single
exponential in $||\Kmc||$ for $\Lmc\in\{\ALC,\ALCI\}$.  By
Lemma~\ref{lem:closureproperties}, we can compute in polynomial time
an automaton that accepts precisely the language on the left-hand side
of the above inequality. It remains to recall that non-emptiness of
2ATAs can be decided in time exponential in the number of states.

\subsection{Upper Bound for \ALCO}

We give here only the upper bound for
Theorem~\ref{thm:complexity-alco}. The lower bound is proved for
conservative extensions in Section~\ref{sec:lower-alco}.

We use the same definition of \Kmc-types as in
\ALCI except that we assume without loss of generality that $\{c\}\in
\text{sub}(\Kmc)$, for all $c\in\text{ind}(\Dmc)$. Denote with
$\text{TP}(\Kmc)$ the set of all types. 

We work with $(\Theta\times\Theta')$-labeled trees for 
\[\Theta' = \{\circ\} \cup 2^{\text{ind}(\Dmc)}\]
A $(\Theta\times\Theta')$-labeled tree $(T,\tau_1,\tau_2)$ is
\emph{well-formed} if $(T,\tau_1)$ is well-formed and for each
$c\in\text{ind}(\Dmc)$, there is exactly one $w\in T$ with $c\in
\tau_2(w)$; we denote this with $v_c$.

\begin{lemma} \label{lem:alco-wellformed}
  There is a 2ATA $\Amc_0'$ which accepts precisely the well-formed
  $(\Theta\times\Theta')$-labeled trees. 
\end{lemma}

\begin{lemma} \label{lem:automata-alco}
  For every $a\in P$, there is a 2ATA $\Amc_a'$ which accepts a
  well-formed $(\Theta\times\Theta')$-labeled tree $(T,\tau_1,\tau_2)$
  iff there is a forest model $\Bmf$ of \Kmc and a
  function $f:\text{dom}(\Bmf)\to \text{dom}(\Amf_{\tau_1})$ such that
  \begin{enumerate}[label=(\alph*)]

    \item $f$ witnesses $\Bmf,a^\Bmf\sim_{\ALCO,\Sigma}^f
      \Amf_{\tau_1},b^{\Amf_{\tau_1}}$;

    \item for every $c\in \text{ind}(\Dmc)$, we have $f(c^\Bmf)=v_c$. 

  \end{enumerate}
  Moreover, $\Amc_a'$ can be
  constructed in exponential time in $||\Kmc||$ and has at most
  exponentially many states.
\end{lemma}


\begin{proof}
  \ Let $(T,\tau_1,\tau_2)$ be a well-formed
  $(\Theta\times\Theta')$-labeled tree. We sketch the function of
  $\Amc'_a$. Intuitively, $\Amc_a'$ constructs an \ALCO-forest model
  \Bmf of \Kmc and the witnessing bisimulation $f$ ``on the fly'' by
  visiting the nodes in the input in states that store the type of the
  currently visited element of \Bmf. The construction of \Bmf is
  started at the individual names, that is, at the nodes $v_c\in T$,
  $c\in \text{ind}(\Dmc)$.  Before the actual construction of \Bmf can
  start, $\Amc_a'$ guesses the types of the individual names $c\in
  \text{ind}(\Dmc)$ in the model $\Bmf$ and keeps that guess in its
  state throughout the entire run. Then it spawns a copy of itself in
  every $v_c\in T$, $c\in\text{ind}(\Dmc)$ in state corresponding to
  $t_c$.

  Whenever $\Amc_a'$ visits a node $w\in T$ in some type $t$, this
  represents an obligation to extend the bisimulation constructed so
  far to an element $e$ of type $t$ in \Bmf and the element $w$ of
  $\Amf_{\tau_1}$. This can be done similarly to what is done in the
  proof of Lemma~\ref{lem:automata-alci} since $\Bmf$ is
  forest-shaped. We need to argue, however, how to deal with the
  individuals since they can be accessed from everywhere in \Bmf. When
  the automaton tries to extend \Bmf with a successor that satisfies
  some nominal $\{c\}$, we have to make sure that that successor gets
  type $t_c$. Conversely, when we try to find a type bisimilar to some
  successor of $w$, we make sure that we can only visit nodes of the
  form $v_c$ with types $t_c$. In both cases, it is crucial that the
  automaton can stop the extension, because we have already started a
  copy of the automaton with type $t_c$ in $v_c$ (in the very
  beginning).
\end{proof}

We finish the proof of the upper bound of
Theorem~\ref{thm:complexity-alco}. By
Lemma~\ref{lem:closureproperties}, we can compute in exponential time
(in the size of $\Amc_{a}'$)  an automaton $\Amc_a$ with
\[L(\Amc_a) = \{(T,\tau_1)\mid (T,\tau_1,\tau_2)\in L(\Amc_a')\},\]
that is, $\Amc_a$ is the projection of $\Amc_a'$ to the first
component $\tau_1$. It should be clear that $\Amc_a$ accepts a forest
model \Amf of \Kmc iff there is a model $\Bmf$ of \Kmc such that
$\Bmf,a^{\Bmf}\sim^f_{\ALCO,\Sigma}\Amf,b^{\Amf}$. Thus, we can
proceed as for \ALCI and compute the desired automaton $\Amc$ such
that it accepts the language
\[L(\Amc)=L(\Amc_0)\cap L(\Amc_\Kmc)\cap \textstyle\bigcap_{a\in
P}\overline{L(\Amc_a)},\] 
which is non-empty iff $(\Kmc,P,\{b\})$ is projectively
$\ALCO(\Sigma)$-satisfiable. Since $\Amc_a$ (and thus \Amc) is a 2ATA
with double exponential many states and non-emptiness can be checked
in exponential time in the number of states, the \ThreeExpTime-upper
bound follows. 

\section{\ThreeExpTime Lower Bound for Conservative Extensions in
\ALCO}\label{sec:lower-alco}

We start with model theoretic characterizatons of conservative
extensions and projective conservative extensions in \ALC.  We define
a model $\Amf$ of an $\ALCO$-ontology $\Omc$ to be a \emph{forest
  model} of $\Omc$ of finite outdegree if it is an $\ALCO$-forest
model of the KB $\Kmc=(\Omc,\Dmc)$ of finite $\ALCO$-outdegree, where
$\Dmc= \{D_{b}(b) \mid b\in \text{ind}(\Omc)\}$ and the $D_{b}$ are
fresh concept names, one for each individual $b$ in $\Omc$.
  	
\begin{theorem}
	Let $\Omc$ and $\Omc'$ be $\ALCO$-ontologies and $\Sigma
       = \text{sig}(\Omc)$. Then 
        \begin{enumerate}
        \item 
the following conditions are equivalent:
	\begin{enumerate}
        \item $\Omc \cup \Omc'$ is a conservative extension
          of $\Omc$;
		\item for every pointed forest model $\Amf,a$ of $\Omc$ of finite outdegree there exists a pointed model $\Bmf,b$ of $\Omc\cup \Omc'$ such that 
		$\Bmf,b \sim_{\ALCO,\Sigma} \Amf,a$.
	\end{enumerate}
\item the following conditions are equivalent:
\begin{enumerate}
	\item $\Omc \cup \Omc'$ is a projective conservative extension of $\Omc$;
          
	\item for every pointed forest model $\Amf,a$ of $\Omc$ of finite outdegree there exists a pointed model $\Bmf,b$ of $\Omc\cup \Omc'$ such that 
	$\Bmf,b \sim_{\ALCO,\Sigma}^{f} \Amf,a$.
        \end{enumerate}
\end{enumerate}
\end{theorem}
%
%
%
\begin{theorem}
\label{thm:threeexpalco}
Given an \ALC-ontology $\Omc$ and an \ALCO-ontology $\Omc'$, it is
\ThreeExpTime-hard to decide whether $\Omc'$ is a conservative
extension of $\Omc$.
\end{theorem}
Since \Omc is an \ALC-ontology, we have no nominals available in
witness concepts and every witness concept must actually be an \ALC-concept.
It is 
interesting to note that we use only a single nominal in the
ontologies $\Omc'$.

The proof of Theorem~\ref{thm:threeexpalco} follows the general
outline of the \TwoExpTime lower bound for conservative
extensions in \ALC that is proved in \cite{DBLP:conf/kr/GhilardiLW06}.
%
As in that paper, we proceed in two steps. We first establish a lower
bound on the size of witness concepts and then extend the involved
ontologies to obtain the \ThreeExpTime lower bound.  In fact, the
first
step is the technically subtle one and we present it in full detail.
The second step is then rather simple and we only sketch the required
constructions. 

\subsection{Large Witness Concepts}
\label{app:largewitness}
\begin{theorem}
\label{thm:succinct}
  For every $n > 0$, there is an \ALC-ontology $\Omc_n$ and an
  \ALCO-ontology $\Omc'_n$ such that the size of $\Omc_n$ and
  $\Omc'_n$ is polynomial in $n$, $\Omc_n \cup \Omc'_n$ is not a
  conservative extension of $\Omc_n$, and every witness concept $C$
  for $\Omc_n$ and $\Omc'_n$ is of size at least $2^{2^{2^{2^n}}}$.
\end{theorem}
%
%
%
%
To prove Theorem~\ref{thm:succinct}, we craft $\Omc_n$ and $\Omc'_n$
so that $\Omc_n \cup \Omc'_n$ is not a conservative extensions of
$\Omc_n$, but every witness concept $C$ for $\Omc_n$ and $\Omc'_n$
must enforce in models of $\Omc_n$ a binary tree of depth at least
$2^{2^{2^n}}$.  This implies that $C$ is of size at least $2^m$
because $C$ is an \ALC-concept and $\Omc_n$ does not introduce
additional elements via existential restrictions. Let us be a bit more
precise about the trees. The `nodes' of the tree are actually paths of
length $2^n \cdot 2^{2^n}$ and no branching occurs inside these
paths. Thus, when counting also the intermediate domain element
on the `node paths', then the trees are really of depth
$2^n \cdot 2^{2^n} \cdot 2^{2^{2^n}}$.  We use a single role name
$r$ to attach successors, both when branching occurs and when no
branching occurs. At branching nodes, the left successor is marked
with concept name $S_L$ and the right successor is marked with concept
name $S_R$. Successors of non-branching nodes must be marked with at
least one of $S_L$ and $S_R$. And finally, the (first element of the)
root (path) is labeled with concept name $A$.

A main feature of $\Omc_n$ and $\Omc'_n$ is to implement a binary
counter that can count up to $m:=2^{2^{2^n}}$, the desired depth of the
trees (not counting intermediate domain elements). In fact, $\Omc_n$
and $\Omc'_n$ implement three binary counters that build upon each
other so that the third counter can achieve the intended counting
range.  Counter~1 has $n$ bits and counts from $0$ to $2^n-1$,
Counter~2 has $2^n$ bits and counts up to ${2^{2^n}}-1$, and Counter~3
has $2^{2^n}$ bits and counts up to ${2^{2^{2^n}}}-1$. We use
Counter~1 to describe the bit positions of Counter~2 and Counter~2 to
describe the bit positions of Counter~3.  Counter~1 and Counter~2
count modulo their maximum value plus one while Counter~3 needs to
reach its maximum value only once.

Counter~1 uses concept names $C_0,\dots,C_{n-1}$ as bits.  Thus a
counter value of Counter~1 can be represented at a single domain
element. In contrast, a value of Counter~2 is spread out accross a
sequence of $2^n$ domain elements, which we call a \emph{Counter~2
  sequence}. The bit positions of Counter~2 in such a sequence are
identified by Counter~1 and concept name $X_2$ is used to indicate the
bit value of Counter~2 at each position. A value of Counter~3, in
turn, is spread out accross a sequence of $2^{2^n}$ Counter~2
sequences, that is, $2^n \cdot 2^{2^n}$ domain
elements in total. We call this a \emph{Counter~3 sequence} and
it is these Counter~3 sequences that constitute the `nodes' of
the trees mentioned above. Each of
the Counter~2 subsequences represents one bit position of Counter~3 and stores
one bit value via the concept name $X_3$ that must be
interpreted uniformly in that subsequence. 

With a \emph{path} in an structure \Amf, we mean a sequence
$p=d_0,\dots,d_n$ of elements of $\text{dom}(\Amf)$ such that
$(d_i,d_{i+1}) \in r^\Amf$ and $d_{i+1} \in S_L^\Amf \cup S_R^\Amf$
for all $i < n$. We say that such a path is \emph{properly counting}
if the concept names $C_0,\dots,C_{n-1},X_2,X_3$ are interpreted along
the path in accordance with the counting strategy outlined above, all
three counters starting with counter value zero. We now formalize the
trees described above. A \emph{counting tree} in \Amf is a collection
$T$ of (not necessarily distinct) domain elements $d_{w,i}$,
$w \in \{L,R\}^*$ with $|w|<m$ and
$0 \leq i < 2^n \cdot 2^{2^n}$, such that the following conditions are
satisfied:
\begin{enumerate}

\item $d_{\varepsilon,0} \in A^\Amf$,

\item $(d_{w,i},d_{w,i+1}) \in r^\Amf$ and $d_{w,i+1} \in S_L^\Amf
  \cup S_R^\Amf$
  for all $d_{w,i+1} \in T$;

\item $(d_{w, 2^n \cdot 2^{2^n}},d_{wL,0}) \in r^\Amf$ and $d_{wL,0}
  \in S_L^\Amf$
  for all \mbox{$d_{wL,0} \in T$}; 

\item $(d_{w, 2^n \cdot 2^{2^n}},d_{wR,0}) \in r^\Amf$ and $d_{wR,0}
  \in S_R^\Amf$
  for all \mbox{$d_{wR,0} \in T$}; 

\item every path in \Amf that uses only elements from $T$ is properly
  counting.

\end{enumerate}
The element $d_{\varepsilon,0}$ is the \emph{root} of the counting
tree.  We say that \Amf is \emph{witnessing} if there is a $d \in
A^\Amf$
such that 
the following conditions are satisfied:
\begin{itemize}

\item $d$ is the root of a counting tree;

\item every path that starts at $d$ and is of length at most $m$ is
  properly counting.

\end{itemize}
We are now in a position to describe more concretely what we want to
achieve. Let $m' := 2^n \cdot 2^{2^n} \cdot m$.  For
$0 \leq i \leq m'$, let $S_i$ denote the set of concept names from
$S:=\{C_0,\dots,C_{n-1},X_2,X_3\}$ that the $i$-th domain element on a
path that is properly counting must satisfy. Then set
%
$$
\begin{array}{rcl}
  D_{0} &:=& \top \\ 
  D_{i+1} &:=& \exists r . (S_L \sqcap D_{i}) \sqcap \exists r . (S_R
               \sqcap D_i) \\
  C_i &:=& A \sqcap D_i \sqcap \bigsqcap_{0 \leq i \leq m'}
           \forall r^i . (\bigsqcap S_i \sqcap \neg \bigsqcup S
           \setminus S_i)
\end{array}
$$
Note that every model of $C_m'$ is witnessing.  We craft $\Omc_n$ and
$\Omc'_n$ such that the following holds.
\begin{lemma}
  \label{lem:threeexpcor}
  ~\\[-4mm]
\begin{enumerate}

\item $C_{m'}$ is a witness concept for $\Omc_n$ and $\Omc'_n$;

\item for every pointed model $\Amf,d$ of $\Omc_n$ of finite outdegree
  that is not witnessing, there is a pointed model $\Bmf,e$ of
  $\Omc_n \cup \Omc'_n$ such that $\Bmf,e \sim^f_{\ALCO,\Sigma}
  \Amf,d$ where $\Sigma=\text{sig}(\Omc_n)$.

\end{enumerate}
\end{lemma}
Point~2 implies that every witness concept for
$\Omc_n$ and $\Omc'_n$ must have size at least $2^m$. In fact,
%
%
if an \ALC-concept $C$ that is satisfiable w.r.t.\ $\Omc_n$ does not
mention all paths in a counting tree, including their labeling with
$S_L$ and $S_R$, then there is a model \Amf of $C$ and $\Omc_n$ in
which there is no counting tree.  Informally, this is because $\Omc_n$
does not introduce additional elements via existential quantifiers;
for a rigorous proof of an almost identical statement, see
\cite{DBLP:conf/kr/GhilardiLW06}. We argue that such a $C$ cannot be a witness
concept. Let $d \in C^\Amf$.  By Point~2, there is a pointed model
$\Bmf,e$ of $\Omc_n \cup \Omc'_n$ such that
$\Bmf,e \sim^f_{\ALCO,\Sigma} \Amf,d$. Since \ALCO-concepts are
preserved under \ALCO-bisimulations, this implies $e \in C^\Bmf$.
Consequently, $C$ is satisfiable w.r.t.\ $\Omc_n \cup \Omc'_n$ and
is not a witness concept for $\Omc_n$ and $\Omc'_n$.

\smallskip

Now for the actual construction of $\Omc_n$ and $\Omc'_n$.
The ontology $\Omc_n$ is given in Figure~\ref{fig:tbox1}.  
\begin{figure}[t!]
  \begin{boxedminipage}[t]{\columnwidth}
\vspace*{-4mm}
  \begin{center}
\begin{eqnarray}
  \top &\sqsubseteq & \forall r. \neg A \\[2mm]
  A & \sqsubseteq & (C_1=0) \\
  \top &\sqsubseteq& \forall r. (C_1{+}{+}) \\[2mm]
  A & \sqsubseteq & F_2 \\
  F_2 \sqcap (C_1<2^n) &\sqsubseteq& \forall r . F_2 \\
  F_2 &\sqsubseteq& \neg X_2\\[2mm]
  (C_1=0) &\sqsubseteq& \neg Z_2 \\
  \neg X_2 \sqcap \neg (C_1=2^n-1) &\sqsubseteq& \forall r . Z_2 \\
  Z_2 \sqcap \neg (C_1=2^n-1) &\sqsubseteq& \forall r . Z_2 \\
  X_2 \sqcap \neg Z_2 \sqcap \neg (C_1=2^n-1)  &\sqsubseteq& \forall r
                                                             . \neg
                                                             Z_2\\[2mm]
  (C_1 < 2^n-1) \sqcap X_3  &\sqsubseteq&  \forall r .X_3 \\
  (C_1 < 2^n-1) \sqcap \neg X_3  &\sqsubseteq& \forall r . \neg X_3
  \\[2mm]
  (C_1=2^n-1) \sqcap \neg Z_2 \sqcap X_2 &\equiv& E_3 \\[2mm]
  A &\sqsubseteq& F_3 \\
  F_3 \sqcap \neg (C_1=2^n-1) &\sqsubseteq&  \forall r. F_3 \\
  F_3 \sqcap (C_1=2^n-1) \sqcap \neg E_3 &\sqsubseteq&  \forall r. F_3 \\
  F_3 &\sqsubseteq& \neg X_3 \\[2mm]
  %
  A & \sqsubseteq & \neg Z_3 \\
  E_3 & \sqsubseteq & \forall r . \neg Z_3\\
  (C < 2^n-1) \sqcap Z_3 &\sqsubseteq& \forall r . Z_3 \\
  (C < 2^n-1) \sqcap \neg Z_3 &\sqsubseteq& \forall r . \neg Z_3 \\
  \neg X_3 \sqcap (C = 2^n-1) \sqcap \neg E_3 &\sqsubseteq& \forall r . Z_3 \\
  Z_3 \sqcap (C = 2^n-1) \sqcap \neg E_3 &\sqsubseteq& \forall r . Z_3 \\
  X_3 \sqcap \neg Z_3 \sqcap (C = 2^n-1) \sqcap \neg E_3 &\sqsubseteq& \forall r .\neg Z_3 \\[2mm]
  E_3 &\sqsubseteq& L_2 \\
  (C_1 < 2^n-1) \sqcap \exists r . L_2 & \sqsubseteq& L_2 \\[2mm]
E_3 \sqcap \neg Z_3 \sqcap X_3&\sqsubseteq& L_3 \sqcap L'_3\\
  (C_1 < 2^n-1) \sqcap \exists r . L'_3 & \sqsubseteq& L_3 \sqcap L'_3\\
  (C_1 < 2^n-1) \sqcap \exists r . L_3 & \sqsubseteq& L_3 \\
  (C_1 < 2^n-1) \sqcap X_2 \sqcap \exists r . L_3 & \sqsubseteq& L_3
                                                                 \sqcap
                                                                 L'_3 \\
  (C_1 = 2^n-1) \sqcap \exists r . L'_3 & \sqsubseteq & L_3 \\
  (C_1 = 2^n-1) \sqcap X_2 \sqcap \exists r . L'_3 & \sqsubseteq & L_3
                                                                   \sqcap
                                                                   L'_3
  \\[2mm]
  S_L \sqcup S_R &\sqsubseteq& \top
\end{eqnarray}
  \end{center}
  \end{boxedminipage}
  \caption{The ontology $\Omc_n$.}
  \label{fig:tbox1}
\end{figure}
We use $C \equiv D$ as an abbreviation for $C \sqsubseteq D$ and
$D \sqsubseteq C$. Line~(1) makes sure that $A$ cannot be true at
non-root nodes of counting trees, which would enable undesired witness
concepts such as $A \sqcap \exists r^{2^n} . A$.  Lines~(2) and~(3)
guarantees that Counter~1 starts with value~0 at $A$ and is
incremented modulo $2^n$ when passing to an $r$-child where
$\forall r .  (C_1{+}{+})$ is an abbreviation for a more complex
concept (of size polynomial in $n$) that achieves this; it is standard
to work out details, see e.g.\ \cite{DBLP:conf/kr/GhilardiLW06}.  The concepts
$(C_1=0)$, $(C_1=2^n-1)$, and $(C_1 < 2^n-1)$ are also abbreviations,
with the obvious meaning.  Lines~(4)~to~(6) make sure that Counter~2
starts with value~0 in all paths that are outgoing from an instance of
$A$. Lines~(7)~to~(10) guarantee that concept name $Z_2$ is true at a
domain element in a Counter~2 sequence iff there was a zero bit
strictly earlier in that sequence. Lines~(11)-(12) ensure that the
concept name $X_3$ which represents bit values for Counter~3 is
interpreted uniformly in Counter~2 sequences (which represent a single
bit position of Counter~3). Line~(13) guarantees that concept name
$E_3$ marks exactly the final domain element on Counter~3 sequences.
Lines~(14)-(17) enforce that Counter~3 starts with value~0 in all
paths that are outgoing from an instance of $A$. Lines~(18)-(24) make
sure that concept name $Z_3$ is true at a domain element in a
Counter~3 sequence iff there was a zero bit strictly earlier in that
sequence. Lines~(25)-(26) gurantee that $L_2$ is true in the last
element of every Counter~2 sequence and Lines~(27)-(32) achive that
$L_3$ is true in the last Counter~2 subsequence of every Counter~3
sequence. In both cases, we do not need the converse (altough it would
be possible to achieve also the converse with further concept
inclusions). Note that we do not use the concept names $S_L$ and
$S_R$ as this does not turn out to be necessary. However, we want
them to be part of $\text{sig}(\Omc_1)$, which is achieved by Line~(33).

\begin{figure*}[t!]
  \begin{boxedminipage}[t]{\textwidth}
\vspace*{-4mm}
  \begin{center}
\begin{eqnarray}
  A & \sqsubseteq & P_1 \sqcup P_2 \sqcup P_3 \\[2mm]
  P_1 \sqcap E_3 \sqcap \neg Z_3
  \sqcap X_3 & \sqsubseteq &  \forall r . (S_L \rightarrow P_1)
  \sqcup \forall r . (S_R \rightarrow P_1) \\
  P_1 \sqcap E_3 \sqcap \neg Z_3
  \sqcap X_3 & \sqsubseteq &
                                                                    \bot \\[2mm]
  P_2 &\sqsubseteq& M_0 \sqcup \exists r . ((S_L \sqcup S_R) \sqcap P_2) \\
  P_2 & \sqsubseteq& \neg (L_2 \sqcap L_3)\\[2mm]
  %
  M_0 &\sqsubseteq& M \sqcap (D_1=0) \\
  M_0 \sqcap \neg Z_2 &\sqsubseteq& NX_2 \leftrightarrow X_2 \\
  M_0 \sqcap Z_2 &\sqsubseteq& NX_2 \leftrightarrow \neg X_2 \\
  M \sqcap (D_1 < 2^n-1) &\sqsubseteq& \forall r . (D_1{+}{+}) \\
  M \sqcap (D_1 < 2^n-1) \sqcap NX_2 &\sqsubseteq& \exists r. ((S_L \sqcup S_R) \sqcap M \sqcap NX_2)\\
  M \sqcap (D_1 < 2^n-1) \sqcap \neg NX_2 &\sqsubseteq& \exists
  r. ((S_L \sqcup S_R) \sqcap M \sqcap \neg NX_2) \\
  M \sqcap (D_1=2^n-1) \sqcap NX_2 &\sqsubseteq& X_2 \\
  M \sqcap (D_1=2^n-1) \sqcap \neg NX_2 &\sqsubseteq& \neg X_2 
\end{eqnarray}
  \end{center}
  \end{boxedminipage}
  \caption{The first part of ontology $\Omc'_n$.}
  \label{fig:tbox2a}
\end{figure*}
%
\smallskip We present the ontology $\Omc'_n$ in two parts, one
pertaining to Counter~2 and one pertaining to Counter~3. The first
part of $\Omc'_n$ can be found in Figure~\ref{fig:tbox2a}. We want to
achieve that for all forest models \Amf of $\Omc_n$ of finite
outdegree and all $d \in A^\Amf$, there is a pointed model $\Bmf,e$ of
$\Omc_n \cup \Omc'_n$ such that $\Bmf,e \sim^f_{\ALCO,\Sigma} \Amf,d$
if and only if one of the following holds:
\begin{enumerate}

\item \Amf does not contain a counting tree rooted at $d$;

\item Counter~2 is not properly incremented on some path in \Amf that starts at $d$;

\item Counter~3 is not properly incremented on some path in \Amf that
  starts at $d$.

\end{enumerate}
We speak of these three options as \emph{defects} of type~1,~2, and~3,
respectively. In Line~(34), we choose a defect type that is present in
the current model \Amf. More precisely, being able to make $P_i$ true
at $e$ in \Bmf corresponds to a defect of type~$i$ in~\Amf.

Lines~(35)-(36) implement defects of type~1.  To see how this works,
first assume that there is a counting tree in \Amf rooted at
$d \in A^\Amf$ and let $\Bmf,e$ be a pointed model of $\Omc'_n$ with
$\Bmf,e \sim^f_{\ALCO,\Sigma} \Amf,d$. We have to argue that
$e \notin P_1^\Bmf$. This is due to Lines~(35) and~(36) which then
require the existence of a path in the counting tree to an element $f$
that has no $r$-successor that satisfies $S_L$ or no $r$-successor
that satisfies $S_R$, such that the maximum value of Counter~3 is not
reached on the way to $f$. But no such path exists.

For the converse, assume that for some $d \in A^\Amf$, there is a no
counting tree in \Amf rooted at $d$. This implies the existence of a
word $w \in \{L,R\}^*$ of length strictly less than $m$ such that
there is no path $p$ in \Amf starting at $d$ that follows branching
pattern $L,R$ and ends in an element where Counter~3 has maximum value
and that has no $r$-successor that satisfies $S_L$ or no $r$-successor
that satisfies $S_R$. Let $\Amf|^\downarrow_d$ denote the restriction
of \Amf to the domain elements that are reachable from $d$, traveling
role names only in the forward direction. Further let \Bmf be obtained
from $\Amf|^\downarrow_d$ by making $P_1$ true on every  element
on path $p$. Then \Bmf is a model of $\Omc'_n$ with $d \in P_1^\Bmf$
and $\Bmf,d \sim^f_{\ALCO,\Sigma} \Amf,d$.

Lines~(37)-(46) verify that, if $P_2$ is chosen in Line~(34), then
there is indeed a defect of type~2. Here we use an auxiliary single
exponential counter $D_1$ based on concept names $D_0,\dots,D_{n-1}$.
Lines~(37) and (38) mark the place where incrementation of Counter~2
fails using the concept name $M_0$. 
Note that Line~(38) ensures that $M_0$ is chosen before the last
Counter~2 sequence in the last Counter~3 sequence is reached.  When a
domain element $d$ is marked with $M_0$, this means that it is a bit
of Counter~2 such that, on some path outgoing from $d$, the
corresponding bit in the subsequent Counter~2 sequence violates
incrementation. There are two ways in which this may happen: first,
there may be no 0-bit lower than the bit marked with $M_0$, but the
corresponding bit in the subsequent Counter~2 sequence is not
toggled. Second, there may be a 0-bit lower than the bit marked with
$M_0$, but the corresponding bit in the subsequent Counter~2 sequence
is toggled. These two cases are distinguished by Lines~(40) and (41).
In the first case, the value of $X_2$ is stored in $NX_2$. In the
second case, the toggled value of $X_2$ is stored in $NX_2$. The
counter $D_1$ is then reset in Line~(39) and incremented in Line~(42)
to identify the corresponding bit in the following
configuration. Through lines (43) and (44), the value of $NX_2$ is
passed on all the way to this bit.  Finally, Lines~(45) and (46)
ensure that the $X_2$-value of the corresponding bit coincides with
$NX_2$.  It is not so difficult to prove formally that this works. In
particular, if there is a path starting at some $d \in A^\Amf$ on
which Counter~2 is not properly incremented, then we can extend
$\Amf|^\downarrow_d$ in a straightforward way to a model \Bmf of
$\Omc'_n$ with $d \in P_2^\Bmf$, by interpreting the concept names in
$\text{sig}(\Omc'_n) \setminus \text{sig}(\Omc_n)$.

\begin{figure*}[t]
  \begin{boxedminipage}[t]{\textwidth}
\vspace*{-4mm}
  \begin{center}
\begin{eqnarray}
  P_3 &\sqsubseteq& M_0 \sqcup \exists r . ((S_L \sqcup S_R) \sqcap P_3)\\
  P_3 & \sqsubseteq& \neg L_3 \\[2mm]
  M_0 &\sqsubseteq& (C_1=0) \sqcap M_1 \sqcap (M_L \sqcup M_R) \\
   K &\sqsubseteq& \forall r . K \qquad \text{ for all } K \in
                   \{M_L,M_R\} \\[2mm]
  M_1 & \sqsubseteq& (D_1=C_1) \sqcap (NX_2 \leftrightarrow X_2) \sqcap
                    (NX_3 \leftrightarrow X_3)\\
  M_1 \sqcap (C_1 < 2^n-1) &\sqsubseteq& \exists r . ((S_L \sqcup S_R)
  \sqcap M_1) \sqcap \exists r . ((S_L \sqcup S_R) \sqcap M_2)\\
    M_2 \sqcap \neg E_3 &\sqsubseteq & \exists r . ((S_L \sqcup S_R) \sqcap M_2) \\[2mm]
  M_i \sqcap K&\sqsubseteq&
                                                                 \forall
                                                                 r
                                                                 . (M_i
                                                                 \rightarrow
                                                                 K) \quad
                                         \text{ for all } i \in
                            \{2,3,4\} \text{ and} \\
                            && K \in \{ C, \neg C \mid C \in
                               \{D_0,\dots,D_{n-1},NX_2,NX_3 \} \}\\[2mm]
  M_2 \sqcap M_L \sqcap E_3 &\sqsubseteq& \exists r. (S_L \sqcap M_3) \\
  M_2 \sqcap M_R \sqcap E_3 &\sqsubseteq& \exists r. (S_R \sqcap M_3)
\\[2mm]
M_3 &\sqsubseteq& \neg E_3 \sqcap (M_4 \sqcup \exists r . ((S_L \cup
S_R) \sqcap M_3) \\[2mm]
  M_4 &\sqsubseteq& (C_1 = D_1) \sqcap M_5 \sqcap (X_2 \leftrightarrow NX_2)
                                         \, \sqcap \\
      && (Z_3 \sqcap (X_3 \leftrightarrow \neg NX_3)) \sqcup
         (\neg Z_3 \sqcap (X_3 \leftrightarrow NX_3))
  \\[2mm]
  M_5 \sqcap (C_1 < 2^n-1) &\sqsubseteq& \exists r . ((S_L \sqcup S_R) \sqcap M_5) \\
  M_5 \sqcap (C_1 = 2^n-1) &\sqsubseteq& \{ c \} 
\end{eqnarray}
  \end{center}
  \end{boxedminipage}
  \caption{The second part of ontology $\Omc'_n$.}
  \label{fig:tbox2b}
\end{figure*}
\smallskip 

The part of $\Omc'_n$ that is concerned with Counter~3 is displayed in
Figure~\ref{fig:tbox2b}. It makes sure that if $P_2$ is chosen in
Line~(34), then there is indeed a defect of type~3. It is here that
using a  nominal is crucial. Lines~(47) and~(48) mark the place where
incrementation of Counter~3 fails using the concept name $M_0$. Note
that Line~(48) ensures that $M_0$ is chosen in some Counter~3 sequence
that is not the final one; we refer to it as the `current' Counter~3
sequence.  Line~(49) further makes sure that the element chosen by
$M_0$ is at the beginning of a Counter~2 sequence, which we refer to
as the `current' Counter~2 sequence. It also chooses via the concept
names $M_L$ and $M_R$ whether the defect occurs in the subsequent
Counter~3 sequence that is a left child of the current sequence, or a
right child. The chosen value is memorized `forever' in Line~(50). Our
aim is to set another marker at the beginning of a Counter~2 sequence
in a subsequent Counter~3 sequence that encodes the same Counter~2
value as the current Counter~2 sequence, and then to compare the two
$X_3$-bit values of the two Counter~2 sequences. 

To achieve this, we need to memorize for later comparison \emph{all}
(exponentially many) $X_2$-bit values of the current Counter~2
sequence. This cannot be done in a single type and thus we use
multiple types. This is implemented in Lines~(51)-(55) in which the
current Counter~2 sequence is traversed from beginning to end. In each
step, a branching takes place via Line~(52). It is important to
understand that this branching is in model \Bmf, but not necessarily
in model \Amf. Recall that we are interested in models \Bmf that have
a functional $\ALCO,\Sigma$-bisimulation to \Amf. Informally, we can
assume the two models to have the same domain and $r$-structure up to
the element $d$ in \Amf in which we have chosen to set the marker
$M_0$.\footnote{Please compare this to defects of type~1 and~2 where
  \Bmf can be assumed to have the same domain and $r$-structure as \Amf; this is also the
  case here, up to the $M_0$ marker, but not beyond.} 
When setting
$M_0$, then we `are' in an element $e$ of \Bmf that is
$\Sigma$-bisimilar to $d$. This and what follows is illustrated in
Figure~\ref{fig:threeexppic}. Now Line~(52) creates two $r$-successors
$f_1$ and $f_2$ of $e$ in \Bmf that are both $\Sigma$-bisimilar to
$r$-successors of $d$ in \Amf. We shall argue a bit later that this
must actually be the same $r$-successor of $d$. In branch $f_1$ of
\Bmf, we stay with marker $M_1$ while in branch $f_2$ of \Bmf, we
switch to marker $M_2$. The $M_1$-branch branches again at the next
point of the Counter~2 sequence while the $M_2$ path does not, and so
on. Via Line~(51), at each point of the Counter~2 sequence we memorize
in the newly generated $M_1$-branch the value of Counter~1 in the
auxiliary counter $D_1$, the value of $X_2$ in $NX_2$, and the value
of $X_3$ in $NX_3$. In contrast, the $M_2$-branches retain their
memory via Lines~(54)-(55).  At the end of the Counter~2 sequence
whose beginning is marked with $M_0$, we have thus generated $2^n$
branches in \Bmf, each storing the $X_2$-bit value for one bit
position of Counter~2, and all of them storing the $X_3$-bit value of
the current Counter~2 sequence. 

Via Line~(53), all the $M_2$-branches extend also beyond the current
Counter~2 sequence to the end of the current Counter~3 sequence.  In
Lines~(56)~and~(57), we make a step to the first element of a
subsequent Counter~3 sequence, switching to marker $M_3$. \emph{All}
$M_2$-branches in \Bmf decide to go to an $S_L$-labeled such
subsequent sequence or all decide to go to an $S_R$-labeled such
sequence, depending on whether we had initially (before the branching)
chosen marker $M_L$ or $M_R$. Via Line~(58), we proceed down the
Counter~3 sequence and set the $M_4$-marker before reaching its end.
What we want to achieve is that the $M_4$-marker is set at the
Counter~2 subsequence that carries the same Counter~2 value as the
Counter~2 sequence, at the Counter~2 bit position that the current
branch has stored in counter $D_1$. We verify independently for each
branch in \Bmf that the bit position is correct, in Lines~(59)-(60),
and there we also make sure that the $X_2$-bit value coincides with
the $X_2$-bit value stored in $NX_2$. We also use Line~(60) to make
sure that there is indeed an incrementation conflict of Counter~3 at
this position.

We are done if we can additionally guarantee that the different
branches in \Bmf have really set the $M_4$-marker at the same
Counter~2 subsequence in the same path of \Amf. So far, however, we do
not not know that this is the case, nor that the different branches
have even followed the same path of \Amf. But this is now easily
rectified: Lines~(61) and~(62) force all branches of \Bmf to further
follow the current Counter~2 sequence, until it's end, and that the
individual name $c$ is satisfied at the end of all branches. Thus the
end of all branches is the same element in \Bmf, which is functionally
bisimilar to some element of \Amf. But \Amf is a forest model and, as
$\Omc_1$ does not use nominals, it is (by definition) even a tree
model.  Consequently, on every branch of \Bmf, the $r$-predecessor of
the final element marked with $c$ is functionally bisimilar to the
same element of \Amf, and so on, all the way up to the element of \Bmf
where the $M_0$-marker was set.
%

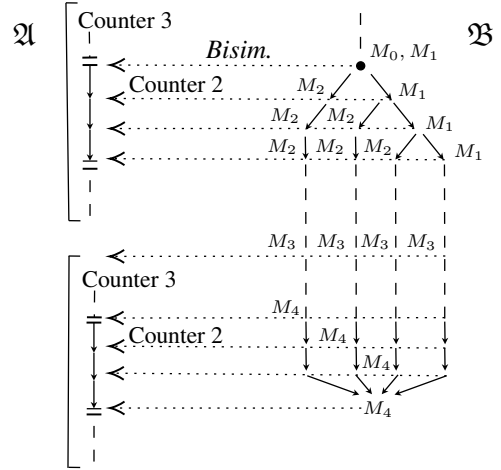
\begin{figure}[t!]
	\begin{center}
		\tikzset{every picture/.style={line width=0.5pt}} 
		
		\begin{tikzpicture}[x=0.5pt,y=0.5pt,yscale=-1,xscale=1]
			
			\draw    (489.03,61.43) -- (475.31,78.43) ;
			\draw [shift={(473.43,80.76)}, rotate = 308.9] [fill={rgb, 255:red, 0; green, 0; blue, 0 }  ][line width=0.08]  [draw opacity=0] (5,-2.5) -- (0,0) -- (5,2.5) -- (3.5,0) -- cycle    ;
			\draw    (504.43,61.86) -- (518.14,78.86) ;
			\draw [shift={(520.03,81.19)}, rotate = 231.1] [fill={rgb, 255:red, 0; green, 0; blue, 0 }  ][line width=0.08]  [draw opacity=0] (5,-2.5) -- (0,0) -- (5,2.5) -- (3.5,0) -- cycle    ;
			\draw    (470.46,84.57) -- (456.74,101.57) ;
			\draw [shift={(454.86,103.9)}, rotate = 308.9] [fill={rgb, 255:red, 0; green, 0; blue, 0 }  ][line width=0.08]  [draw opacity=0] (5,-2.5) -- (0,0) -- (5,2.5) -- (3.5,0) -- cycle    ;
			\draw    (521.57,83.57) -- (535.29,100.57) ;
			\draw [shift={(537.17,102.9)}, rotate = 231.1] [fill={rgb, 255:red, 0; green, 0; blue, 0 }  ][line width=0.08]  [draw opacity=0] (5,-2.5) -- (0,0) -- (5,2.5) -- (3.5,0) -- cycle    ;
			\draw    (510.74,83.43) -- (497.03,100.43) ;
			\draw [shift={(495.14,102.76)}, rotate = 308.9] [fill={rgb, 255:red, 0; green, 0; blue, 0 }  ][line width=0.08]  [draw opacity=0] (5,-2.5) -- (0,0) -- (5,2.5) -- (3.5,0) -- cycle    ;
			\draw    (492.74,109.14) -- (492.96,123.33) ;
			\draw [shift={(493,126.33)}, rotate = 269.14] [fill={rgb, 255:red, 0; green, 0; blue, 0 }  ][line width=0.08]  [draw opacity=0] (5,-2.5) -- (0,0) -- (5,2.5) -- (3.5,0) -- cycle    ;
			\draw    (538.46,107.14) -- (524.74,124.14) ;
			\draw [shift={(522.86,126.48)}, rotate = 308.9] [fill={rgb, 255:red, 0; green, 0; blue, 0 }  ][line width=0.08]  [draw opacity=0] (5,-2.5) -- (0,0) -- (5,2.5) -- (3.5,0) -- cycle    ;
			\draw    (543.86,107.29) -- (557.57,124.28) ;
			\draw [shift={(559.46,126.62)}, rotate = 231.1] [fill={rgb, 255:red, 0; green, 0; blue, 0 }  ][line width=0.08]  [draw opacity=0] (5,-2.5) -- (0,0) -- (5,2.5) -- (3.5,0) -- cycle    ;
			\draw    (454.74,109.14) -- (454.96,123.33) ;
			\draw [shift={(455,126.33)}, rotate = 269.14] [fill={rgb, 255:red, 0; green, 0; blue, 0 }  ][line width=0.08]  [draw opacity=0] (5,-2.5) -- (0,0) -- (5,2.5) -- (3.5,0) -- cycle    ;
			\draw  [dash pattern={on 4.5pt off 4.5pt}]  (455,130.19) -- (455.22,199.5) ;
			\draw  [dash pattern={on 4.5pt off 4.5pt}]  (493,130.19) -- (493.22,199.5) ;
			\draw  [dash pattern={on 4.5pt off 4.5pt}]  (523,130.19) -- (523.22,199.5) ;
			\draw  [dash pattern={on 4.5pt off 4.5pt}]  (560,130.19) -- (560.22,199.5) ;
			\draw  [dash pattern={on 4.5pt off 4.5pt}]  (455.22,201.83) -- (455.22,249.17) ;
			\draw  [dash pattern={on 4.5pt off 4.5pt}]  (493.22,201.83) -- (493.22,249.17) ;
			\draw  [dash pattern={on 4.5pt off 4.5pt}]  (523.22,201.83) -- (523.22,249.17) ;
			\draw  [dash pattern={on 4.5pt off 4.5pt}]  (560.22,201.83) -- (560.22,249.17) ;
			\draw    (455.22,249.17) -- (455.43,263.36) ;
			\draw [shift={(455.48,266.36)}, rotate = 269.14] [fill={rgb, 255:red, 0; green, 0; blue, 0 }  ][line width=0.08]  [draw opacity=0] (5,-2.5) -- (0,0) -- (5,2.5) -- (3.5,0) -- cycle    ;
			\draw    (493.22,249.17) -- (493.43,263.36) ;
			\draw [shift={(493.48,266.36)}, rotate = 269.14] [fill={rgb, 255:red, 0; green, 0; blue, 0 }  ][line width=0.08]  [draw opacity=0] (5,-2.5) -- (0,0) -- (5,2.5) -- (3.5,0) -- cycle    ;
			\draw    (523.22,249.17) -- (523.43,263.36) ;
			\draw [shift={(523.48,266.36)}, rotate = 269.14] [fill={rgb, 255:red, 0; green, 0; blue, 0 }  ][line width=0.08]  [draw opacity=0] (5,-2.5) -- (0,0) -- (5,2.5) -- (3.5,0) -- cycle    ;
			\draw    (560.22,249.17) -- (560.43,263.36) ;
			\draw [shift={(560.48,266.36)}, rotate = 269.14] [fill={rgb, 255:red, 0; green, 0; blue, 0 }  ][line width=0.08]  [draw opacity=0] (5,-2.5) -- (0,0) -- (5,2.5) -- (3.5,0) -- cycle    ;
			\draw    (455.22,269.17) -- (455.43,283.36) ;
			\draw [shift={(455.48,286.36)}, rotate = 269.14] [fill={rgb, 255:red, 0; green, 0; blue, 0 }  ][line width=0.08]  [draw opacity=0] (5,-2.5) -- (0,0) -- (5,2.5) -- (3.5,0) -- cycle    ;
			\draw    (493.22,269.17) -- (493.43,283.36) ;
			\draw [shift={(493.48,286.36)}, rotate = 269.14] [fill={rgb, 255:red, 0; green, 0; blue, 0 }  ][line width=0.08]  [draw opacity=0] (5,-2.5) -- (0,0) -- (5,2.5) -- (3.5,0) -- cycle    ;
			\draw    (523.22,269.17) -- (523.43,283.36) ;
			\draw [shift={(523.48,286.36)}, rotate = 269.14] [fill={rgb, 255:red, 0; green, 0; blue, 0 }  ][line width=0.08]  [draw opacity=0] (5,-2.5) -- (0,0) -- (5,2.5) -- (3.5,0) -- cycle    ;
			\draw    (560.22,269.17) -- (560.43,283.36) ;
			\draw [shift={(560.48,286.36)}, rotate = 269.14] [fill={rgb, 255:red, 0; green, 0; blue, 0 }  ][line width=0.08]  [draw opacity=0] (5,-2.5) -- (0,0) -- (5,2.5) -- (3.5,0) -- cycle    ;
			\draw    (492.89,289.5) -- (503.29,301.87) ;
			\draw [shift={(505.22,304.17)}, rotate = 229.94] [fill={rgb, 255:red, 0; green, 0; blue, 0 }  ][line width=0.08]  [draw opacity=0] (5,-2.5) -- (0,0) -- (5,2.5) -- (3.5,0) -- cycle    ;
			\draw    (524.89,289.5) -- (514.49,301.87) ;
			\draw [shift={(512.56,304.17)}, rotate = 310.06] [fill={rgb, 255:red, 0; green, 0; blue, 0 }  ][line width=0.08]  [draw opacity=0] (5,-2.5) -- (0,0) -- (5,2.5) -- (3.5,0) -- cycle    ;
			\draw    (456.22,290.17) -- (492.74,303.47) ;
			\draw [shift={(495.56,304.5)}, rotate = 200.02] [fill={rgb, 255:red, 0; green, 0; blue, 0 }  ][line width=0.08]  [draw opacity=0] (5,-2.5) -- (0,0) -- (5,2.5) -- (3.5,0) -- cycle    ;
			\draw    (561.56,290.17) -- (525.04,303.47) ;
			\draw [shift={(522.22,304.5)}, rotate = 339.98] [fill={rgb, 255:red, 0; green, 0; blue, 0 }  ][line width=0.08]  [draw opacity=0] (5,-2.5) -- (0,0) -- (5,2.5) -- (3.5,0) -- cycle    ;
			\draw  [fill={rgb, 255:red, 0; green, 0; blue, 0 }  ,fill opacity=1 ] (499.68,55.54) .. controls (499.68,53.86) and (498.32,52.5) .. (496.64,52.5) .. controls (494.96,52.5) and (493.6,53.86) .. (493.6,55.54) .. controls (493.6,57.22) and (494.96,58.58) .. (496.64,58.58) .. controls (498.32,58.58) and (499.68,57.22) .. (499.68,55.54) -- cycle ;
			\draw  [dash pattern={on 4.5pt off 4.5pt}]  (496.5,15.58) -- (496.5,50) ;
			\draw    (273.67,7.83) -- (275.87,172.73) ;
			\draw    (273.67,7.83) -- (282.44,7.94) ;
			\draw    (275.87,172.73) -- (284.64,172.84) ;
			\draw    (276,199.5) -- (275.2,359.93) ;
			\draw    (275.67,199.5) -- (284.44,199.61) ;
			\draw  [dash pattern={on 0.84pt off 2.51pt}]  (496.64,55.54) -- (309.23,54.92) ;
			\draw [shift={(307.23,54.92)}, rotate = 360.19] [color={rgb, 255:red, 0; green, 0; blue, 0 }  ][line width=0.75]    (10.93,-4.9) .. controls (6.95,-2.3) and (3.31,-0.67) .. (0,0) .. controls (3.31,0.67) and (6.95,2.3) .. (10.93,4.9)   ;
			\draw  [dash pattern={on 0.84pt off 2.51pt}]  (520.03,81.19) -- (308.77,80.21) ;
			\draw [shift={(306.77,80.2)}, rotate = 360.27] [color={rgb, 255:red, 0; green, 0; blue, 0 }  ][line width=0.75]    (10.93,-4.9) .. controls (6.95,-2.3) and (3.31,-0.67) .. (0,0) .. controls (3.31,0.67) and (6.95,2.3) .. (10.93,4.9)   ;
			\draw  [dash pattern={on 0.84pt off 2.51pt}]  (537.17,102.9) -- (309.3,103) ;
			\draw [shift={(307.3,103)}, rotate = 359.98] [color={rgb, 255:red, 0; green, 0; blue, 0 }  ][line width=0.75]    (10.93,-4.9) .. controls (6.95,-2.3) and (3.31,-0.67) .. (0,0) .. controls (3.31,0.67) and (6.95,2.3) .. (10.93,4.9)   ;
			\draw  [dash pattern={on 0.84pt off 2.51pt}]  (560.15,126.62) -- (309.03,125.81) ;
			\draw [shift={(307.03,125.8)}, rotate = 360.19] [color={rgb, 255:red, 0; green, 0; blue, 0 }  ][line width=0.75]    (10.93,-4.9) .. controls (6.95,-2.3) and (3.31,-0.67) .. (0,0) .. controls (3.31,0.67) and (6.95,2.3) .. (10.93,4.9)   ;
			\draw  [dash pattern={on 0.84pt off 2.51pt}]  (560.66,199.5) -- (308.77,199.8) ;
			\draw [shift={(306.77,199.8)}, rotate = 359.93] [color={rgb, 255:red, 0; green, 0; blue, 0 }  ][line width=0.75]    (10.93,-4.9) .. controls (6.95,-2.3) and (3.31,-0.67) .. (0,0) .. controls (3.31,0.67) and (6.95,2.3) .. (10.93,4.9)   ;
			\draw  [dash pattern={on 0.84pt off 2.51pt}]  (560.66,269.17) -- (308.77,267.81) ;
			\draw [shift={(306.77,267.8)}, rotate = 360.31] [color={rgb, 255:red, 0; green, 0; blue, 0 }  ][line width=0.75]    (10.93,-4.9) .. controls (6.95,-2.3) and (3.31,-0.67) .. (0,0) .. controls (3.31,0.67) and (6.95,2.3) .. (10.93,4.9)   ;
			\draw  [dash pattern={on 0.84pt off 2.51pt}]  (561.56,290.17) -- (308.5,288.22) ;
			\draw [shift={(306.5,288.2)}, rotate = 360.44] [color={rgb, 255:red, 0; green, 0; blue, 0 }  ][line width=0.75]    (10.93,-4.9) .. controls (6.95,-2.3) and (3.31,-0.67) .. (0,0) .. controls (3.31,0.67) and (6.95,2.3) .. (10.93,4.9)   ;
			\draw  [dash pattern={on 0.84pt off 2.51pt}]  (497.11,314.67) -- (309.03,314.6) ;
			\draw [shift={(307.03,314.6)}, rotate = 360.02] [color={rgb, 255:red, 0; green, 0; blue, 0 }  ][line width=0.75]    (10.93,-4.9) .. controls (6.95,-2.3) and (3.31,-0.67) .. (0,0) .. controls (3.31,0.67) and (6.95,2.3) .. (10.93,4.9)   ;
			\draw  [dash pattern={on 4.5pt off 4.5pt}]  (291.33,29.67) -- (291.67,49.25) ;
			\draw [shift={(291.67,49.25)}, rotate = 269.02] [color={rgb, 255:red, 0; green, 0; blue, 0 }  ][line width=0.75]    (0,5.59) -- (0,-5.59)   ;
			\draw    (291.83,54.92) -- (291.83,77.42) ;
			\draw [shift={(291.83,80.42)}, rotate = 270] [fill={rgb, 255:red, 0; green, 0; blue, 0 }  ][line width=0.08]  [draw opacity=0] (5,-2.5) -- (0,0) -- (5,2.5) -- (3.5,0) -- cycle    ;
			\draw [shift={(291.83,54.92)}, rotate = 270] [color={rgb, 255:red, 0; green, 0; blue, 0 }  ][line width=0.75]    (0,5.59) -- (0,-5.59)   ;
			\draw    (291.83,80.42) -- (291.83,100.92) ;
			\draw [shift={(291.83,103.92)}, rotate = 270] [fill={rgb, 255:red, 0; green, 0; blue, 0 }  ][line width=0.08]  [draw opacity=0] (5,-2.5) -- (0,0) -- (5,2.5) -- (3.5,0) -- cycle    ;
			\draw    (291.83,103.92) -- (291.83,124.42) ;
			\draw [shift={(291.83,127.42)}, rotate = 270] [fill={rgb, 255:red, 0; green, 0; blue, 0 }  ][line width=0.08]  [draw opacity=0] (5,-2.5) -- (0,0) -- (5,2.5) -- (3.5,0) -- cycle    ;
			\draw    (286.4,127.63) -- (297.4,127.63) ;
			\draw  [dash pattern={on 4.5pt off 4.5pt}]  (291.87,169.13) -- (291.9,133.45) ;
			\draw [shift={(291.9,133.45)}, rotate = 450.05] [color={rgb, 255:red, 0; green, 0; blue, 0 }  ][line width=0.75]    (0,5.59) -- (0,-5.59)   ;
			\draw  [dash pattern={on 4.5pt off 4.5pt}]  (295.33,225.67) -- (295.13,245.93) ;
			\draw [shift={(295.13,245.93)}, rotate = 270.57] [color={rgb, 255:red, 0; green, 0; blue, 0 }  ][line width=0.75]    (0,5.59) -- (0,-5.59)   ;
			\draw    (294.83,249.92) -- (294.83,269.79) ;
			\draw [shift={(294.83,272.79)}, rotate = 270] [fill={rgb, 255:red, 0; green, 0; blue, 0 }  ][line width=0.08]  [draw opacity=0] (5,-2.5) -- (0,0) -- (5,2.5) -- (3.5,0) -- cycle    ;
			\draw [shift={(294.83,249.92)}, rotate = 270] [color={rgb, 255:red, 0; green, 0; blue, 0 }  ][line width=0.75]    (0,5.59) -- (0,-5.59)   ;
			\draw    (294.83,272.79) -- (294.83,290.86) ;
			\draw [shift={(294.83,293.86)}, rotate = 270] [fill={rgb, 255:red, 0; green, 0; blue, 0 }  ][line width=0.08]  [draw opacity=0] (5,-2.5) -- (0,0) -- (5,2.5) -- (3.5,0) -- cycle    ;
			\draw    (294.83,293.86) -- (294.83,311.94) ;
			\draw [shift={(294.83,314.94)}, rotate = 270] [fill={rgb, 255:red, 0; green, 0; blue, 0 }  ][line width=0.08]  [draw opacity=0] (5,-2.5) -- (0,0) -- (5,2.5) -- (3.5,0) -- cycle    ;
			\draw    (289.4,315.13) -- (300.4,315.13) ;
			\draw  [dash pattern={on 4.5pt off 4.5pt}]  (294.87,355.13) -- (294.9,319.45) ;
			\draw [shift={(294.9,319.45)}, rotate = 450.05] [color={rgb, 255:red, 0; green, 0; blue, 0 }  ][line width=0.75]    (0,5.59) -- (0,-5.59)   ;
			\draw  [dash pattern={on 0.84pt off 2.51pt}]  (560.22,246.17) -- (308.32,246.46) ;
			\draw [shift={(306.32,246.47)}, rotate = 359.93] [color={rgb, 255:red, 0; green, 0; blue, 0 }  ][line width=0.75]    (10.93,-4.9) .. controls (6.95,-2.3) and (3.31,-0.67) .. (0,0) .. controls (3.31,0.67) and (6.95,2.3) .. (10.93,4.9)   ;
			\draw    (275.67,359.5) -- (284.44,359.61) ;
			
			\draw (500.93,37.33) node [anchor=north west][inner sep=0.75pt]  [font=\scriptsize] [align=left] {$\displaystyle M_{0} ,M_{1}$};
			\draw (445.6,63) node [anchor=north west][inner sep=0.75pt]  [font=\scriptsize] [align=left] {$\displaystyle M_{2}$};
			\draw (522.6,67) node [anchor=north west][inner sep=0.75pt]  [font=\scriptsize] [align=left] {$\displaystyle M_{1}$};
			\draw (543.17,91.71) node [anchor=north west][inner sep=0.75pt]  [font=\scriptsize] [align=left] {$\displaystyle M_{1}$};
			\draw (565.22,113.29) node [anchor=north west][inner sep=0.75pt]  [font=\scriptsize] [align=left] {$\displaystyle M_{1}$};
			\draw (425.89,85.43) node [anchor=north west][inner sep=0.75pt]  [font=\scriptsize] [align=left] {$\displaystyle M_{2}$};
			\draw (423.93,109.05) node [anchor=north west][inner sep=0.75pt]  [font=\scriptsize] [align=left] {$\displaystyle M_{2}$};
			\draw (468.31,85.38) node [anchor=north west][inner sep=0.75pt]  [font=\scriptsize] [align=left] {$\displaystyle M_{2}$};
			\draw (459.98,109.05) node [anchor=north west][inner sep=0.75pt]  [font=\scriptsize] [align=left] {$\displaystyle M_{2}$};
			\draw (496.31,109.38) node [anchor=north west][inner sep=0.75pt]  [font=\scriptsize] [align=left] {$\displaystyle M_{2}$};
			\draw (423.6,181.05) node [anchor=north west][inner sep=0.75pt]  [font=\scriptsize] [align=left] {$\displaystyle M_{3}$};
			\draw (460.93,181.05) node [anchor=north west][inner sep=0.75pt]  [font=\scriptsize] [align=left] {$\displaystyle M_{3}$};
			\draw (493.93,181.05) node [anchor=north west][inner sep=0.75pt]  [font=\scriptsize] [align=left] {$\displaystyle M_{3}$};
			\draw (528.93,181.38) node [anchor=north west][inner sep=0.75pt]  [font=\scriptsize] [align=left] {$\displaystyle M_{3}$};
			\draw (426.6,231.05) node [anchor=north west][inner sep=0.75pt]  [font=\scriptsize] [align=left] {$\displaystyle M_{4}$};
			\draw (496.27,307.71) node [anchor=north west][inner sep=0.75pt]  [font=\scriptsize] [align=left] {$\displaystyle M_{4}$};
			\draw (462.1,252.05) node [anchor=north west][inner sep=0.75pt]  [font=\scriptsize] [align=left] {$\displaystyle M_{4}$};
			\draw (495.1,272.55) node [anchor=north west][inner sep=0.75pt]  [font=\scriptsize] [align=left] {$\displaystyle M_{4}$};
			\draw (575,21) node [anchor=north west][inner sep=0.75pt]  [font=\Large] [align=left] {$\displaystyle \mathfrak{B}$};
			\draw (231.67,21) node [anchor=north west][inner sep=0.75pt]  [font=\Large] [align=left] {$\displaystyle \mathfrak{A}$};
			\draw (281,12.5) node [anchor=north west][inner sep=0.75pt]  [font=\small] [align=left] {Counter 3};
			\draw (377.67,34.33) node [anchor=north west][inner sep=0.75pt]   [align=left] {\textit{Bisim.}};
			\draw (319.4,61.3) node [anchor=north west][inner sep=0.75pt]  [font=\small] [align=left] {Counter 2};
			\draw (284,208.5) node [anchor=north west][inner sep=0.75pt]  [font=\small] [align=left] {Counter 3};
			\draw (319.23,250.8) node [anchor=north west][inner sep=0.75pt]  [font=\small] [align=left] {Counter 2};

		\end{tikzpicture}
		\caption{Strategy for comparing bit values of Counter~3}
		\label{fig:threeexppic}
	\end{center}
\end{figure}

\smallskip

Based on what was said above, it can be verified that
Lemma~\ref{lem:threeexpcor} indeed holds. We refrain from giving
details.

\subsection{\ThreeExpTime-Hardness}

An \emph{Alternating Turing Machine (ATM)} is of the form $\mathcal{M}
= (Q,\Sigma,\Gamma,q_0,\Delta)$. The set of \emph{states} $Q =
Q_\exists \uplus Q_\forall \uplus \{q_a\} \uplus \{q_r\}$ consists of
\emph{existential states} from $Q_\exists$, \emph{universal states}
from $Q_\forall$, an \emph{accepting state} $q_a$, and a
\emph{rejecting state} $q_r$; $\Sigma$ is the \emph{input alphabet}
and $\Gamma$ the \emph{work alphabet} containing a \emph{blank symbol}
$\square$ and satisfying $\Sigma \subseteq \Gamma$; $q_0 \in
Q_\exists$ is the \emph{starting} state; and the \emph{transition
  relation} $\delta$ is of the form
$$
\delta \; \subseteq \; Q \times \Gamma \times Q \times \Gamma
\times \{ L,R \}.
$$
We write $\delta(q,a)$ for
$\{ (q',b,M) \mid (q,a,q',b,M) \in \delta \}$.  As usual, we assume
that $q \in Q_\exists \cup Q_\forall$ implies
$\delta(q,b) \neq \emptyset$ for all $b \in \Gamma$ and
$q \in \{ q_a,q_r \}$ implies $\delta(q,b) = \emptyset$ for all
$b \in \Gamma$. For what follows, we also assume w.l.o.g.\ that for
each $q \in Q_\forall \cup Q_\exists$ and each $b \in \Sigma$, the set
$\delta(q,b)$ has exactly two elements. We assume for notational convenience that
these elements are ordered, i.e., $\delta(q,b)$ is an ordered pair
$((q',b',M'),(a'',b'',M''))$.

A \emph{configuration} of an ATM is a word $wqw'$ with $w,w' \in
\Gamma^*$ and $q \in Q$. The intended meaning is that the tape
contains the word $ww'$ (with only blanks before and behind it), the
machine is in state $q$, and the head is on the leftmost symbol of
$w'$. The \emph{successor configurations} of a configuration $wqw'$
are defined in the usual way in terms of the transition relation
$\delta$. A \emph{halting configuration} is of the form $wqw'$
with $q \in \{ q_a, q_r \}$.

A \emph{computation path} of an ATM $\mathcal{M}$ on a word $w$ is a
(finite or infinite) sequence of configurations $c_1,c_2,\dots$ such
that $c_1=q_0w$ and $c_{i+1}$ is a successor configuration of $c_i$
for $i \geq 0$.  All ATMs considered in this paper have only
\emph{finite} computation paths on any input. A halting configuration
is \emph{accepting} iff it is of the form $wq_aw'$.  For non-halting
configurations $c=w q w'$, the acceptance behaviour depends on $q$: if
$q \in Q_\exists$, then $c$ is accepting iff at least one successor
configuration is accepting; if $q \in Q_\forall$, then $c$ is
accepting iff all successor configurations are accepting.  Finally,
the ATM $\mathcal{M}$ with starting state $q_0$ \emph{accepts} the
input $w$ iff the \emph{initial configuration} $q_0w$ is accepting.
We use $L(\Mmc)$ to denote the language accepted by \Mmc, i.e.,
$L(\Mmc) = \{ w \in \Sigma^* \mid \Mmc \text{ accepts } w \}$.

To obtain a witness for the acceptance of an input by an ATM, it is
common to arrange configurations in a tree.  Such an \emph{acceptance
  tree of an ATM} $\mathcal{M}$ with starting state $q_0$ \emph{on a word
$w$} is a finite tree whose nodes are labelled with configurations such
that
\begin{itemize}

\item the root node is labelled with the initial configuration $q_0w$;

\item if a node $s$ in the tree is labelled with $wqw'$,
  $q \in Q_\forall$, then $s$ has exactly two successors, labeled with
  the two successor configurations of $wqw'$;

\item if a node $s$ in the tree is labelled with $wqw'$, $q \in
  Q_\exists$, then $s$ has exactly two successors, both
  labelled with a successor configuration of $wqw'$;\footnote{A single
    successor would of course be sufficient; we only use two
    successors (which can carry the same label) to enable a more
    uniform reduction.}

\item leaves are labelled with accepting halting configurations.

\end{itemize}
It is clear that there exists an acceptance tree of $\mathcal{M}$ on
$w$ if and only if $\mathcal{M}$ accepts $w$.

According to Theorem 3.4 of \cite{chandraAlternation1981}, there is a double
exponentially space bounded ATM $\mathcal{M}$ whose word problem is
\ThreeExpTime-hard.  We may w.l.o.g.\ assume that the length of every
computation path of $\mathcal{M}$ on any input $w \in \Sigma^n$ is
bounded by $2^{2^{2^n}}$, and all the configurations $wqw'$ in such
computation paths satisfy $|ww'| \leq 2^{2^n}$.

We prove Theorem~\ref{thm:threeexpalco} by reduction from the word
problem for \Mmc. Thus let $w \in \Sigma^*$ be an input to
$\mathcal{M}$.  We have to construct an \ALC-ontology $\Omc$ and an
\ALCO-ontology $\Omc'$ such that $\Omc \cup \Omc'$ is a conservative
extension of $\Omc$ if and only if $M$ does not accept $w$.  This can
be achieved by extending the ontologies $\Omc_n$ and $\Omc'_n$ from
Section~\ref{app:largewitness}, where $n$ is the length of $w$.  The
main idea is to do this such that models \Amf of witness concepts for
$\Omc$ and $\Omc'$ describe an acceptance tree of $\mathcal{M}$ on $w$
instead of a counting tree. In fact, such models \Amf describe a tree
that is a counting tree and at the same time a computation tree. We
again use only a single role name~$r$. Each configuration of \Mmc is
represented by a Counter~3 sequence and each tape cell of a
configuration is represented by a Counter~2 sequence. Thus, each node
of the acceptance tree is spread out over $2^n \cdot 2^{2^n}$ elements
in the model and the role name $r$ might indicate moving to the next
tape cell in the same configuration, moving to the first tape cell of
a successor configuration, and also moving to the next element in the
Counter~2 sequence that represents the current tape cell.

So we only need to represent the computation of \Mmc on $w$ on top of
a counting tree that we have already enforced by $\Omc_n$ and
$\Omc'_n$, and we can exploit the three counters that $\Omc_n$ and
$\Omc'_n$ give us. Such a representation is in fact fairly standard
and no additional technical tricks are needed, see for instance the
\TwoExpTime-hardness proof for conservative extensions in \ALC given
in \cite{DBLP:conf/kr/GhilardiLW06}. Since fully worked out concept inclusions are
nevertheless lengthy and difficult to comprehend, we only sketch the
required extensions of $\Omc_n$ and~$\Omc'_n$.

In the extension \Omc of $\Omc_n$, we use an additional concept name
$S_q$ for every state $q \in Q$, $S_a$ for every symbol $a \in
\Sigma$, and $H$ for indicating the head position. The ontology \Omc
then makes sure that all these symbols are interpreted uniformly in
Counter~2 sequences, that exactly one concept name $S_a$ is true at
every Counter~2 sequence, and that there is exactly one Counter~2
sequence in each Counter~3 sequence where $H$ is true, together with a
unique concept name $S_q$. It also guarantees that the computation
starts with the initial configuration $q_0w$. All remaining properties
of computation trees are achieved by $\Omc'$. In preparation for this,
we add more concept names to \Omc, namely primed versions $S'_q$,
$S'_q$, and $H'$ of all concept names introduced above as well as
concept names $Y_{q,a,M}$ and $Y_{q,a,M}$ for all $q \in Q$, $a \in
\Sigma$, and $M \in \{L,R\}$. Informally, a concept name $Y_{q,a,M}$
indicates that for moving to the current configuration, the Turing
machine has decided to write symbol $a$, switch to state $q$, and move
in direction $M$.  Still in $\Omc$, we make sure that the primed
concept names satisfy the same constraints as their unprimed siblings,
that the concept names $Y_{q,a,M}$ are set in successor configurations
(Counter~3 sequences) in accordance with the unprimed concept names
and the transition relation, and that if some $Y_{q,a,M}$ is set in
the current configuration, then the non-primed and the primed concept
names relate to each other accordingly.

It remains for $\Omc'$ to make sure that the transition relation of
\Mmc is respected and that the the rejecting state is not reached on
any branch. Given what was already done in \Omc, the former can be
achieved by enforcing that whenever some unprimed concept name
$S_q$, $S_a$, or $H$ is true in some Counter~2 sequence, then its
primed version is true in the Counter~2 sequence that represents
the same Counter~2 value of all subsequent Counter~3 sequences.
We refer to this as \emph{correct copying}.

Strictly speaking, the second ontology $\Omc'$ is not an extension of
\Omc because we need to replace Line~(34) with the following, which
admits five different types of defects in place of two:
$$
  A \sqsubseteq P_1 \sqcup P_2 \sqcup P_3 \sqcup P_4 \sqcup P_5.
$$
$P_4$ is for checking that some branch reaches the rejecting state.
This is easy to implement, using existential restrictions as in
Lines~(37)-(46) (rather than universal restrictions as in Lines~(35)
and~(36)). $P_5$ is for verifying that correct copying is taking
place. This is achieved by a slight variation of the concept
inclusions in Figure~\ref{fig:tbox2b} (which also use existential
restrictions). In fact, those concept inclusions can be viewed as
copying the value of an $X_3$-bit to same-value Counter~2
subsequences of subsequent Counter~3 sequences. We copy
the information stored in the concept names $S_q$, $S_a$, and $H$
instead.

\section{Proofs for Section~\ref{sec:dfstrong}}

\thmfostrongsep*
\begin{proof} \
The \ExpTime upper bound follows from the fact that the 
complement of the problem to decide $\varphi_{\Kmc,\Sigma,a}(x) \models \neg \varphi_{\Kmc,\Sigma,b}(x)$ can be equivalently formulated as a concept
satisfiability problem in the extension $\mathcal{ALCIO}^{u}$ of $\mathcal{ALCIO}$ with the universal 
role $u$. To see this, obtain $C_{\Kmc,\Sigma,a}$ 
from $\Kmc$ by taking the conjunction of the following concepts:
\begin{itemize}
	\item $\forall u.(C \rightarrow D)$, for $C \sqsubseteq D\in \Omc$;
	\item $\forall u.(\{c\}\rightarrow \exists R.\{d\})$, for $R(c,d)\in \Dmc$;
	\item $\forall u.(\{c\}\rightarrow A)$, for $A(c)\in \Dmc$;
\end{itemize}
and then replacing 
\begin{itemize}
	\item all concept and role names $X$ not in $\Sigma$ by
fresh and distinct symbols $X_{a}$;
    \item all individual names $c$ not in $\Sigma\cup \{a\}$ by fresh and distinct individual names $c_{a}$;
    \item the individual name $a$ by a fresh individual name $m$; 
    \item if $a\in \Sigma$ then add $\{m\} \leftrightarrow \{a\}$ as a conjunct. 
\end{itemize}
Define $C_{\Kmc,\Sigma,b}$ in the same way with $a$ replaced by $b$.
Then $\varphi_{\Kmc,\Sigma,a}(x) \wedge \varphi_{\Kmc,\Sigma,b}(x)$ is satisfiable if $m \wedge C_{\Kmc,\Sigma,a}\wedge C_{\Kmc,\Sigma,b}$ is satisfiable.
\end{proof}

\thmstrongFOequivalent*

\begin{proof} \
	``2. $\Rightarrow$ 1.'' is trivial. For the converse direction, assume that Condition~1. holds.
	
	Note that we can view $\Kmc_{\Sigma,a}$ as the union of $\Omc_{a}$ and $\Dmc_{a}$, where 
	\begin{itemize}
		\item $\Omc_{a}$ is a copy of $\Omc$ in which all concept and role names  $X\not\in \Sigma$ have been replaced by fresh symbols $X_{a}$;
		\item $\Dmc_{a}$ is a copy of $\Dmc$ in which every concept and role name $X\not\in \Sigma$ is replaced by $X_{a}$ and in which every individual $c\not\in \Sigma\cup \{a\}$ is replaced by a variable $x_{c,a}$ and $a$ is replaced by $x$. Moreover, $x=a$ is added if $a\in \Sigma$.   
	\end{itemize}
Thus, by taking the conjunction of all members of $\Dmc_{a}$ and existentially quantifying over all variables distinct from $x$ we obtain a formula in CQ$^{\ALCI}$.
$\Kmc_{\Sigma,b}$ can be viewed accordingly with $a$ replaced by $b$.

In what follows we write
\begin{itemize}
	\item $\Amf,d \Leftrightarrow_{\text{CQ}^{\ALCIO},\Sigma} \Bmf,e$ if
$\Amf\models \varphi(d)$ iff $\Bmf\models \varphi(e)$, for all
$\varphi(x)$ in $\text{CQ}^{\ALCIO}(\Sigma)$.
    \item $\Amf,d \Leftrightarrow_{\text{CQ}^{\ALCIO},\Sigma}^{\text{mod}} \Bmf,e$
if for all finite $D\subseteq \text{dom}(\Amf)$ 
containing $d$ we have
$\Amf,d\rightarrow_{D,\ALCIO,\Sigma} \Bmf,e$, and vice versa.
\end{itemize}
By Condition~1, we have $\varphi_{\Kmc,\Sigma,a}(x) \models \neg\varphi_{\Kmc,\Sigma,b}(x)$. Assume there does not exist a separating formula in BoCQ$^{\mathcal{ALCIO}}(\Sigma)$. We first show the following claim. 
	
	\medskip

\noindent	
	\emph{Claim 1}. There exist pointed structures $\Amf,d$ and $\Bmf,e$ such that
	$\Amf\models \varphi_{\Kmc,\Sigma,a}(d)$ and
	$\Bmf \models \varphi_{\Kmc,\Sigma,b}(e)$ and $\Amf,d \Leftrightarrow_{\text{CQ}^{\ALCIO},\Sigma} \Bmf,e$.
		
	\medskip
	For the proof of Claim 1, consider the set $\Gamma$ of all formulas $\psi(x)$ in BoCQ$^{\mathcal{ALCIO}}(\Sigma)$ such that $\varphi_{\Kmc,\Sigma,a}(x) \models \psi(x)$. By compactness and our assumption, 
	$$
	\Gamma \cup \{\varphi_{\Kmc,\Sigma,b}(x)\}
	$$ 
	is satisfiable. Take a pointed model $\Bmf,e$ of
	$\Gamma \cup \{\varphi_{\Kmc,\Sigma,b}(x)\}$. Next, let $\Psi$ be the set of all
	$\psi(x)$ in BoCQ$^{\mathcal{ALCIO}}(\Sigma)$ such that $\Bmf \models \psi(e)$.
	By compactness and assumption 
	$$
	\Psi \cup\{\varphi_{\Kmc,\Sigma,a}(x)\}
	$$
	is satisfiable. Take a pointed model $\Amf,d$ of $\Psi \cup\{\varphi_{\Kmc,\Sigma,a}(x)\}$.
	By definition, the pointed models $\Amf,d$ and $\Bmf,e$ are as required in Claim 1.
	
	\medskip
	Using $\omega$-saturated elementary extensions of the pointed models $\Amf,d$ and $\Bmf,e$ from Claim 1 we obtain the following claim using Lemma~\ref{lem:equivalence2}.
	
	\medskip
	
\noindent
	\emph{Claim 2}. There exist pointed structures $\Amf,d$ and $\Bmf,e$ such that
	$\Amf\models \varphi_{\Kmc,\Sigma,a}(d)$ and
	$\Bmf \models\varphi_{\Kmc,\Sigma,b}(e)$ and such that 
	$\Amf,d \Leftrightarrow_{\text{CQ}^{\ALCI},\Sigma}^{\text{mod}} \Bmf,e$.
	
	\medskip
	Now take assignments $v_{a}$ from the variables of $\varphi_{\Kmc,\Sigma,a}$ into 
	$\Amf$ witnessing $\Amf\models \varphi_{\Kmc,\Sigma,a}(d)$ and $v_{b}$ from
	the variables of $\varphi_{\Kmc,\Sigma,b}$ into 
	$\Bmf$ witnessing $\Bmf\models \varphi_{\Kmc,\Sigma,b}(e)$. Let $D_{a}$ and $D_{b}$ be the images of $v_{a}$ in $\Amf$ and of $v_{b}$ in $\Bmf$, respectively. By definition, we have $\Sigma$-homomorphisms 
	\begin{itemize}
		\item $h_{a}: \Amf_{|D_{a}} \rightarrow \Bmf$ mapping $d$ to $e$ and such  that $\Amf,c \sim_{\mathcal{ALCIO},\Sigma} \Bmf,h_{a}(c)$ for all $c\in D$;
		\item $h_{b}: \Bmf_{|D_{b}} \rightarrow \Amf$ mapping $e$ to $d$ and
	such that $\Bmf,c \sim_{\mathcal{ALCIO},\Sigma}\Amf, h_{b}(c)$ for all $c\in D_{b}$. 
	\end{itemize}
     We also have by definition that for any
	$c\in \text{dom}(\Amf)$ there exists a $c'\in \text{dom}(\Bmf)$ such that $\Amf,c \sim_{\mathcal{ALCIO},\Sigma} \Bmf,c'$, and vice versa.
	We now merge $\Amf$ and $\Bmf$ to a single structure by taking their bisimulation product $\Cmf$, defined as follows. 
	The domain of $\Cmf$ is 
	$$
	\{ (c,c') \in \text{dom}(\Amf) \times \text{dom}(\Bmf) \mid \Amf,c \sim_{\ALCI,\Sigma} \Bmf,c'\}
	$$
	and we set 
	\begin{itemize}
		\item $(c,c')\in A^{\Cmf}$ if $c\in A^{\Amf}$ (equivalently, if $c'\in A^{\Bmf}$) for all $A\in \Sigma$;
		\item $(c,c')\in A^{\Cmf}$ if $c\in A^{\Amf}$ for all $A \in \text{sig}(\varphi_{\Kmc,\Sigma,a})\setminus \Sigma$;
		\item 
		$(c,c')\in A^{\Cmf}$ if  $c'\in A^{\Bmf}$ for all $A \in \text{sig}(\varphi_{\Kmc,\Sigma,b})\setminus \Sigma$;
		\item $((c_{1},c_{1}'), (c_{2},c_{2}')) \in r^{\Cmf}$ if $(c_{1},c_{2})\in r^{\Amf}$ and $(c_{1}',c_{2}')\in r^{\Bmf}$ for all $r\in \Sigma$;
	    \item$((c_{1},c_{1}'), (c_{2},c_{2}')) \in r^{\Cmf}$ if $(c_{1},c_{2})\in r^{\Amf}$ for all $r\in \text{sig}(\varphi_{\Kmc,\Sigma,a})\setminus \Sigma$;
        \item$((c_{1},c_{1}'), (c_{2},c_{2}')) \in r^{\Cmf}$ if $(c_{1}',c_{2}')\in r^{\Bmf}$ for all $r\in \text{sig}(\varphi_{\Kmc,\Sigma,b})\setminus\Sigma$;
        \item $c^\Cmf = (c^{\Amf},c^{\Bmf})$ for all $c\in \Sigma$.
	\end{itemize}
We show that $\Cmf\models (\varphi_{\Kmc,\Sigma,a}\wedge \varphi_{\Kmc,\Sigma,b})(d,e)$ which contradicts the assumption that
$\varphi_{\Kmc,\Sigma,a}(x)\models \neg \varphi_{\Kmc,\Sigma,b}(x)$.
To show that $\Cmf$ is a model of $\Omc_{a}$ and $\Omc_{b}$ it suffices to show 
the following claim.

\medskip
\noindent
\emph{Claim 3}. (1) The projection $p_{a}:\Cmf \rightarrow \Amf$ defined by setting $p_{a}(c,c')=c$ is an $\mathcal{ALCIO}(\text{sig}(\varphi_{\Kmc,\Sigma,a}))$-bisimulation between $\Cmf$ and $\Amf$.

(2) The projection $p_{b}:\Cmf \rightarrow \Bmf$ defined by setting $p_{b}(c,c')=c'$ is an $\mathcal{ALCIO}(\text{sig}(\varphi_{\Kmc,\Sigma,b}))$-bisimulation between $\Cmf$ and $\Bmf$.

\medskip
The proof of Claim~3 is straightforward and omitted. It follows from Claim~3 and the assumption that $\Amf$ is a model of $\Omc_{a}$ and $\Bmf$ a model of $\Omc_{b}$ that $\Cmf$ is a model of $\Omc_{a}\cup \Omc_{b}$.

Next we lift the variable assignments $v_{a}$ and $v_{b}$ from $\Amf$ and, respectively, $\Bmf$ to $\Cmf$. Thus, we set 
\begin{itemize}
	\item $\bar v_{a}(x_{c})= (v_{a}(x_{c}),h_{a}(v_{a}(x_{c}))$ for all variables of the form $x_{c}$ in $\varphi_{\Kmc,\Sigma,a}$ and 
	\item $\bar v_{b}(y_{c})= (v_{b}(y_{c}),h_{b}(v_{b}(y_{c}))$ for all variables of the form $y_{c}$ in $\varphi_{\Kmc,\Sigma,b}$.
    \item $\bar v_{a}(x)=\bar v_{b}(x) = (v_{a}(x),v_{b}(x))$.
\end{itemize}
The following claim is straightforward now.

\medskip
\noindent
\emph{Claim 4.} $\Cmf\models_{\bar v_{a}} \Dmc_{a}(d,e)$ and $\Cmf\models_{\bar v_{b}} \Dmc_{b}(d,e)$.

\medskip
Claim 4 implies $\Cmf\models (\varphi_{\Kmc,\Sigma,a}\wedge \varphi_{\Kmc,\Sigma,b})(d,e)$ as we have established already that $\Cmf$ is
a model of $\Omc_{a}\cup \Omc_{b}$.
This concludes the proof.
\end{proof}

\thmtwoexpdl*

We first prove the upper bounds. To this end, we formulate the complexity results proved in~\cite{AJMOW-AAAI21} for interpolant existence in detail. Let $\Lmc\in \DLS$. 
Let $\Omc_{1},\Omc_{2}$ be $\Lmc$-ontologies and let $C_{1},C_{2}$ be $\Lmc$-concepts. Let $\Sigma= \text{sig}(\Omc_{1},C_{1}) \cap \text{sig}(\Omc_{2},C_{2})$. An $\Lmc$-interpolant for the \emph{$\Lmc$-tuple} $\Omc_{1},\Omc_{2},C_{1},C_{2}$ is an $\Lmc(\Sigma)$-concept $C$ such that
\begin{itemize}
	\item $\Omc_{1} \models C_{1} \sqsubseteq C$;
	\item $\Omc_{2} \models C \sqsubseteq C_{2}$.
\end{itemize}
The following is shown in~\cite{AJMOW-AAAI21}.
\begin{theorem}
	Let $\Lmc\in \DLS$. Then the problem to decide whether an $\Lmc$-interpolant
    exists for $\Lmc$-tuples $\Omc_{1},\Omc_{2},C_{1},C_{2}$ is \TwoExpTime-complete.
\end{theorem}
Now assume that $\Lmc\in \DLS$ and a labeled $\Lmc$-KB $(\Kmc,\{a\},\{b\})$ with $\Kmc=(\Omc,\Dmc)$ and $\Sigma\subseteq \text{sig}(\Kmc)$ are given. Let $\Lmc\Omc$ denote the extension of $\Lmc$ by nominals (if $\Lmc$ contains nominals already then set $\Lmc\Omc=\Lmc$).
Obtain an $\Lmc\Omc$-ontology $\Omc_{\Kmc,\Sigma,a}$ from $\Kmc$ by taking the following inclusions:
\begin{itemize}
	\item all inclusions in $\Omc$;
	\item $\{c\}\sqsubseteq \exists R.\{d\}$, for $R(c,d)\in \Dmc$;
	\item $\{c\}\sqsubseteq A$, for $A(c)\in \Dmc$;
\end{itemize}
and then replacing 
\begin{itemize}
	\item all concept and role names $X$ not in $\Sigma$ by
	a fresh symbol $X_{a}$;
	\item all individuals $c$ not in $\Sigma\cup \{a\}$ by fresh and distinct individuals $c_{a}$;
	\item the individual $a$ by a fresh individual $m_{a}$. If 
	$a\in \Sigma$ then the CI $\{m_{a}\} \equiv \{a\}$ is added.
\end{itemize}
$\Omc_{\Kmc,\Sigma,b}$ is obtained from $\Kmc$ in the same way by replacing $a$ by $b$.
Observe that an $\Lmc\Omc(\Sigma)$-concept strongly separates $(\Kmc,\{a\},\{b\})$ iff it is an $\Lmc\Omc$-interpolant for the $\Lmc\Omc$-tuple $\Omc_{\Kmc,\Sigma,a}$, $\Omc_{\Kmc,\Sigma,b}$, $m_{a},\neg m_{b}$. If $\Lmc$ contains nominals, then the upper bounds stated in Theorem~\ref{thm:thmtwoexpdl} follow immediately. If $\Lmc$ does not contain nominals, then we may assume that $\Sigma$ does not contain individual names.
Then $\Omc_{\Kmc,\Sigma,a}$ and $\Omc_{\Kmc,\Sigma,b}$ do not share any
individual names and therefore an $\Lmc(\Sigma)$-concept strongly separates $(\Kmc,\{a\},\{b\})$ iff it is an $\Lmc\Omc$-interpolant for the $\Lmc\Omc$-tuple $\Omc_{\Kmc,\Sigma,a}$, $\Omc_{\Kmc,\Sigma,b}$, $m_{a},\neg m_{b}$. Thus, the upper bound follows again. 

Now we come to the lower bounds. We first give a
model-theoretic characterization of strong
$\Lmc$-separability using $\Lmc$-bisimulations. 	

\begin{lemma}\label{thm:alcistrongcrit}
	Let $\Lmc \in \DLS$. Let $(\Kmc,\{a\},\{b\})$ be a labeled $\Lmc$-KB and $\Sigma\subseteq
	\text{sig}(\Kmc)$ a signature.  Then the following conditions are
	equivalent: 
	\begin{enumerate} 
		
		\item $(\Kmc,\{a\},\{b\})$ is strongly $\Lmc(\Sigma)$-separable; 
		
		\item There are no models $\Amf$ and $\Bmf$ of $\Kmc$ such that
		$\Amf,a^{\Amf}\sim_{\Lmc,\Sigma} \Bmf,b^{\Bmf}$.
		
	\end{enumerate} 
	
\end{lemma}
The proof is straightforward using Lemma~\ref{lem:equivalence}. 
The following result is shown as part of the lower bound proof for interpolant existence in~\cite{AJMOW-AAAI21}:
\begin{theorem}\label{thm:model2exp}
	Let $\Lmc \in \DLS$. For $\Lmc$-ontologies $\Omc$ and database $\Dmc$ of the form $\{R(a,a)\}$ it
	is \TwoExpTime-hard to decide the following: do there exist models $\Amf$ and $\Bmf$ of $(\Omc,\Dmc)$ such that
	$\Amf,a^{\Amf}\sim_{\Lmc,\Sigma} \Bmf,d$ for some $d\not= a^{\Bmf}$.
\end{theorem}	
	\TwoExpTime-hardness of Condition~2 of Lemma~\ref{thm:alcistrongcrit}
	is a direct consequence of Theorem~\ref{thm:model2exp}. For suppose
	that $\Lmc$, $\Kmc=(\Omc,\Dmc)$, $\Sigma$, and $a$ from Theorem~\ref{thm:model2exp} are given. Let $b$ be a fresh individual name and $A_{1},A_{2}$ fresh concept names. Add $A_{1}(a)$ and $A_{2}(b)$ to $\Dmc$ to obtain $\Dmc'$ and add $A_{1} \sqsubseteq \neg A_{2}$ to $\Omc$ to obtain $\Omc'$. 
	Then exist models $\Amf$ and $\Bmf$ of $(\Omc,\Dmc)$ such that
	$\Amf,a^{\Amf}\sim_{\Lmc,\Sigma} \Bmf,d$ for some $d\not= a^{\Bmf}$
	iff there exist models $\Amf$ and $\Bmf$ of $\Kmc$ such that $\Amf,a^{\Amf} \sim_{\Lmc,\Sigma} \Bmf,b^{\Bmf}$.
\section{Proofs for Section~\ref{sec:GF}}

\thmmaingf*

The proof of Theorem~\ref{thm:main-gf} is inspired by the proof of the 
undecidability of conservative extensions and projective conservative
extensions in every extension of the three-variable
fragment GF$^3$ of GF~\cite{JLMSW17}. Unfortunately, it
is not clear how to achieve a direct reduction of conservative
extensions to separability for GF. The relativization that was used 
for languages in $\DLS$ does not work because non-conservativity in 
GF is witnessed by \emph{sentences} while separability is witnessed by \emph{open formulas}.

The proof is by a reduction from the halting problem of two-register
machines. A (deterministic) \emph{two-register machine (2RM)}\index{Two-register
	machine}\index{2RM} is a pair $M=(Q,P)$ with $Q = q_0,\dots,q_{\ell}$
a set of \emph{states} and $P = I_0,\dots,I_{\ell-1}$ a sequence of
\emph{instructions}.  By definition, $q_0$ is the \emph{initial
	state}, and $q_\ell$ the \emph{halting state}. For all $i < \ell$,
\begin{itemize}
	
	\item either $I_i=+(p,q_j)$ is an \emph{incrementation instruction}
	with $p \in \{0,1\}$ a register and $q_j$ the subsequent state;
	
	\item or $I_i=-(p,q_j,q_k)$ is a \emph{decrementation instruction}
	with $p \in \{0,1\}$ a register, $q_j$ the subsequent state if
	register~$p$ contains~0, and $q_k$ the subsequent state otherwise.
	
\end{itemize}
A \emph{configuration} of $M$ is a triple $(q,m,n)$, with $q$ the
current state and $m,n \in \mathbb{N}$ the register contents.  We
write $(q_i,n_1,n_2) \Rightarrow_M (q_j,m_1,m_2)$ if one of the
following holds:
\begin{itemize}
	
	\item $I_i = +(p,q_j)$, $m_p = n_p +1$, and $m_{1-p} =
	n_{1-p}$;
	
	\item $I_i = -(p,q_j,q_k)$, $n_p=m_p=0$, and $m_{1-p} =
	n_{1-p}$;
	
	\item $I_i = -(p,q_k,q_j)$, $n_p > 0$, $m_p = n_p -1$, and
	$m_{1-p} = n_{1-p}$.
	
\end{itemize}
The \emph{computation} of $M$ on input $(n,m) \in \mathbb{N}^2$ is
the unique longest configuration sequence $(p_0,n_0,m_0)
\Rightarrow_M (p_1,n_1,m_1) \Rightarrow_M \cdots$ such that $p_0 =
q_0$, $n_0 = n$, and $m_0 = m$.
The halting problem for 2RMs is to decide, given a 2RM $M$, whether
its computation on input $(0,0)$ is finite (which implies that its
last state is $q_\ell$).

We convert a given 2RM $M$ into a labeled GF$^3$ KB
$(\Kmc,\{a\},\{b\})$, $\Kmc = (\Omc,\Dmc)$ and signature $\Sigma$ such that $M$ halts iff
$(\Kmc,\{a\},\{b\})$ is (non-)projectively GF$(\Sigma)$-separable iff
$(\Kmc,\{a\},\{b\})$ is (non-)projectively $\ALC(\Sigma)$-separable.
Let $M=(Q,P)$ with $Q = q_0,\dots,q_{\ell}$ and
$P = I_0,\dots,I_{\ell-1}$. We assume w.l.o.g.\ that $\ell \geq 1$ and
that if $I_i=-(p,q_j,q_k)$, then $q_j \neq q_k$. In $\Kmc$,
we use the following set of relation symbols:
\begin{itemize}
	
	\item a binary symbol $N$ connecting a configuration to
	its successor configuration;
	
	\item binary symbols $R_{1}$ and $R_{2}$ that represent the
	register contents via the length of paths;
	
	\item unary symbols $q_{0},\ldots,q_{\ell}$ representing the states of $M$;
	
	\item a unary symbol $S$ denoting points where a computation starts.
	
	\item a unary symbol $D$ used to represent that there is 
	some defect;
	
	\item binary symbols $D^{+}_{p},D_{p}^{-},D_{p}^{=}$ used to
	describe defects in incrementing, decrementing, and keeping
	register $p\in \{0,1\}$;
	
	\item ternary symbols
	$H_{1}^{+},H_{2}^{+},H_{1}^{-},H_{2}^{-},H_{1}^{=},H_{2}^{=}$ used
	as guards for existential quantifiers.
	
\end{itemize}
The signature $\Sigma$ consists of the symbols from the first four
points above.

We define the ontology $\Omc$ as the set of several GF$^3$
sentences.\footnote{The formulas that are not syntactically guarded
	can easily be rewritten into such formulas.} The first sentence
initializes the starting configuration:
\[\forall x  (Sx \rightarrow (q_0x \wedge \neg \exists y \, R_0xy
\wedge \neg \exists y \, R_1xy))\]
Second, whenever $M$ is not in the final state, there is a
next configuration with the correctly updated state. For $0 \leq i <
\ell$, we include: 
$$ \begin{array}{l} 
  \forall x (q_ix \rightarrow \exists y\, Nxy) \\[1mm]
  \forall x (q_{i}x \rightarrow \forall y (Nxy \rightarrow q_{j}y))
  \hspace{1.5cm}\text{ if } I_i=+(p,q_j) \\[1mm] 
  \forall x ((q_{i}x \wedge \neg \exists y R_{p}xy) \rightarrow\forall
  y (Nxy \rightarrow q_{j}y))\\[1mm] 
  \hspace{5.6cm}\text{ if } I_i=-(p,q_j,q_k) \\[1mm] 
  \forall x ((q_{i}x \wedge \exists y R_{p}xy)  \rightarrow \forall y
  (Nxy \rightarrow q_{k}y)) \\[1mm] 
  \hspace{5.6cm}\text{ if } I_i=-(p,q_j,q_k) \end{array} 
$$
Moreover, if $M$ is in the final state, there is no successor
configuration:
\[\forall x(q_\ell x\rightarrow \neg \exists y\,Nxy).\]
%
%
The next conjunct expresses that either $M$ does not halt or the
representation of the computation of $M$ contains a defect. It
crucially uses non-$\Sigma$ relation symbols.  It takes the shape of
\[ \forall x \, (Dx \rightarrow \exists y \, (Nxy \wedge \psi xy)) \]
where $\psi xy$ is the following disjunction which ensures that there
is a concrete defect ($D_p^+,D_p^-,D_p^=$) here or some defect ($D$)
in some successor state:
\[ \begin{array}{l}
	D(y) \vee{} \\[2mm]
	\displaystyle \bigvee_{I_i=+(p,q_j)} (q_ix \wedge q_jy \wedge (D^+_pxy
	\vee D^=_{1-p}xy))\vee{}\\[6mm]
	\displaystyle \bigvee_{I_i=-(p,q_j,q_k)} (q_ix \wedge q_ky \wedge
	(D^-_pxy \vee D^=_{1-p}xy))\vee{} \\[6mm]
	\displaystyle \bigvee_{I_i=-(p,q_j,q_k)} (q_ix \wedge q_jy \wedge
	(D^=_pxy \vee D^=_{1-p}xy))
\end{array} \]
Finally, using the ternary symbols we make sure that the defects are
realized, for example, by taking:
\[ \begin{array}{l}
	\forall x \forall y \, \big(D^+_pxy \rightarrow \\[1mm]
	\hspace*{2mm} (\neg \exists z \, R_pyz \vee (\neg \exists z \, R_pxz
	\wedge \exists z \, (R_pyz \wedge \exists x R_pzx))\vee{}\\[2mm]
	\hspace*{3mm} \exists z (H_{1}^{+}xyz \wedge R_{p}xz \wedge \exists x
	(H_{2}^{+}xzy \wedge R_pyx \wedge D^+_{p}zx)))\big).
\end{array} \]
Similar conjuncts implement the desired behaviour of $D^=_p$ and
$D^-_p$; since they are constructed analogously to the last three
lines above (but using guards $H^{-}_j$ and $H^{=}_j$), details are
omitted. 

Finally, we define a database $\Dmc$ by taking 
\[\Dmc = \{S(a),D(a),S(b)\}.\]
Lemmas~\ref{lem:gfundec1} and~\ref{lem:gfundec2} below establish
correctness of the reduction and thus Theorem~\ref{thm:main-gf}. 
\begin{lemma}\label{lem:gfundec1}
	If $M$ halts, then there is an $\ALC(\Sigma)$ concept that
	non-projectively separates $(\Kmc,\{a\},\{b\})$.
	%
	%
	%
	%
	%
	%
\end{lemma}
%
%
%
%
%
%
\begin{proof}\ The idea is that the separating $\ALC(\Sigma)$ concept
	describes the halting computation of $M$, up to
	$\ALC(\Sigma)$-bisimulations. More precisely, assume that $M$ halts.
	We define an $\ALC(\Sigma)$ concept $C$ such that $\Kmc\models\neg
	C(a)$, but $\Kmc\not\models\neg C(b)$.  Intuitively, $C$ represents
	the computation of $M$ on input $(0,0)$, that is: if the computation
	is $(q_0,n_0,m_0),\dots,(q_k,n_k,m_k)$, then there is an $N$-path of
	length $k$ (but not longer) such that any object reachable in $i
	\leq k$ steps from the beginning of the path is labeled with $q_i$,
	has an outgoing $R_0$-path of length $n_i$ and no longer outgoing
	$R_0$-path, and likewise for $R_1$ and $m_i$. In more detail,
	consider the $\Sigma$-structure $\Amf$ with
	\begin{align*}
		\text{dom}(\Amf) = \{0,\ldots,k\} \cup {} & \{ a_{j}^{i} \mid
		0<i\leq k, 0<j<n_{i}\} \cup{} \\ 
		& \{ b_{j}^{i} \mid 0<i\leq k, 0<j< m_{i}\}
	\end{align*}
	in which
	$$ \begin{array}{rcl}
		N^{\Amf} &=& \{ (i,i+1) \mid i<k\} \\[1mm]
		R_{1}^{\Amf} &=& \bigcup_{i\leq k}\{ (i,a_{1}^{i}),
		(a_{1}^{i},a_{2}^{i}),\ldots,(a_{n_{i}-2}^{i},a_{n_{i}-1}^{i})\}
		\\[1mm]
		R_{2}^{\Amf} &=& \bigcup_{i\leq k}\{ (i,b_{1}^{i}),
		(b_{1}^{i},b_{2}^{i}),\ldots,(b_{m_{i}-2}^{i},b_{m_{i}-1}^{i})\}
		\\[1mm]
		S^{\Amf} &=& \{0\} \\[1mm]
		%
		%
		q^{\Amf} &=& \{ i \mid q_{i}=q\} \text{ for any } q\in Q.
	\end{array} $$
	Then let $C$ be the $\ALC(\Sigma)$ concept that describes $\Amf$
	from the point of $0$ up to $\ALC(\Sigma)$-bisimulations. Clearly,
	$\Kmc\cup\{C(b)\}$ is satisfiable. However, $\Kmc\cup\{C(a)\}$ is
	unsatisfiable since the enforced computation does not contain a
	defect and cannot be extended to have one. In particular, there are
	no $N$-paths of length $>k$ in any model of $\Kmc\cup\{C(a)\}$ and
	there are no defects in register updates in any model of
	$\Kmc\cup\{C(a)\}$.
	%
	%
\end{proof}
The following lemma implies that if $M$ does not halt, then
$(\Kmc,\{a\},\{b\})$ is neither projectively $\Lmc(\Sigma)$-separable
nor non-projectively $\Lmc(\Sigma)$-separable for $\Lmc=\text{GF}$ and
in fact for every logic \Lmc between GF and FO.
\begin{lemma} \label{lem:gfundec2}
	If $M$ does not halt, then for every model $\Amf$ of \Kmc, there is
	a model $\Bmf$ of \Kmc such that $(\Amf,b^\Amf)$ is
	$\Gamma$-ismorphic to $(\Bmf,a^\Bmf)$ where $\Gamma$ consists of all
	symbols except $\text{sig}(\Omc)\setminus \Sigma$.
\end{lemma}

\begin{proof}
  \ Let $\Amf$ be a model of \Kmc. We obtain $\Bmf$ from \Amf by
	re-interpreting $a^\Bmf=b^\Amf$ and inductively defining the
	extensions of the symbols from 
	\begin{align*}
	  \text{sig}(\Omc)\setminus\Sigma = \{D,D^+_p,&D^-_p,D^=_p,\\
	  & H_1^+,H_2^+,H_1^-,H_2^-,H_1^=,H_2^=\}.
	\end{align*}
	We start with $D^\Bmf = \{a^\Bmf\}$ and $X^\Bmf=\emptyset$ for all
	other symbols $X$ from $\text{sig}(\Omc)\setminus \Sigma$. Then,
	whenever $d\in D^\Bmf$ we distinguish two cases: 
	\begin{itemize}
		
		\item If there is an $N$-successor $e$ of $d$ such that the
		counters below $d$ and $e$ are not correctly updated with
		respect to the states at $d,e$, set the extensions of the
		symbols in
		$D^+_p,D^-_p,D^=_p,H_1^+,H_2^+,H_1^-,H_2^-,H_1^=,H_2^=$ so as to
		represent the defect and finish the construction of $\Bmf$.
		
		\item Otherwise, choose an $N$-successor $e$ of $d$ and add
		$e$ to $D^\Bmf$.
		
	\end{itemize}
	Note that, since $M$ does not halt, we can always find such an
	$N$-successor as in the second item. 
\end{proof}

%
%

\thmgffo*

The proof of Theorem~\ref{thm:gffotwo} relies on the encoding of GF and FO$^{2}$-KBs into formulas in GF and FO$^{2}$.

For Points~1 and 2 we can use the extensions of GF and FO$^{2}$ with constants as we only require Craig interpolants in FO and the upper bounds for satisfiability of GF and FO$^{2}$ still hold for their extensions with constants. We start with GF. Assume a labeled GF-KB $(\Kmc,\{a\},\{b\})$ with $\Kmc=(\Omc,\Dmc)$ is given. Let $\Sigma$ be a signature.
Obtain $\varphi_{\Kmc,\Sigma,a}'$ from 
$\Kmc$ by
\begin{itemize}
	\item replacing all relation symbols $R\not\in\Sigma$ by fresh relation symbols $R_{a}$;
	\item replacing all individual names $c\not\in\Sigma\cup \{a\}$ by fresh individual names $c_{a}$;
	\item replacing $a$ by a fresh variable $x$ and adding $x=a$ if $a\in \Sigma$;
\end{itemize} 
and taking the conjunction of the resulting set of formulas.
Define $\varphi_{\Kmc,\Sigma,b}'$ in the same way but with $a$ replaced
by $b$. Then an FO-formula $\varphi$ strongly $\Sigma$-separates $(\Kmc,\{a\},\{b\})$ iff $\varphi$ is an FO-interpolant for $\varphi_{\Kmc,\sigma,a}',\neg\varphi_{\Kmc,\Sigma,b}'$, and we have proved Point~1. For Point~2, observe that $\varphi_{\Kmc,\sigma,a}',\varphi_{\Kmc,\Sigma,b}'$ are in FO$^{2}$
if $\Kmc$ is an FO$^{2}$-KB. Thus, the argument applies to FO$^{2}$ as well and we have proved Point~2.
	
For Points~3 and 4 we assume that $\Sigma$ is a relational signature and as we aim to apply results on the complexity of interpolant existence that have been proved for GF and FO$^{2}$ \emph{without} constants we cannot use constants in the construction of the encodings.

Assume a labeled GF-KB $(\Kmc,\{a\},\{b\})$ with $\Kmc=(\Omc,\Dmc)$ is given and $\Sigma$ is a relational signature. Consider
the formulas $\varphi_{\Kmc,\Sigma,a}$ and $\varphi_{\Kmc,\Sigma,b}$ from the main paper (defined in the obvious way for GF). To obtain GF-formulas $\varphi^{\text{GF}}_{\Kmc,\Sigma,a}$ and
$\varphi^{\text{GF}}_{\Kmc,\Sigma,b}$, take
fresh relation symbols $R_{\Dmc,a}$ and $R_{\Dmc,b}$
of arity $n$, where $n$ is the number of individuals
in $\Dmc$. Then add $R_{\Dmc,a}(\vec{y})$ to
$\Kmc_{\Sigma,a}$ when constructing
$\varphi_{\Kmc,\Sigma,a}(x)$, where $\vec{y}$ is an enumeration 
of the variables in $\Kmc_{\Sigma,a}$. 
Do the same to construct $\varphi^{\text{GF}}_{\Kmc,\Sigma,b}(x)$,
using $R_{\Dmc,b}$ instead of $R_{\Dmc,a}$.
The formulas $\varphi^{\text{GF}}_{\Kmc,\Sigma,a}$ and $\neg\varphi^{\text{GF}}_{\Kmc,\Sigma,b}$ are in GF and play the same role
as the formulas $\varphi_{\Kmc,\Sigma,a}$ and $\varphi_{\Kmc,\Sigma,b}$. In particular, for any formula $\varphi$ the following are equivalent:
\begin{enumerate}
	\item $\varphi$ strongly $\Sigma$-separates $(\Kmc,\{a\},\{b\})$;
	\item $\varphi$ is an interpolant for $\varphi^{\text{GF}}_{\Kmc,\Sigma,a}(x),\neg \varphi^{\text{GF}}_{\Kmc,\Sigma,b}(x)$.
\end{enumerate}	
The complexity upper bound now follows from the result that interpolant existence in GF is decidable in 3\ExpTime \cite{jung2020living}. 
For FO$^{2}$ we proceed as follows. We introduce for every individual name $c$
in $\Dmc$ a unary relation symbol $A_{c}$ and encode $\Dmc$ using the sentences
$\exists x A_{c}(x)\wedge \forall x\forall y (A_{c}(x) \wedge A_{c}(y) \rightarrow x=y)$,
for $c\in \text{ind}(\Dmc)$, and 
\begin{itemize}
	\item $\exists x \exists y R(x,y) \wedge A_c(x) \wedge A_{c'}(y)$ for every $R(c,c')\in \Dmc$;
	\item $\exists x A(x) \wedge A_c(x)$ for every $A(c)\in \Dmc$.
\end{itemize}
Let $\varphi^{2}_{\Kmc,\Sigma,a}$ be the conjunction of the sentences above, the sentences in $\Omc$, and the formula $A_{a}(x)$, where we also replace all relation symbols in $\Kmc$ that are not in $\Sigma$ by fresh relation symbols $R_{a}$.
Define $\varphi^{2}_{\Kmc,\Sigma,b}$ in the same way with $a$ replaced by $b$
and where the unary relation symbols $A_{c}'$ used to encode individuals $c$ are disjoint from the unary relation symbols used for this purpose in $\varphi^{2}_{\Kmc,\Sigma,a}$. Then we have that $\varphi^{2}_{\Kmc,\Sigma,a}$ and $\varphi^{2}_{\Kmc,\Sigma,b}$ are in FO$^{2}$ and a formula $\varphi$ strongly $\Sigma$-separates $(\Kmc,\{a\},\{b\})$ iff it is an interpolant for $\varphi_{\Kmc,\Sigma,a}^{2}(x),\neg \varphi_{\Kmc,\Sigma,b}^{2}(x)$.
The complexity upper bound now follows from the result that interpolant 
non-existence in FO$^{2}$ is decidable in {\sc N}\TwoExpTime \cite{jung2020living}.  

The complexity lower bounds can be proved by generalizing in a straightforward way the reductions from interpolant existence to the existence of strongly separating formulas from the languages in $\DLS$ in the previous section to GF and FO$^{2}$ and using the complexity lower bounds for interpolant existence proved in~\cite{jung2020living}.

\cleardoublepage

\end{document}